%% file: main.tex
\documentclass[acmsmall, nonacm, preprint]{acmart}

\usepackage{listings}
\input{listingLangs} % custom listing styles
\usepackage{wrapfig} % for wrapping around figures, supported package.

\usepackage{cleveref}

\usepackage[utf8x]{inputenc}

\usepackage{svg}

\usepackage{longtable}

\acmJournal{PACMPL}
\begin{document}
\title{Sound and Complete Invariant-Based Heap Encodings}

\author{Zafer Esen}
\orcid{0000-0002-1522-6673}
\affiliation{%
  \institution{Uppsala University}
  \city{Uppsala}
  \country{Sweden}
}
\email{zafer.esen@it.uu.se}

\author{Philipp Rümmer}
\orcid{0000-0002-2733-7098}
\affiliation{%
  \institution{University of Regensburg}
  \city{Regensburg}
  \country{Germany}
}
\affiliation{%
  \institution{Uppsala University}
  \city{Uppsala}
  \country{Sweden}
}
\email{philipp.ruemmer@it.uu.se}

\author{Tjark Weber}
\orcid{0000-0001-8967-6987}
\affiliation{%
  \institution{Uppsala University}
  \city{Uppsala}
  \country{Sweden}
}
\email{tjark.weber@it.uu.se}




% Macros
\input{heapTheoryMacros} % heap-theory related macros
\input{macros}           % macros specific to this paper

\input{content/abstract}

\maketitle

\section{Introduction}\label{sec:intro}
  \input{content/introduction}

\newpage
\section{Motivating Example}\label{sec:motivating}
  \input{content/motivating}

\section{Preliminaries}\label{sec:preliminaries}
  \input{content/preliminaries}

% \section{Heap Encodings}\label{sec:encodings}
% Breaks the structure a bit, contains two sections R and RW.
  \input{content/encodings}

\section{Approximations and Extensions}\label{sec:approximations}
  \input{content/extensions}

\section{Evaluation}\label{sec:evaluation}  
  \input{content/evaluation}

\section{Related Work}\label{sec:related}
  \input{content/related-work}

\section{Conclusion}\label{sec:conclusion}
  \input{content/conclusion}

\section{Data-Availability Statement}
All materials, including the prototype tool \tricerare, our crafted benchmarks and other scripts used in the experiments, are publicly available as part of an accompanying artifact at~\cite{artifact-zenodo}.

\begin{acks}
We want to thank the anonymous reviewers for helpful feedback.
The work was supported by the \grantsponsor{VR}{Swedish Research Council}{https://doi.org/10.13039/501100004359}
through the grants~\grantnum{VR}{2021-06327}, \grantnum{VR}{2024-05195} and \grantnum{VR}{2025-06185}, and by the Ethereum Foundation under grant~\grantnum{EF}{FY25-2124}.
The experimental evaluation was enabled by resources in project NAISS 2024/22-46 provided by the National Academic Infrastructure for Supercomputing in Sweden (NAISS) at UPPMAX, funded by the Swedish Research Council through grant agreement no. 2022-06725.
\end{acks}

\bibliographystyle{ACM-Reference-Format}
\bibliography{refs}

\appendix
\newpage
\section{Appendix}\label{sec:appendix}
\input{content/appendix}

\end{document}
\endinput

%% file: listingLangs.tex
\definecolor{listinggray}{gray}{0.9}
\definecolor{lbcolor}{rgb}{0.9,0.9,0.9}
\definecolor{burgundy}{rgb}{0.5, 0.0, 0.13}

\lstdefinestyle{myCPP}{
	tabsize=4,    
	language=[GNU]C++,
	basicstyle=\tiny,
	upquote=true,
	aboveskip={1.5\baselineskip},
	columns=fixed,
	showstringspaces=false,
	extendedchars=false,
	prebreak = \raisebox{0ex}[0ex][0ex]{\ensuremath{\hookleftarrow}},
	frame=single,
	numbers=left,
    numbersep=.8em,
	showtabs=false,
	showspaces=false,
	showstringspaces=false,
	identifierstyle=\ttfamily,
	keywordstyle={\ttfamily\color[rgb]{0,0,1}},
	commentstyle={\itshape\color{green!40!black}},
	stringstyle=\color[rgb]{0.627,0.126,0.941},
	numberstyle=\color[rgb]{0.205, 0.142, 0.73},
	morekeywords={NULL, datatype, match},
    escapeinside={(*}{*)}
}

\lstdefinestyle{myJava}{
	backgroundcolor=\color{lbcolor},
	tabsize=4,    
	language=[GNU]C++,
	basicstyle=\scriptsize,
	upquote=true,
	aboveskip={1.5\baselineskip},
	columns=fixed,
	showstringspaces=false,
	extendedchars=false,
	prebreak = \raisebox{0ex}[0ex][0ex]{\ensuremath{\hookleftarrow}},
	frame=single,
	numbers=left,
	showtabs=false,
	showspaces=false,
	showstringspaces=false,
	identifierstyle=\ttfamily,
	keywordstyle=\color[rgb]{0,0,1},
	commentstyle=\color[rgb]{0.026,0.112,0.095},
	stringstyle=\color[rgb]{0.627,0.126,0.941},
	numberstyle=\color[rgb]{0.205, 0.142, 0.73},
}

\lstdefinelanguage{heapProlog}{
	backgroundcolor=\color{lbcolor},
	tabsize=4,    
	basicstyle=\scriptsize,
	upquote=true,
	aboveskip={1.5\baselineskip},
	columns=fixed,
	showstringspaces=false,
	extendedchars=false,
	escapeinside     = {(*}{*)},
	prebreak = \raisebox{0ex}[0ex][0ex]{\ensuremath{\hookleftarrow}},
	frame=single,
	numbers=left,
	showtabs=false,
	showspaces=false,
	showstringspaces=false,
	identifierstyle=\ttfamily,
	keywordstyle=\color[rgb]{0,0,1},
	commentstyle=\color[rgb]{0.026,0.112,0.095},
	stringstyle=\color[rgb]{0.627,0.126,0.941},
	numberstyle=\color[rgb]{0.205, 0.142, 0.73},
	alsoletter = {:,-},
	morekeywords={:-, emptyHeap, allocate, write, read, valid, true, false},
}

\lstdefinelanguage{SMT-LIB} {
	alsoletter       = {-, [=], >, <},
	morecomment      = [l]{;},
	morekeywords     = {define-fun, declare-const,define-sort, declare-heap,check-sat,ite,
		Bool,
		bvadd, bvsub, bvneg, bvsgt, bvsle,
		assert, and, or, not, let, =, =>,
	    declare-datatype, declare-datatypes, Int, false, forall},
	escapeinside     = {(*}{*)},
	backgroundcolor=\color{lbcolor},
	tabsize=4,    
	basicstyle=\scriptsize,
	upquote=true,
	aboveskip={1.5\baselineskip},
	columns=fixed,
	showstringspaces=false,
	extendedchars=false,
	prebreak = \raisebox{0ex}[0ex][0ex]{\ensuremath{\hookleftarrow}},
	frame=single,
	numbers=left,
	showtabs=false,
	showspaces=false,
	showstringspaces=false,
	identifierstyle=\ttfamily,
	keywordstyle={\ttfamily\color[rgb]{0,0,1}},
	commentstyle=\color{burgundy},
	stringstyle=\color[rgb]{0.627,0.126,0.941},
	numberstyle=\color[rgb]{0.205, 0.142, 0.73},
    float=tb,
    floatplacement=tb,
    abovecaptionskip=-20pt,
    belowskip=-15pt
}
\lstdefinelanguage{SMT-Prolog} {
	alsoletter       = {-, [=], :},
	morecomment      = [l]{;},
	morekeywords     = {define-fun, declare-const, declare-heap,
		Bool,
		bvadd, bvsub, bvneg, bvsgt, bvsle,
		select, store, assert, and, or, not, let,
		declare-datatype, declare-datatypes, Int,
	    :-, emptyHeap, allocate, write, read, valid, true, false, newAddress, newHeap, alloc},
	escapeinside     = {(*}{*)},
	backgroundcolor=\color{lbcolor},
	tabsize=4,    
	basicstyle=\scriptsize,
	upquote=true,
	aboveskip={1.5\baselineskip},
	columns=fixed,
	showstringspaces=false,
	extendedchars=false,
	prebreak = \raisebox{0ex}[0ex][0ex]{\ensuremath{\hookleftarrow}},
	frame=single,
	numbers=left,
	showtabs=false,
	showspaces=false,
	showstringspaces=false,
	identifierstyle=\ttfamily,
	keywordstyle={\ttfamily\color[rgb]{0,0,1}},
	commentstyle={\ttfamily\color{burgundy}},
	stringstyle=\color[rgb]{0.627,0.126,0.941},
	numberstyle=\color[rgb]{0.205, 0.142, 0.73},
}

%% file: heapTheoryMacros.tex
 	\newcommand{\cmt}[1]{#1}
	\newcommand{\fun}[1]{\textrm{$\mathsf{#1}$}}
	\newcommand{\rd}{\fun{read}}
	\newcommand{\wt}{\fun{write}}
	\newcommand{\ia}{\fun{valid}}
	\newcommand{\eh}{\fun{emptyHeap}}
	\newcommand{\ct}{\fun{heapSize}}
	\newcommand{\nh}{\fun{newHeap}}
	\newcommand{\na}{\fun{newAddr}}
	\newcommand{\nbh}{\fun{newBatchHeap}}
	\newcommand{\nar}{\fun{newAddrRange}}
	\newcommand{\ntha}{\fun{nthAddress}}
	\newcommand{\nthar}{\fun{nthInAddressRange}}
	\newcommand{\cts}{\fun{withinAddressRange}}
	\newcommand{\alloc}{\fun{allocate}}
        \newcommand{\alloch}{\fun{allocateHeap}}
        \newcommand{\alloca}{\fun{allocateAddress}}
	\newcommand{\balloc}{\fun{batchAllocate}}
	\newcommand{\defObj}{\mathit{defObj}}
	\newcommand{\allocRes}{AllocationResult}
	\newcommand{\ballocRes}{BatchAllocationResult}
	\newcommand{\nullAddr}{\fun{nullAddress}}
        \newcommand{\fr}{\fun{free}}
	\newcommand{\all}[1]{\forall #1.}
	\newcommand{\fls}{\mathit{false}}
	\newcommand{\tru}{\mathit{true}}
	\newcommand{\sel}{\mathit{select}}
	\newcommand{\str}{\mathit{store}}
	\newcommand{\addr}{\mathit{Addr}}
    \newcommand{\integer}{\mathit{Int}}
	\newcommand{\addrRange}{\mathit{AddressRange}}
	\newcommand{\hp}{\mathit{Heap}}
	\newcommand{\obj}{\mathit{Obj}}
	\newcommand{\Int}{\mathit{Int}}
	\newcommand{\Bool}{\mathit{Bool}}
	\newcommand{\Nat}{\mathit{Nat}}
	\newcommand{\pair}[2]{\langle #1,#2 \rangle}
        \newcommand{\axiomref}[1]{[\ref{#1}]}

%% file: macros.tex
\newcommand{\funs}{F}
\newcommand{\preds}{P}
\newcommand{\invs}{R}
\newcommand{\sorts}{S}
\newcommand{\sortint}{\mathit{Int}}
\newcommand{\vars}{\mathcal{X}}
\newcommand{\interp}{\mathcal{I}}
\newcommand{\sortinterp}{\mathcal{S}}
\newcommand{\fvars}{\mathit{vars}}
\newcommand{\dom}{\mathit{dom}}
\newcommand{\args}{\mathit{args}}
\newcommand{\chcs}{\mathbb{C}}
\newcommand{\sentences}{\mathbf{S}}
\newcommand{\inds}{\iota}
\newcommand{\indschc}{\kappa}
\newcommand{\val}{\mathit{val}_{M,\beta,\interp}}
\newcommand{\ico}{T}
\newcommand{\tinvariant}{\Phi}
\newcommand{\transform}{\mathfrak{T}}
\newcommand{\head}{\mathit{head}}
\newcommand{\body}{\mathit{body}}

\newcommand{\princess}{\textsc{Prin\-cess}}
\newcommand{\eldarica}{\textsc{Eld\-arica}}
\newcommand{\tricera}{\textsc{Tri\-Cera}}
\newcommand{\tricerare}{\textsc{Tri\-CeraHI}}
\newcommand{\jayhorn}{\textsc{Jay\-Horn}}
\newcommand{\seahorn}{\textsc{Sea\-Horn}}
\newcommand{\cpachecker}{\textsc{CPA\-checker}}
\newcommand{\predatorhp}{\textsc{Pred\-atorHP}}
\newcommand{\predator}{\textsc{Pred\-ator}}
\newcommand{\forester}{\textsc{Fores\-ter}}
\newcommand{\viper}{\textsc{Vi\-per}}

\newcommand{\tri}{\textsc{Tri}}
\newcommand{\trire}{\textsc{TriHI}}
\newcommand{\sea}{\textsc{Sea}}
\newcommand{\cpa}{\textsc{CPA}}
\newcommand{\pred}{\textsc{Pred}}

\newcommand{\push}{\texttt{push}}
\newcommand{\pull}{\texttt{pull}}

\newcommand{\inv}{I}
\newcommand{\chchead}{\mathit{Head}}
\newcommand{\chcinvs}{\mathit{Rels}}

\newcommand{\ltrue}{\mathit{\top}} % maybe \top?
\newcommand{\lfalse}{\mathit{\bot}} % maybe \bottom?

%%%%%%%%%%%%%%%%%%%%%%%%%%%%%%%%%%%%%%%%%%%%%%%%%%%%%%%%%%%%%%%%%%%%%%%%%%%%%%%
% Syntax

\newcommand{\lang}{\textsc{UPLang}}
\newcommand{\alang}{an \textsc{UPLang}}

\newcommand{\assignsym}{:=}
\newcommand{\allockw}{\textbf{alloc}}
\newcommand{\readkw}{\textbf{read}}
\newcommand{\writekw}{\textbf{write}}
\newcommand{\skipkw}{\textbf{skip}}
\newcommand{\ifkw}{\textbf{if}}
\newcommand{\thenkw}{\textbf{then}}
\newcommand{\elsekw}{\textbf{else}}
\newcommand{\whilekw}{\textbf{while}}
\newcommand{\dokw}{\textbf{do}}
\newcommand{\assumekw}{\textbf{assume}}
\newcommand{\assertkw}{\textbf{assert}}
\newcommand{\nullkw}{\textbf{null}}
\newcommand{\defobjkw}{\textbf{defObj}}
\newcommand{\validkw}{\textbf{valid}}

\newcommand{\assigncmd}[2]{#1 \ \assignsym\ #2}
\newcommand{\alloccmd}[2]{#1 \ \assignsym\ \allockw(#2)}
\newcommand{\readcmd}[2]{#1 \ \assignsym\ \readkw(#2)}
\newcommand{\writecmd}[2]{\writekw(#1, #2)}
\newcommand{\skipcmd}{\skipkw}
\newcommand{\seqcmd}[2]{#1; #2}
\newcommand{\ifcmd}[3]{\ifkw\ #1\ \thenkw\ \{#2\}\ \elsekw\ \{#3\}}
\newcommand{\ifshortcmd}[2]{\ifkw\ #1\ \{#2\}}
\newcommand{\whilecmd}[2]{\whilekw\ #1\ \dokw\ \{#2\}}
\newcommand{\assumecmd}[1]{\assumekw(#1)}
\newcommand{\assertcmd}[1]{\assertkw(#1)}
\newcommand{\validcmd}[1]{\validkw(#1)}

\newcommand{\cmd}[1]{\text{``$#1$''}}
% or use \ensuremath{\text{\textquotedbl} #1 \text{\textquotedbl}}?

\newcommand{\progvars}{\mathit{vars}}

%%%%%%%%%%%%%%%%%%%%%%%%%%%%%%%%%%%%%%%%%%%%%%%%%%%%%%%%%%%%%%%%%%%%%%%%%%%%%%%
% Semantics

\newcommand{\stacks}{\mathit{Stack}}
\newcommand{\statements}{\mathit{Stmt}}

\newcommand{\tupleget}[2]{#1_{#2}}

\newcommand{\bigstepeval}{\delta_{\interp}}
\newcommand{\bigstepevalp}[1]{\delta_{#1}}
\newcommand{\sem}[1]{\llbracket #1 \rrbracket_s}
\newcommand{\semtwo}[2]{\llbracket #1 \rrbracket_{#2}}
\newcommand{\upd}[2]{s[#1 \mapsto #2]}
\newcommand{\updp}[3]{#1[#2 \mapsto #3]}
\newcommand{\success}{\top}
\newcommand{\error}{\bot}
\newcommand\doubleplus{\mathbin{+\mkern-10mu+}}
\newcommand{\nontermination}{\infty}
\newcommand{\execp}{\mathcal{E}_p}
\newcommand{\execencR}{\mathcal{E}_{\encR(p)}}

%%%%%%%%%%%%%%%%%%%%%%%%%%%%%%%%%%%%%%%%%%%%%%%%%%%%%%%%%%%%%%%%%%%%%%%%%%%%%%%
% Encodings
\newcommand{\encN}{\mathit{Enc}_{n}}
\newcommand{\encNc}{c}

\newcommand{\pstar}{p^*}
\newcommand{\pstarstar}{p^{**}}

\newcommand{\encR}{\mathit{Enc}_{R}}
\newcommand{\encRW}{\mathit{Enc}_{RW}}
\newcommand{\encRWfun}{\mathit{Enc}_{RWfun}}
\newcommand{\encRWmem}{\mathit{Enc}_{RWmem}}
% R-encoding specific
\newcommand{\lastn}{\mathit{last}}

% RW-encoding specific
\newcommand{\lastwtcnt}{\mathit{\mathit{cnt}_{last}}}
\newcommand{\readtmpcnt}{\mathit{t}}

% Common
\newcommand{\inG}{\mathit{in}}
\newcommand{\cnt}{\mathit{cnt}}
\newcommand{\pg}{\mathit{last_{\addr}}}
\newcommand{\allocctr}{\mathit{cnt}_\mathit{alloc}}
\newcommand{\readresult}{\mathit{x}}
\newcommand{\readtmp}{\mathit{t}}
\newcommand{\dethavockw}{\mathit{havoc}}
\newcommand{\dethavoccall}[1]{\dethavockw(#1)}
\newcommand{\havocvar}{\mathit{seed}}

\newcommand{\foutcomes}{\phi_\sigma}
\newcommand{\fstacks}{\phi_{\mathit{stacks}}}
\newcommand{\freads}{\phi_{\mathit{reads}}}
\newcommand{\freadsone}{\phi_{\mathit{reads1}}}
\newcommand{\freadstwo}{\phi_{\mathit{reads2}}}
\newcommand{\fallocs}{\phi_{\mathit{allocs}}}

% ==== Toggle to change highlighting ====
\newif\ifhighlightchanges
\highlightchangesfalse % set to false to hide highlights

% ==== Define change macros ====
\newcommand{\newtext}[1]{%
  \ifhighlightchanges
    \textcolor{blue}{#1}%
  \else
    #1%
  \fi
}

\ifhighlightchanges
  \newenvironment{oldtext}{\color{red}}{\color{black}}
\else
  \excludecomment{oldtext}
\fi

\ifhighlightchanges
  \newenvironment{newtextt}{\color{blue}}{\color{black}}
\else
  \newenvironment{newtextt}{}{}
\fi

\newcommand{\edittext}[2]{%
  \ifhighlightchanges
    \textcolor{red}{#1} \textcolor{blue}{#2}%
  \else
    #2%
  \fi
}

% Examples:
% A new paragraph: \newtext{This is new text.}
% Here is a change: \edittext{old text}{new text}.

%% file: content/abstract.tex
\begin{abstract}
Verification of programs operating on heap-allocated data structures,
for instance lists or trees, poses significant challenges due to the 
potentially unbounded size of such data structures.
We present \emph{time-indexed heap invariants}, a novel invariant-based 
heap encoding leveraging uninterpreted predicates and prophecy variables
to reduce verification of heap-manipulating programs to verification of programs over integers only.
Our encoding of heap is general and agnostic to specific data structures. 
To the best of our knowledge, our approach is the first heap invariant-based 
method that achieves both soundness and completeness. We provide formal proofs 
establishing the correctness of our encodings. Through an experimental evaluation, 
we demonstrate that time-indexed heap invariants significantly extend the
capability of existing verification tools, allowing automatic verification
of programs with heap that were previously out of reach for state-of-the-art tools.
\end{abstract}

%% file: content/introduction.tex
% introduction.tex
%Verification of programs that manipulate mutable, heap-allocated data
%structures is a significant challenge in software verification.
% Add some citations here and expand.
% The potential for unbounded data structures like linked lists and trees
% introduces substantial complexity into the verification process.

The verification of programs operating on mutable, heap-allocated data structures is a long-standing challenge in verification~\cite{DBLP:conf/cade/BohmeM11}. Automatic verification tools implement a plethora of methods to this end, including techniques based on a representation of heap in first-order logic and the theory of arrays~\cite{DBLP:conf/fmcad/KomuravelliBGM15, DBLP:journals/fuin/AngelisFPP17a}, methods based on separation logic and shape analysis~\cite{DBLP:conf/tacas/DistefanoOY06, DBLP:conf/cav/DudkaPV11, DBLP:conf/sas/ChangRN07, DBLP:conf/cav/BerdineCCDOWY07}, or techniques based on refinement types and liquid types~\cite{refinementtypes,liquidtypes, DBLP:conf/csfw/BengtsonBFGM08, DBLP:journals/toplas/KnowlesF10, DBLP:conf/sas/MonniauxG16, DBLP:conf/sas/BjornerMR13}. Despite the amount of research that was invested, it is still easy, however, to find small and simple programs that are beyond the capabilities of the state-of-the-art tools, as we illustrate in \Cref{sec:motivating}.

In this paper, we focus on methods inspired by refinement types, which have received attention in particular in the context of verification using Constrained Horn Clauses (CHCs)~\cite{DBLP:conf/birthday/BjornerGMR15,DBLP:conf/sas/BjornerMR13,DBLP:conf/sas/MonniauxG16,DBLP:conf/lpar/KahsaiKRS17}. Such methods are also called \emph{invariant-based,} since they infer local invariants on the level of variables, array cells, or heap objects, and this way derive that no assertion violation can happen in a program. Refinement types and CHCs fit together well, since the type inference problem can naturally be automated using an encoding as Horn clauses. This line of research has, among others, led to methods that can infer universally quantified invariants for arrays~\cite{DBLP:conf/sas/BjornerMR13,DBLP:conf/sas/MonniauxG16}, universally quantified invariants about objects on the heap~\cite{DBLP:conf/lpar/KahsaiKRS17}, or the shape of heap data structures~\cite{DBLP:conf/cav/WolffGEHRW25}.

More generally, invariant-based methods can be seen as a form of \emph{verification by transformation:} instead of inferring, for instance, a universally quantified formula capturing some property that has to hold for \emph{all heaps} that a program can construct, we infer a quantifier-free property that uniformly has to hold for all \emph{objects} that can occur on the heap. This change of perspective can be described as a transformation that rewrites a program~$p$ to a program~$p'$, in such a way that the correctness of $p'$ implies the correctness of $p$, but the program~$p'$ is easier to verify than $p$. The existing invariant-based verification techniques have in common that they are \emph{sound} (whenever the transformed program~$p'$ is correct, the original program ~$p$ is correct) but not \emph{complete}: it can happen that $p'$ is incorrect even though $p$ is correct. Incompleteness usually occurs because the inferred invariants are local properties about individual heap objects (or variables, or array cells), and therefore cannot express, for instance, global properties about the shape of the heap that might be necessary to infer the correctness of a program.

We present, to the best of our knowledge, the first invariant-based verification approach for programs with heap that is \emph{both sound and complete}. The approach is based on \emph{time-indexed heap invariants}, a novel encoding of heap operations that completely eliminates the heap by deriving a relation between read and write accesses to the heap. The program~$p'$ after transformation operates only on integers and can be verified automatically by off-the-shelf verification tools supporting linear integer arithmetic and uninterpreted predicates~\cite{DBLP:conf/ecoop/WesleyCNTWG24, tricera}.\footnote{It should be noted that the possibility of encoding arbitrary heap data structures using integers is obvious, using the standard G\"odel encodings from computability theory. Such encodings are, however, purely theoretical and not intended for program verification, as complex nonlinear invariants would be needed for verifying even the simplest programs. Our method, in contrast, is competitive with state-of-the-art verification tools.} In our experiments, using the verification tools \seahorn~\cite{seahorn} and \tricera~\cite{tricera} as back-ends and programs from the SV-COMP as benchmarks, we find that our approach is competitive with the best existing verification tools in the SV-COMP c/ReachSafety-Heap category; in addition, it makes it possible to verify automatically even challenging programs that are beyond the capabilities of state-of-the-art verification tools.

We present two different encodings with time-indexed heap invariants, both of which are sound and complete, but differ in the precise way in which invariants are used to represent heap operations. We also offer insights into modifying these base encodings in order to make the encodings more \emph{verifier-friendly,} by introducing and tracking
additional variables to aid verification tools in their computation of invariants. We precisely characterize which of those modifications preserve completeness, and which modifications sacrifice completeness for performance in practice.

\subsection{Contributions}
The main contributions of this paper are:
\begin{enumerate}
    \item \emph{Time-indexed heap invariants}, a novel heap encoding that reduces the correctness of programs operating on mutable, heap-allocated data structures to the correctness of programs operating on integers that are specified using uninterpreted predicates.
    \item Extensions to the base encoding that aim to make the verification scale to more problems, while maintaining soundness and offering control over completeness.
    \item Proofs of soundness and completeness of the base encoding.
    \item A prototypical implementation of the verification approach in the tool \tricerare, extending the open-source software model checker \tricera.
    \item An experimental evaluation of the different time-indexed heap invariant encodings introduced in the paper using benchmarks taken from the SV-COMP, as well as simple but challenging crafted benchmarks. We release all artifacts and scripts to facilitate reproduction.
\end{enumerate}

\subsection{Organization of the Paper}
This paper is organized as follows: Section~\ref{sec:motivating} provides a motivating example to introduce and motivate the encodings. Section~\ref{sec:preliminaries} introduces the syntax and semantics of the language used throughout the paper. \Cref{sec:r-encoding} and \Cref{sec:rw-encoding} define the notion of time-indexed heap invariants, accompanied by detailed proofs of correctness. Section~\ref{sec:approximations}  presents various extensions and approximations of these encodings. Finally, we discuss our experimental evaluation in Section~\ref{sec:evaluation}.

%% file: content/motivating.tex
% motivating.tex
\begin{wrapfigure}{r}{0.45\textwidth}
  \vspace{-3em}
  \hfill
  \begin{minipage}{0.41\textwidth}
\begin{lstlisting}[
      style=myCPP,captionpos=b,
      caption={A C program allocating and iterating over a linked list. The assertion checks that the list can only have inner nodes with the value $2$, and a last node with the value $3$.},
      label={lst:mot-program-1}
    ]
typedef struct Node {
  int data;
  struct Node *next; } Node;

int main(int N) { // N is a program input
  Node *head = (Node*) malloc(sizeof(Node));
  Node *cur = head;

  for (int i = 0; i < N; i++) {
    cur->next = (Node*) malloc(sizeof(Node));
    cur->data = 2;
    cur = cur->next;
  }
  cur->data = 3; cur->next = NULL;

  cur = head;
  while (cur != NULL) {
    if(cur->next != NULL)
      assert(cur->data == 2);
    else
      assert(cur->data == 3);
    cur = cur->next;
  }
  return 0;
}
\end{lstlisting}
  \end{minipage}
  \vspace{-1em}
\end{wrapfigure}

The C program in \Cref{lst:mot-program-1} is inspired by SV-COMP benchmarks; it allocates and iterates over a singly-linked list, setting all nodes' \texttt{data} fields to $2$, then setting the last node's data field to $3$. The last node points to \texttt{NULL}. The resulting shape of the list is illustrated in \Cref{fig:list-mot-program-1}.

The property verified (lines~19, 21) is that inner nodes hold value~$2$ and the last node value~$3$.
Verifying this property is challenging for multiple reasons. Since the size of the list is unbounded, bounded approaches do not work. The verification system must come up with an invariant that not only reasons about the shape of the list, but also about the node values.
This reasoning is made even more complex due to how the updates are spread out over multiple program locations: the first node is allocated at line 6, continuing with further allocations and updates inside the loop starting at line 9 (for $\texttt{N} > 0$), with a final update to the last node at line 14.

We tried to verify this program using state-of-the-art verification tools \cpachecker~\cite{cpachecker1,cpachecker2}, \predatorhp~\cite{predator,predatorhp}, \seahorn~\cite{seahorn} and \tricera~\cite{tricera}, and only \predatorhp\ (the winner of \textit{MemSafety} and \textit{ReachSafety-Heap} categories at SV-COMP 2024~\cite{svcomp24})
could show that it is safe. When the program updates become slightly more complicated, making interval analysis insufficient (e.g., writing different values depending on a conditional), \predatorhp\ also fails to verify the program. Our approach can verify both the initial program and the more challenging variant using  off-the-shelf
verification tools that support uninterpreted predicates in the input language, such as \seahorn{}~\cite{seahorn,DBLP:conf/ecoop/WesleyCNTWG24} and \tricera{}~\cite{tricera}.

\subsection{The Language of the Encodings}\label{subsec:normalisation}

\begin{wrapfigure}{r}{0.45\textwidth}
    \vspace*{-5ex}
    \hspace*{\fill}
  \begin{minipage}{0.42\textwidth}
    \includegraphics[width=\linewidth]{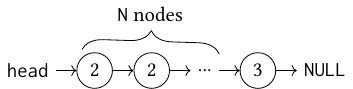}
    \caption{The final shape of the linked list created in Listing~\ref{lst:mot-program-1} for positive \texttt{N}. For non-positive \texttt{N}, the shape will contain only the last node with the value $3$.}
    \label{fig:list-mot-program-1}
  \end{minipage}
\end{wrapfigure}

For greater clarity and to facilitate simpler proofs, we assume a simpler C-like language where all heap operations take one of the following forms: \texttt{*p = v} (a write), \texttt{v = *p} (a read) or \texttt{p = alloc(v)} (an allocation). Furthermore, we introduce the value \texttt{defObj}, available as a program variable with the same name, in order to represent undefined objects that are returned by \emph{invalid} reads. (An invalid read occurs when a program accesses memory that has either not been allocated or has been allocated but not yet initialised, for example, memory returned by a \texttt{malloc} operation in C.). The language also supports \texttt{assert} and \texttt{assume} statements with their usual semantics~\cite{flanagan01popl}. 
% Later we introduce and use the more formal \uplang\ language.

\Cref{lst:common-code} is the result of normalizing the program in \Cref{lst:mot-program-1} into this language. In \Cref{lst:common-code} we introduce the operations \texttt{read} and \texttt{write}, which will be provided definitions by the encodings we introduce. Intuitively, the semantics of \texttt{v = read(p)} corresponds to \texttt{v = *p} and the semantics of \texttt{write(p, v)} to \texttt{*p = v}. The \texttt{malloc} operation is also replaced with a deterministic \texttt{alloc} function that always succeeds and returns a fresh address, and the calls to \texttt{malloc} are replaced with calls to \texttt{alloc} using the introduced \texttt{defObj}.

\begin{figure}[tbp]
%\vspace*{-3ex}
\centering
% First column
\begin{minipage}{0.43\textwidth}
\begin{lstlisting}[style=myCPP, captionpos=b, caption={The program from \Cref{lst:mot-program-1} that is rewritten to use the \texttt{read} and \texttt{write} functions from the encodings. The \texttt{malloc} function is also replaced with a deterministic \texttt{alloc} function that writes a default invalid object (\texttt{defObj}) to the newly-allocated address.}, label={lst:common-code}]
typedef struct Node { 
  int data;
  struct Node *next; 
} Node;

// Prototypes for the read and write functions
Node read(void *p);
void write(void *p, Node v);

void writeNode(Node *p, int data, Node* next) {
  Node n;
  n.data = data;
  n.next = next;
  write(p, n);
}

int cnt_alloc = 0;
void* alloc(Node v) {
  void* p = (void*) ++cnt_alloc;
  write(p, v);
  return p;
}

// represents an invalid Node value
Node defObj;

int main(int N) { // N is the program input
  Node *head = alloc(defObj);
  Node *cur = head;

  for (int i = 0; i < N; i++) {
    writeNode(cur, 2, alloc(defObj));
    cur = read(cur).next;
  }
  writeNode(cur, 3, NULL);

  cur = head;
  while (cur != NULL) {
    Node n = read(cur);
    if(n.next != NULL) {
      assert(n.data == 2);
    } else {
      assert(n.data == 3);
    }
    cur = n.next;
  }
  return 0;
}
\end{lstlisting}
\end{minipage}
\quad\quad
% Second column
\begin{minipage}{0.455\textwidth}
\begin{lstlisting}[style=myCPP, captionpos=b, caption={Trace-based encoding. For the sake of presentation, we represent the trace using a functional data-type with structural pattern matching to destructure the trace, instead of a more verbose C equivalent.}, label={lst:trace-encoding}]
datatype Trace = Empty|Append(Trace,Node*,Node)
Trace H = Empty;

Node read(void *p) {
  Trace curTrace = H;
  while (curTrace != Empty) curTrace match {
    case Append(_, ptr, node) if ptr == p =>
      return node;
    case Append(nextTrace, _, _) =>
      curTrace = nextTrace;
  }
}

void write(void *p, Node v) {
  H = Append(H, p, v);
}
\end{lstlisting}
    %
    % \vspace{5pt}
    %
\begin{lstlisting}[style=myCPP, captionpos=b, caption={Time-indexed heap invariant encoding.}, label={lst:full-encoding}]
unsigned int cnt = 0; // heap update counter
Node last = defObj;   // last object at last_addr
void* last_addr = *;  // the last written addr
int in = N;           // program input
// we assume N is assigned to in at entry to main

// uninterpreted predicate R
R(int in, int cnt_r, Node n);

Node read(void *p) {
  Node result;
  ++cnt;
  if (last_addr == p) {
    assert(R(in, cnt, last));
    result = last;
  } else {
    result = *;  // assign nondet value to result
    assume(R(in, cnt, result));
  }
  return result;
}

void write(void *p, Node v) {
  if (last_addr == p && 0 < p <= cnt_alloc)
    last = v;
}
\end{lstlisting}
\end{minipage}

\Description{A collection of three listings. The listings are explained in their own captions.}
\caption{A sound and complete encoding for heap operations using traces (\Cref{lst:trace-encoding}) and  uninterpreted predicates (\Cref{lst:full-encoding}). The program in \Cref{lst:mot-program-1} is rewritten to be compatible with either encoding in \Cref{lst:common-code}.}\label{fig:full-encoding}
\end{figure}

\subsection{Trace-Based View of Heap Operations}\label{subsec:trace-view}
\begin{figure}
    \centering
    \includegraphics[width=1\linewidth]{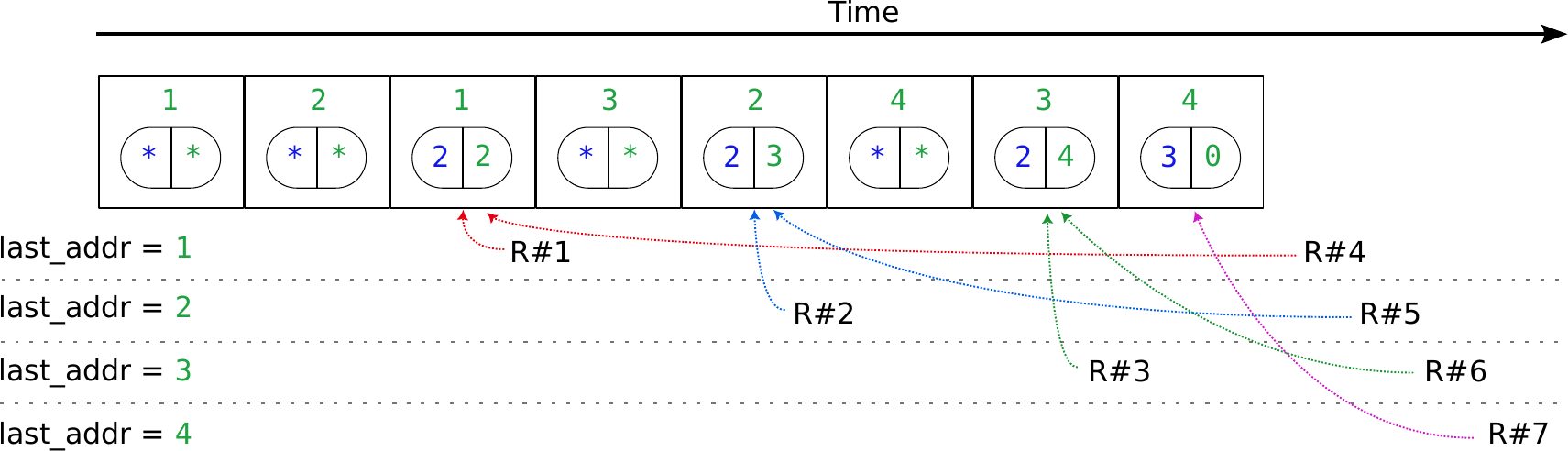}
    \Description{Diagram illustrating heap traces and relational heap encoding for the motivating program.}
    \caption{
    For the program in \Cref{lst:mot-program-1}, the upper part of the diagram illustrates the heap trace for \texttt{N = 3}. In the trace each cell corresponds to a write to the written (address -- \texttt{Node}) pair. The \texttt{Node}s are represented by the rounded rectangles containing the values for the \texttt{data} and \texttt{next} fields respectively, with the number above a node showing that node's address.
    %For instance, the first cell is created at line 26, then the second cell is added by the \texttt{mallocNode} at line 30, and the third cell by the \texttt{writeNode} also at line 30. The loop at line 29 adds more cells, and the last cell corresponds to the final write at line 33.
    %
    The lower part of the diagram shows the relationship between traces and time-indexed heap invariants. In the latter encoding, \texttt{last\_addr} implicitly quantifies over all possible addresses. When \texttt{last\_addr} matches the address being read, the last \texttt{Node} value residing at that address is asserted to be an element of the uninterpreted predicate \texttt{R}, represented by the arrows in the diagram. Due to the universal quantification, \texttt{R} is therefore required to be the union of these values. The reads can be distinguished due to the unique counter value (in the diagram the numbers appended to \texttt{R\#}). The reads \texttt{R\#1} -- \texttt{R\#3} happen at line 33 in \Cref{lst:common-code}, and the rest of the reads happen at line 39.}\label{fig:relational-encoding-explanation}
\end{figure}

Before presenting time-indexed heap invariants, we first provide the intuition through an idealized trace-based view of heap operations.
%using a C-like language extended with \texttt{assert} and \texttt{assume} statements with their usual semantics~\cite{flanagan01popl}.
%
In this view the heap is treated as a chronological trace of write operations. Whenever an object $o$ is written to some address $p$, the trace is extended with the tuple $\langle p,o \rangle$. A read operation from address $p$ returns the most recent tuple containing~$p$ from the trace. \Cref{lst:trace-encoding} provides an implementation for this encoding by defining \texttt{read} and \texttt{write} this way, and using a variable \texttt{H}, representing the heap trace. (\Cref{lst:common-code} and \Cref{lst:trace-encoding} can be combined to obtain the complete trace-based encoding of the program in \Cref{lst:mot-program-1}.)

The trace-based encoding of the heap is sound, because every write operation is registered in the trace. It is also complete, because the content of the heap at every address is known at all times.

The upper part of \Cref{fig:relational-encoding-explanation} illustrates how the heap trace of the program in \Cref{lst:mot-program-1} would look like for \texttt{N = 3}. For instance, the write corresponding to the first \texttt{malloc} operation at line 6 (\Cref{lst:mot-program-1}) is recorded in the first cell at address $1$ and an arbitrary object $*$. The second write arises from the \texttt{malloc} operation at line 10, adding a new cell at address $2$.

% The only value read from address $1$ in the trace is the value recorded in the third cell that was written by the \texttt{writeNode} operation at line 30.

The trace encoding method introduces several challenges for safety proofs.
First, because the trace records every write, it often includes more information than needed to prove a property. Second, the read operation introduces a loop. The proof of safety will often require complicated invariants over the \texttt{Trace} datatype in the invariant of this loop, and of the loops of the original program. These collectively complicate the proof and its automation.

\subsection{Time-Indexed Heap Invariants}\label{subsec:encoding-clike}
Time-indexed heap invariants overcome the limitations of the trace-based approach by requiring invariants only over integers, which simplifies the verification process. In the trace-based encoding, the heap trace contains \emph{all} heap operations, even those that become irrelevant due to subsequent overwrites, that the read operation iterates over. With time-indexed heap invariants, in contrast, the introduced predicate only records the values when a read happens (the arrows in \Cref{fig:relational-encoding-explanation}). Furthermore, time-indexed heap invariants use only basic types and do not use a loop for reads. The program in \Cref{lst:common-code} can be verified using the encoding given in \Cref{lst:full-encoding}. A variant of this program where the written value is not constant (the value $2$ at line 32 of \Cref{lst:common-code})
but is conditionally assigned, is currently beyond the reach of automatic verification tools including \cpachecker\ and \predatorhp\ (the winner of SV-COMP 2024's \textit{MemSafety} category and \textit{ReachSafety-Heap} sub-category), but is verified within seconds on modern hardware using time-indexed heap invariants.

\subsubsection{The Encoding}

To eliminate the explicit heap trace, we employ two main techniques:
\emph{uninterpreted predicates}~\cite{tricera, DBLP:conf/ecoop/WesleyCNTWG24} and \emph{auxiliary variables} including \emph{history variables}~\cite{DBLP:phd/us/Owicki75} (recording past accesses) and \emph{prophecy variables}~\cite{DBLP:journals/tcs/AbadiL91} (anticipating future accesses). Uninterpreted predicates are a straightforward extension to verification tools based on CHCs, and are already supported by \seahorn~\cite{seahorn,DBLP:conf/ecoop/WesleyCNTWG24} and \tricera~\cite{tricera}, CHC-based verification tools for C programs. In particular, the time-indexed heap invariants encoding introduces:
\begin{itemize}
    \item \texttt{cnt}: a history variable incremented at each \texttt{read},
    \item \texttt{last\_addr}: a prophecy variable nondeterministically assigned to force the proof to consider all addresses,
    \item \texttt{last}: a history variable that tracks the object at \texttt{last\_addr} during reads,
    \item \texttt{in}: a history variable that holds the program input value.
\end{itemize}

Uninterpreted predicates are declared similarly to function prototypes (e.g., $R$ at line~8 of \Cref{lst:full-encoding}) and used solely in \texttt{assert} and \texttt{assume} statements. The semantics of an uninterpreted predicate~$P$ corresponds to that of first-order logic: $P$ represents an unknown set-theoretical relation over the argument types of $P$. The meaning of an uninterpreted predicate is fixed throughout the program execution, i.e., a program can test, using \texttt{assert} and \texttt{assume}, whether the relation represented by~$P$ holds for given arguments, but it cannot modify the relation. In order to tell whether a program execution succeeds in the presence of uninterpreted predicates, we therefore first have to define the relation~$\interp(P)$ represented by every predicate~$P$ occurring in the program by assuming some interpretation function~$\interp$. We say that a program involving uninterpreted predicates is \emph{safe} if and only if there is an interpretation of the uninterpreted predicates for which no assertion can fail.

An \texttt{assert} statement involving an uninterpreted predicate~$P$ asserts that the relation~$\interp(P)$ includes the tuple of values given as arguments of $P$; otherwise, the execution of the \texttt{assert} statement fails. In our encoding, the \texttt{assert} statement in line~14 of \Cref{lst:full-encoding} enforces that $R$ holds for the tuple consisting of program input \texttt{in}, the unique read identifier \texttt{cnt}, and the object stored in variable \texttt{last}; the latter is the object written at the last write to \texttt{last\_addr}. 
%(We formally establish this property in \Cref{sec:r-encoding}.)

Conversely, an \texttt{assume} statement blocks program execution (but does not fail) if an uninterpreted predicate does not hold for the given arguments. In the encoding, the \texttt{assume} in line~18 is used to query whether some value is present in $R$. For this, the variable \texttt{result} is first set to a non-deterministic value. The \texttt{assume} then constrains the program to executions consistent with the values represented by $R$.
%In a deterministic program, the \texttt{cnt} variable at each read will hold the same value at each read, and the predicate $R$ is guaranteed to be defined for exactly one result.

% \paragraph{Auxiliary variables} Auxiliary (or ghost) variables are those added to the program only for the purpose of verification.
% The first variant, which is commonly used in verification tasks, is called a \emph{history variable} and it records past actions~\cite{DBLP:phd/us/Owicki75}. The second variant, the mirror image of a history variable, is called a prophecy variable~\cite{DBLP:journals/tcs/AbadiL91}. Instead of recording past behaviours, a prophecy variable \emph{guesses} future behaviours.

The intention behind this encoding is to represent all values ever read from the heap using the predicate~$R$. Consider \Cref{fig:relational-encoding-explanation}, which illustrates executions of the program in \Cref{lst:mot-program-1} for a fixed input (\texttt{N = 3}). Because \texttt{cnt} increments with each read, each read operation is uniquely identifiable (shown as \texttt{R\#1} -- \texttt{R\#7} in the figure).  The prophecy variable \texttt{last\_addr} is set to a non-deterministic value in the beginning of the program execution (line~3, \Cref{lst:full-encoding}) and remains unchanged throughout the execution; this forces the verification to consider all possible addresses. Each chosen address corresponds to a sequence of write-read pairs, illustrated as separate rows in the figure for the address values $1$--$4$.
Assertions involving the $R$ predicate are made precisely when \texttt{last\_addr} matches the address~\texttt{p} currently being read, and therefore precisely when the \texttt{last} variable contains the object that is supposed to be read. Thus, the assertions capture exactly the relationship between the program input \texttt{in}, each unique read identifier \texttt{cnt}, and the object last written to the address~\texttt{p}. The \texttt{assume} statement at line~18 of \Cref{lst:full-encoding} is reached when the read concerns an address that does \emph{not} match \texttt{last\_addr}, in which case the program will instead ``read'' the value provided by the $R$ relation. Consequently, each read operation has as one possible result the object that was last written to address~\texttt{p}.

\subsubsection{Verifying the Example Program.}
Verification tools such as \seahorn~\cite{seahorn,DBLP:conf/ecoop/WesleyCNTWG24} and \tricera~\cite{tricera} verify programs by inferring first-order logic formulas that can be substituted for the uninterpreted predicates in the programs. Those formulas are inferred in conjunction with other annotations, e.g., loop invariants, as annotations can depend on each other: choosing a particular interpretation of the uninterpreted predicates might necessitate the right loop invariants, and vice versa.

Verification of the example succeeds with the following formula chosen for the time-indexed invariant~$R$, which can be computed automatically by \seahorn, together with the corresponding loop invariants (\tricera\ can automatically verify a  modified version of the example, see \Cref{sec:rw-encoding}):
\begin{align}
    R(\mathtt{in}, \mathtt{c}, \mathtt{n})
    ~\equiv~~ &
    (\mathtt{in} < 0 \wedge \mathtt{c} = 1 \wedge \mathtt{n}.\mathtt{next} = \mathtt{NULL} \wedge \mathtt{n}.\mathtt{data} = 3) \vee\mbox{} \notag
    \\
    &
    (\mathtt{in} \geq 0 \wedge 1 \leq \mathtt{c} \wedge \mathtt{c} \leq \mathtt{in} \wedge \mathtt{n}.\mathtt{next} = \mathtt{c} + 1) \vee\mbox{} \label{eq:Rsolution}
    \\
    &
      (\mathtt{in} \geq 0 \wedge \mathtt{in} < \mathtt{c} \wedge \mathtt{c} \leq 2\mathtt{in} \wedge \mathtt{n}.\mathtt{next} = \mathtt{c} - \mathtt{in} + 1 \wedge \mathtt{n}.\mathtt{data} = 2) \vee\mbox{} \notag
    \\
    &
      (\mathtt{in} \geq 0 \wedge \mathtt{c} = 2\mathtt{in} + 1 \wedge \mathtt{n}.\mathtt{next} = \mathtt{NULL} \wedge \mathtt{n}.\mathtt{data} = 3) \notag
\end{align}
The formula considers four different cases. The first disjunct considers the case where the program input~$\mathtt{N}$ (stored as argument $\mathtt{in}$) is negative; in this case, a list of length~1 is constructed and the program calls function~\texttt{read} only once in line~39 of \Cref{lst:common-code}. This one call to \texttt{read} returns a \texttt{Node n} with \texttt{next} pointer \texttt{NULL} and \texttt{data} field $3$.

The second case covers the situation that $\mathtt{in} \geq 0$, in which case the function \texttt{read} will be called repeatedly in line~33. Those calls happen with the read count~$\lstinline!cnt_r! \in [1, \mathtt{in}]$ and return \texttt{Node}s \texttt{n} with \texttt{next} pointer pointing to the respective next object~$\lstinline!cnt_r! + 1$ allocated in line~32. The value of the \texttt{data} field is irrelevant in those reads, since \texttt{data} is overwritten in line~32 or 35; the \texttt{data} field is therefore not mentioned in the second case.

The third case covers the first \texttt{N} calls to \texttt{read} in line~39, which return the list nodes with \texttt{data} field~$2$ and \texttt{next} pointer pointing to the respective next node in the list.

Finally, the fourth case specifies the last call to \texttt{read} in line~39, which returns the last list node with \texttt{data} field~$3$ and \texttt{next} pointer \texttt{NULL} and leads to termination of the second loop.

\begin{wrapfigure}{r}{0.36\textwidth}
  \hfill
  \begin{minipage}{0.9\linewidth}
  \vspace{-1.5em}
\begin{lstlisting}[style=myCPP,label={lst:space-invariants},captionpos=b,caption={Basic encoding following the space invariants approach.}]
// uninterpreted predicate I
I(void *p, Node n);

Node read(void *p) {
  Node result = *; // nondet value
  assume(I(p, result));
  return result;
}

void write(void *p, Node v) {
  assert(I(p, v));
}
\end{lstlisting}
  \end{minipage}
  \vspace{-1em}
\end{wrapfigure}
The invariant~\eqref{eq:Rsolution} captures the exact sequence of \texttt{read} operations happening in the program, and this way makes it possible to distinguish the list nodes with value~$2$ and the last node with value~$3$. The solution is a linear arithmetic formula and can effectively be found thanks to the inference techniques implemented in verification systems and CHC solvers; through the encoding, heap accesses are eliminated entirely, and verification is turned into a purely arithmetic problem. As we show formally in this paper, our encoding is \emph{complete} and applicable for arbitrarily complex computations, regardless of the shape of the heap data structures constructed by the program.

\subsubsection{Comparison with Incomplete Invariant-Based Methods.}

Just like with loop invariants, the complexity of the required time-indexed heap invariants depends on the program and the assertions to be verified. Solution~\eqref{eq:Rsolution} is relatively complex, as the program forces the verification tool to distinguish between list nodes at different indexes. In many practical cases, also simpler formulas suffice. For instance, if lines~40--44 of \Cref{lst:common-code} are replaced by the simple assertion \verb!assert(n.data >= 0)!, then for successful verification it is enough to find the formula $R(\mathtt{in}, \mathtt{c}, \mathtt{n}) \equiv (\mathtt{n}.\mathtt{data} \geq 0)$ that no longer distinguishes between different read indexes.

The simpler program with \verb!assert(n.data >= 0)! can also be
verified using incomplete invariant-based methods. We compare with the
space invariants encoding~\cite{DBLP:conf/lpar/KahsaiKRS17} from
\jayhorn~\cite{DBLP:conf/cav/KahsaiRSS16}, to which our method is most
closely related. The most basic version of the space invariants
encoding can be rephrased in terms of the implementations of
\texttt{read} and \texttt{write} shown in
\Cref{lst:space-invariants}. The uninterpreted predicate~$I$
represents an invariant to be satisfied by all data put on the heap;
when writing to the heap (a ``push'' in \jayhorn{} terminology), it is
asserted that $I$ holds for the address and object written to, while
reading (``pulling'') from the heap can assume $I$. The version of
\Cref{lst:common-code} with lines~40--44 replaced by
\verb!assert(n.data >= 0)! can be verified by
setting~$I(\mathtt{p}, \mathtt{n}) \equiv (\mathtt{n}.\mathtt{data}
\geq 0)$, i.e., using the same formula as with time-indexed heap invariants. More generally, any space invariant~$I$ can be translated to
a time-indexed heap invariant~$R$ by conjoining with $I$ a formula
relating the program input~$\mathtt{in}$, the read
counter~$\mathtt{c}$, and the accessed address~$\mathtt{p}$; programs
that can be verified using space invariants are therefore particularly
simple to verify using our approach. Vice versa, not every time-indexed heap invariant can be translated to a space invariant. While our encoding
with explicit read counters is able to represent the exact sequence of
data values stored at every address, space invariants have to
summarize the possible values of objects at each address using a
single formula; space invariants are therefore not able to describe
global properties involving multiple objects on the heap either.

%More generally, whenever the space invariant encoding
%succeeds in verifying a program with the invariant~$I$, then our
%relational encoding is able to verify the same program with a
%relational
%invariant~$R(\mathtt{in}, \mathtt{c}, \mathtt{x}) \equiv \exists
%\mathtt{p}.\, (I(\mathtt{p}, \mathtt{x}) \wedge
%\mathit{ptr}(\mathtt{in}, \mathtt{c}, \mathtt{p}))$, where
%$\mathit{ptr}$ is a formula relating the program input~$\mathtt{in}$
%and read counter~$\mathtt{c}$ with the accessed address~$\mathtt{p}$.

To increase the precision of space invariants, \jayhorn{} uses various
refinements of the basic encoding~\cite{DBLP:conf/lpar/KahsaiKRS17},
some of which we discuss in \Cref{sec:approximations}. However, the
encodings are all incomplete, which means that correct programs exist
that cannot be verified using any choice of invariant~$I$; for
instance, in experiments with \jayhorn{} we found that no
configuration was able to verify a Java version of the (unmodified)
program in \Cref{lst:common-code}.

%\begin{figure}[tbp]
%
%\caption{Most basic encoding following the space invariants approach~\cite{DBLP:conf/lpar/KahsaiKRS17}.}\label{fig:space-invariants}
%\end{figure}

%It should also be noted that our encoding is agnostic of the shape of the data-structure constructed by the program ...

\subsubsection{Properties of the Encoding}
A program with uninterpreted predicates is safe if and only if an interpretation of the predicates exists such that all assertions hold.
The time-indexed heap invariants encoding of a program is \emph{equi-safe} to the original program, that is, it is both sound and complete. The encoding explicitly records each read operation in the uninterpreted predicate $R$ across the executions corresponding to the different values assigned to \texttt{last\_addr}. Since no read goes unrecorded, soundness follows. The encoding is also complete, because in the strongest interpretation of the predicate~$R$ for which none of the \texttt{assert}s in line~14 of \Cref{lst:full-encoding} can fail, the predicate~$R$ \emph{exactly} represents the values that would have been read in the original program; this implies that no spurious errors can occur. We give a formal proof of this result in \Cref{sec:r-encoding-proof}. It is important to note, however, that the completeness result only holds when time-indexed heap invariants encoding is applied to \emph{deterministic} programs, i.e., programs whose execution is uniquely determined by the value of the input variable~\texttt{in}. The relation~$R$ could otherwise mix up the values read during different unrelated executions of the program. As we discuss in \Cref{sec:semantics}, this is not a restriction, since nondeterminism can always explicitly be modeled as part of the program input.

%% file: content/preliminaries.tex
% preliminaries.tex
\subsection{Overview of \lang\ (Uninterpreted Predicate Language)}\label{sec:lang-overview}
In order to provide a simple language to present our encodings and their proofs, we introduce \lang, an imperative and deterministic language that integrates standard heap operations with \emph{uninterpreted predicates}. In \lang, $\assertkw$ and $\assumekw$ statements may be over concrete formulas or over (applications of) uninterpreted predicates. The language's semantics is built upon the theory of heaps~\cite{esenTheoryHeapConstrained2021,DBLP:conf/lopstr/EsenR20} (see \Cref{tbl:heap-ops}), as this theory provides sorts and operations (such as $\rd$, $\wt$ and $\alloc$) similar to ours, along with well-defined formal semantics.

An \emph{uninterpreted predicate} $P(\bar{x})$ is a predicate whose interpretation is not fixed by the language semantics.
An \emph{interpretation} $\interp$ assigns to each uninterpreted predicate symbol $P$ of arity $n$ a subset of $\sortinterp(\tau_1) \times \dots \times \sortinterp(\tau_n)$, where each $\tau_i$ corresponds to the type of the $i$-th argument of $P$, and $\sortinterp$ is the sort interpretation function defined in \Cref{sec:notation}. Interpretations form a complete lattice ordered by the relation $\sqsubseteq$, defined pointwise: given interpretations $\interp_1, \interp_2$, we have $\interp_1 \sqsubseteq \interp_2$ iff for every predicate $P$, $\interp_1(P) \subseteq \interp_2(P)$.
%The domain $D$ of values is the union of all sort interpretations: $D = \sortinterp(\sortint) \cup \sortinterp(\addr) = \mathbb{Z}$. 

\subsection{Basic Notation and Definitions}\label{sec:notation}
In the rest of this paper, a \emph{program} refers to a program written in \lang, whose syntax and semantics are defined in \Cref{sec:semantics}. \lang\ has two basic types: integers ($\sortint$) and addresses~($\addr$). For simplicity, these language types directly correspond to their mathematical counterparts given by the sort interpretation function $\sortinterp$, defined by $\sortinterp(\sortint) = \mathbb{Z}$ and $\sortinterp(\addr) = \mathbb{N}$. In addition, the $\obj$ type allows defining non-recursive algebraic data-types (ADTs)~\cite{BarFT-RR-17} within the language, adding support for product and sum types. The $\obj$ type can, for instance, be declared as the union of all types that can be on the heap, each of which can itself be an ADT type. We refer to~\cite{DBLP:conf/smt/EsenR22} for an example in the context of the theory of heaps. We interpret the $\obj$ type as the object sort ($\obj$) in the semantics. On the semantic side we also use the heap sort ($\hp$), interpreted as a sequence of objects. Note that the operations ($\rd$, $\wt$, $\alloc$) are distinct from the program statements ($\readkw$, $\writekw$, $\allockw$), and are semantic functions defined by the theory heaps as shown in Table~\ref{tbl:heap-ops}.
The notation $\tupleget{t}{i}$ denotes the $i^{\text{th}}$ component of a tuple $t$.

A \emph{stack} $s$ maps variables to their values. The notation $s(x)$ accesses the value of variable $x$, and~$s[x \mapsto v]$ denotes the stack identical to $s$ except that variable $x$ maps to value $v$. The evaluation of an expression $e$ under stack $s$ is denoted by $\sem{e}$. The function $\progvars(p)$ returns the set of variables in program $p$.

% \subsection{The Theory of Heaps} \label{sec:theory-of-heaps}
\begin{table}
\caption{Theory of heaps operations and their interpretations as defined in~\cite{DBLP:conf/lopstr/EsenR20}. 
    %The sorts are interpreted as $\sortinterp(\hp) = \obj^*$ and $\sortinterp(\addr) = \mathbb{N}$.
}\label{tbl:heap-ops}
    \begin{tabular}{@{}p{0.15\linewidth} p{0.32\linewidth} p{0.45\linewidth}@{}}
    \toprule
    \textbf{Operation}  & \textbf{Signature} & \textbf{Interpretation} \\
    \midrule
    $\nullAddr$ & $() \to \addr$ & $\makebox[10.5ex][l]{$\langle  \rangle \mapsto$} 0$ \\
    $\eh$       & $() \to \hp$   & $\makebox[10.5ex][l]{$\langle  \rangle \mapsto$}\epsilon$ \\
    $\alloc$    & $\hp \times \obj \to \hp \times \addr$ & $\makebox[10.5ex][l]{$\langle h, o \rangle \mapsto$} \langle h \doubleplus [o], |h|+1 \rangle$ \\
    % $\ia$       & $\hp \times \addr \to \mathit{Bool}$ & $0 < a \leq |h|$ \\
    $\rd$       & $\hp \times \addr \to \obj$ & $\makebox[10.5ex][l]{$\langle h, a \rangle \mapsto$} \begin{cases} 
                                    h[a-1]~~~~~~& \text{if~} 0 < a \leq |h|\text{,} \\
                                    \defObj & \text{otherwise.} 
                                  \end{cases}$ \\
    $\wt$       & $\hp \times \addr \times \obj \to \hp$ & $\makebox[10.5ex][l]{$\langle h, a, o\rangle \mapsto$} \begin{cases} 
                                        h[a-1 \mapsto o]& \text{if~} 0 < a \leq |h|\text{,} \\
                                        h & \text{otherwise.} 
                                      \end{cases}$ \\
    \bottomrule
    \end{tabular}
    % \Description{The operation emptyHeap returns an empty heap that is invalid for all addresses, nullAddress returns a null address that is invalid in all heaps. allocate takes a heap and a object and returns a new heap where that object is allocated and a new address pointing to that object. read fetches an object residing at a heap x address pair, and write updates a heap x address pair with the passed object. The right-side of the figure provides interpretations for these operations where addresses are interpreted as the set of natural numbers and the heap is an array.}
\end{table}

\subsection{Syntax and Semantics of \lang}\label{sec:semantics}
\begin{figure}
\begin{align*}
  \text{Types} \quad & \tau ::= \sortint \mid \addr \mid \obj \\
  \text{Variables} \quad & p: \tau, \ x: \tau,  y: \tau, z: \tau, \dots \\
  \text{Constants} \quad & n \in \mathbb{Z},\ \nullkw: \addr, \defobjkw: \obj \\
  \text{Expressions} \quad & e ::= n \mid x \mid \text{unary-op}\ e \mid e\ \text{op}\ e \mid \nullkw \mid \defobjkw %\mid \havocexpr^\tau \\
  \mid \textit{ctor}(e, \ldots, e) \mid \textit{sel}(e) \mid \textit{is-ctor}(e)\\
  \text{Statements} \quad & \statements ::= \assigncmd{x}{e} \quad \text{($x$  and $e$ same type)} \\
  & \quad\mid \alloccmd{p}{e} \mid \readcmd{x}{p}  \mid \writecmd{p}{e} \quad (p : \addr) \\
  % &\ \mid \validcmd{x} \quad (x : \addr) \\
  & \quad\mid \skipcmd \mid \statements;\ \statements \mid \ifcmd{e}{\statements}{\statements} \mid \whilecmd{e}{\statements} \\
  & \quad\mid \assumecmd{e} \mid \assertcmd{e} \mid \assumecmd{P(e_1, \dots, e_n)} \mid \assertcmd{P(e_1, \dots, e_n)} \\
\end{align*}

\vspace*{-2ex}
\caption{
The syntax of \lang. The language is deterministic and supports uninterpreted predicates that can be used inside $\assertkw$ and $\assumekw$ statements. Pointer arithmetic (i.e., arithmetic manipulation of $\addr$ variables) is not permitted.}\label{fig:syntax}
\Description{The program types are int or address. The language supports typical statements such as if and while statements, but also supports assume and assert statements, both over expressions, and over uninterpreted predicates.}
\end{figure}

% , and an algebraic data-type (ADT) for objects stored on the heap. The language also recognises the constructor, selector, and tester operations defined over ADTs with SMT-LIB v2.6~\cite{SMT-LIB-v2.6} semantics. We omit details about ADT syntax and semantics, which do not affect our proofs. \ze{Currently no examples that uses ADTs with this language -- our motivating example would need it though. If we do not use it maybe it can be dropped?}
%
The syntax of \lang\ is given in \Cref{fig:syntax}.
In \Cref{fig:syntax}, \(\text{op}\) includes standard arithmetic (\(+, -, \times, /\)) and logical (\(<, \leq, >, \geq, =, \neq, \land, \lor\)) operators, while \(\text{unary-op}\) includes the unary operators (\(-, \neg\)). The language also supports constructor ($\textit{ctor}$), selector ($\textit{sel}$) and tester ($\textit{is-ctor}$) expressions for ADTs with their semantics as defined in the SMT-LIB standard~\cite{BarFT-RR-17}.
Pointer arithmetic over $\addr$ variables is not permitted. We use the shorthand $\cmd{\ifshortcmd{e}{\statements}}$ for the statement $\cmd{\ifcmd{e}{\statements}{\skipcmd}}$.

% $\bigl\{\error(P,\bar{x}) \mid P\text{ a predicate symbol},\,\bar{x}\text{ a value list}\bigr\}$

We define the (partial) big-step evaluation function $\bigstepeval$ relative to the fixed interpretation $\interp$ as follows:
\[
\bigstepeval : \statements \times \stacks \times \mathit{Heap} \to \{\success,\error(P,\bar{v})\}\times \stacks\times \mathit{Heap}
\]
where $\success$ represents a successful evaluation and $\error(P, \bar{v})$ represents a failed evaluation due to a failing assertion over the atom $P(\bar{v})$. We use a special 0-ary predicate symbol $F$ to represent assertion failures over expressions. The function~$\bigstepeval$ is partial because it is undefined for inputs where an $\assumekw$ statement fails or a $\whilekw$ statement does not terminate; the semantics is given in \Cref{tbl:semantics-compact}.

\begin{table}
  \caption{The big-step operational semantics of \lang, defined using the partial big-step evaluation function~$\bigstepeval$, relative to the fixed interpretation $\interp$. The definition for the $\whilekw$ statement is given by a recursive equation that unfolds the loop body until the guard becomes false. This equation determines $\bigstepeval$ only for terminating executions; for non-terminating loops or failed $\assumekw$ statements, $\bigstepeval$ is undefined. $C$, $S$, $S_1$ and $S_2$ range over \lang\ statements, $s$ over stacks, and $h$ over heaps.}\label{tbl:semantics-compact}
\begin{tabular}{@{}ll@{}}
\toprule
\textbf{Statement $C$} & $\bigstepeval(C, s, h)$ \\
\midrule

$\assigncmd{x}{e}$
  & $\bigl(\success,\upd{x}{\sem{e}},h\bigr)$ \\

$\alloccmd{p}{e}$
  & $\text{let } (h',a) = \alloc(h,\sem{e}) \text{ in } \bigl(\success,\upd{p}{a},h'\bigr)$ \\

$\readcmd{x}{p}$
  & $\bigl(\success,\upd{x}{\rd(h,s(p))},h\bigr)$ \\

$\writecmd{p}{e}$
  & $\bigl(\success,s,\wt(h,s(p),\sem{e})\bigr)$ \\

$\skipcmd$
  & $\bigl(\success,s,h\bigr)$ \\

$\seqcmd{S_1}{S_2}$
  & $\begin{aligned}[t]
    \text{let } &(\sigma_1,s_1,h_1) = \bigstepeval(S_1,s,h) \\
    \text{in } &\begin{cases}
      (\error(P,\bar{v}),s_1,h_1) & \text{if } \sigma_1 = \error(P,\bar{v})\\
      \bigstepeval(S_2,s_1,h_1) & \text{if } \sigma_1 = \success
    \end{cases}
    \end{aligned}$ \\

$\ifcmd{e}{S_1}{S_2}$
  & $\begin{cases}
    \bigstepeval(S_2,s,h) & \text{if } \sem{e} = 0 \\
    \bigstepeval(S_1,s,h) & \text{otherwise}
    \end{cases}$ \\

$\whilecmd{e}{S}$
& $\begin{cases}
    \bigl(\success,s,h\bigr) & \text{if } \sem{e} = 0 \\
    \begin{aligned}
      \text{let } &(\sigma',s',h') = \bigstepeval(S,s,h) \\
      \text{in } &\begin{cases}
        (\error(P,\bar{x}),s',h') & \text{if } \sigma' = \error(P,\bar{x})\\
        \bigstepeval(C,s',h') & \text{if } \sigma' = \success
      \end{cases}
    \end{aligned} & \text{otherwise}
    \end{cases}$ \\[1ex]

$\assumecmd{e}$
  & $\begin{cases}
    \text{undefined} & \text{if } \sem{e} = 0 \\
    \bigl(\success,s,h\bigr) & \text{otherwise}
    \end{cases}$ \\

$\assumecmd{P(e_1,\dots,e_n)}$
  & $\begin{cases}
    \text{undefined} & \text{if } (\sem{e_i})_i^n \not\in \interp(P) \\
    \bigl(\success,s,h\bigr) & \text{otherwise}
    \end{cases}$ \\

$\assertcmd{e}$
  & $\begin{cases}
    \bigl(\error(F, ()),s,h\bigr) & \text{if } \sem{e} = 0 \\
    \bigl(\success,s,h\bigr) & \text{otherwise}
    \end{cases}$ \\

$\assertcmd{P(e_1,\dots,e_n)}$
  & $\begin{aligned}[t]
    %\text{let } t &= (\sem{e_1},\dots,\sem{e_n}) \\
    &\begin{cases}
      \bigl(\error(P,(\sem{e_1},\dots,\sem{e_n})),s,h\bigr) & \text{if } (\sem{e_1},\dots,\sem{e_n}) \not\in \interp(P) \\
      \bigl(\success,s,h\bigr) & \text{otherwise}
    \end{cases}
    \end{aligned}$ \\
\bottomrule
\end{tabular}
\end{table}

\begin{definition}[Program Execution]\label{def:execution}
  Given \alang\ program $p$, an interpretation $\interp$ and an initial configuration $(s_0, h_0)$, its \emph{execution} is the derivation
  $\bigstepeval(p,s_0,h_0) = (\sigma,s_f,h_f)$,
  with $\sigma \in \{\success,\error(P,\bar{v})\}$.
\end{definition}

% In a deterministic \lang\ program (i.e., without any $\havocexpr$ expressions), the same initial configuration always leads to the same execution. In a non-deterministic \lang\ program, each $\havocexpr$ affecting program state \emph{branches} the execution into as many executions as the possible values for that $\havocexpr$.

\begin{definition}[Safety]\label{def:safety}
  An \lang\ program $p$ is \emph{safe} if there is some $\interp$ under which no execution of $p$ (with any initial configuration) results in an error.
  Formally,
  \[
  \exists I.\ \forall s \in \stacks.\ \tupleget{\bigstepeval(p, s, \eh)}{1} \neq \error(P, \bar{v}) \quad \text{(for any predicate $P$ and $\bar{v}$).}
  \]
\end{definition}

\begin{definition}[Equi-safety]\label{def:equisafety}
  Two \lang\ programs $p_1$ and $p_2$ are equi-safe if $p_1$ is safe if and only if $p_2$ is safe.
\end{definition}

We call a translation $E: \statements \to \statements$ an \emph{encoding}, and say that an encoding $E$ is \emph{correct} if for every program $p \in \statements$, $p$ and its encoding $E(p)$ are equi-safe. Equi-safety implies both soundness and completeness of the encoding. Specifically: \begin{itemize} \item \emph{Soundness}: If the original program $p$ is unsafe, then its encoding $E(p)$ is also unsafe. \item \emph{Completeness}: If the encoded program $E(p)$ is unsafe, then the original program $p$ is also unsafe. \end{itemize}

\paragraph{Handling of non-determinism.}\label{par:handling-nondet}
It can be observed that \lang{} programs are always deterministic,
which is a prerequisite of our encodings. To translate non-deterministic
programs to \lang{}, we therefore have to rewrite those programs, introducing
additional program inputs to determinize execution. This can be done
using the following macro reading bits from an additional integer
input~$\havocvar$ and setting $x$ to an arbitrary value:
\begin{equation*}
  \dethavoccall{x} \triangleq
  \begin{array}[t]{@{}l@{}}
    x \assignsym -(\havocvar\%2);\;
    \havocvar \assignsym \havocvar / 2;\\
    \whilecmd{\havocvar\%2 = 1}{
    \havocvar \assignsym \havocvar / 2; 
    x \assignsym 2\times x + (\havocvar\%2);
    \havocvar \assignsym \havocvar / 2
    };\\
    \havocvar \assignsym \havocvar / 2
  \end{array}
\end{equation*}
As the bits are consumed from $\havocvar$, calling the macro
repeatedly produces a sequence of independent \emph{havoc}ed values
determined by the initial value of $\havocvar$. This ensures
that every finite execution of a non-deterministic program is modelled by
some execution of the resulting deterministic program.

An alternative
implementation, also discussed in \Cref{sec:eval-impl}, is to add arrays as
a type and read subsequent values from an input array. If non-deterministic
values of different types are required, the types can be boxed into a single
algebraic data-type, or multiple arrays can be used, one for each type. 
These arrays serve solely as sources of non-deterministic values and are never modified by the program.
We refer to the combination of the original program inputs and these
additional auxiliary inputs (either $\havocvar$ or array(s)) as the
\emph{global input} $\inG$. For simplicity, we treat $\inG$ as a single
integer in our presentation.

% Depending on the implementation of the $\dethavockw$ macro, $\inG$ can be a single integer variable, a tuple, or an array, provided it covers program inputs and the required bits for $\dethavockw$. Multiple value types can be supported by using one array per value type. Arrays in this context serve solely as input sources for non-deterministic operations and do not re-introduce the complexity of mutable heap structures.

Our encodings also introduce non-determinism themselves (line~17 in
\Cref{lst:full-encoding}). For consistency, we use the
$\mathit{havoc}$ macro also to represent those statements. In
principle, however, there is no need to keep the program \emph{after}
the encoding deterministic; since verification tools like \seahorn\
and \tricera\ provide native support for $\mathit{havoc}$, it would be
fine (and more efficient) to use such native non-deterministic
statements as part of our encoding.

%In our encodings, we use a macro, $\dethavockw(x)$, that deterministically assigns an arbitrary value to a variable $x$. This macro is straightforwardly implementable within \lang\ using standard constructs (typically via bitwise operations) applied to an additional program input, for example as shown on the right. The value is derived from an additional program input, represented by the variable $\havocvar$, which is initialised to an arbitrary integer at the start of execution. The macro consumes the bits of $\havocvar$ to construct a new integer, therefore the sequence of \emph{havoc}ed values is fully determined by the initial value of $\havocvar$. In languages with non-determinism, the havoc expressions used in our encodings can be replaced with a non-deterministic havoc without affecting the correctness of the encodings, this is because in the encodings the $\dethavockw$ calls are immediately followed by $\assumekw$ statements, where only a single value from that havoc \emph{survives} in the rest of the execution. However, for the encodings we present to remain correct, it is important that the input program itself is deterministic: it cannot contain any havoc expressions except for assigning an arbitrary value for the program input. This is not a limitation, as a deterministic havoc can be implemented as described earlier.

\subsection{Fixed-Point Interpretation of Predicates}\label{sec:ico}
In our correctness proofs, we will rely on the fact that the semantics of uninterpreted predicates can also be defined through a fixed-point construction, deriving the strongest interpretation in which all $\assertkw$ statements hold. We first define the \emph{immediate consequence operator} that iteratively refines predicate interpretations based on assertion failures.

\begin{definition}[Immediate Consequence Operator]
  For program $p$, the immediate consequence operator $\ico_p: \mathit{Interps} \to \mathit{Interps}$ is defined as:
\begin{equation}\label{eq:ico}
  \ico_p(\interp)(r) = \interp(r) \cup \left\{ 
      \bar{v} \mid \exists s \in \stacks.\ \tupleget{\bigstepeval(p, s, \eh)}{1} = \error(r, \bar{v}) 
  \right\}
\end{equation}
  where $\eh$ denotes the empty heap. For each predicate $r$, $\ico_p$ adds all value tuples $\bar{v}$ to a given interpretation~$\interp$ that cause assertion failures when starting from $\eh$.
\end{definition}

\begin{lemma}[Monotonicity of $\ico$]
  For any program $p$, $\ico_p$ is monotonic: if $\interp_1 \sqsubseteq \interp_2$, then $\ico_p(\interp_1) \sqsubseteq \ico_p(\interp_2)$.
\end{lemma}

\begin{proof}
  Assume $\interp_1 \sqsubseteq \interp_2$ (i.e., $\forall r.\; \interp_1(r) \subseteq \interp_2(r)$). Fix an arbitrary predicate $r$, for $i \in \{1, 2\}$ let
  \[
    Q_i = \left\{ \bar{v} \mid \exists s \in \stacks.\ \tupleget{(\bigstepevalp{\interp_i}(p, s, \eh))}{1} = \error(r, \bar{v}) \right\}.
  \]
  By \eqref{eq:ico}, $\ico_p(\interp_i)(r) = \interp_i(r) \cup Q_i$. To show that $\ico$ is monotonic, we need to show $\interp_1(r) \cup Q_1 \subseteq \interp_2(r) \cup Q_2$.

  Let $\bar{v} \in \interp_1(r) \cup Q_1$. If $\bar{v} \in \interp_1(r)$, then $\bar{v} \in \interp_2(r)$ since $\interp_1 \sqsubseteq \interp_2$. If $\bar{v} \in Q_1$, there exists $(s, h)$ where $\bigstepevalp{\interp_1}(p, s, h)$ results in $\error(r, \bar{v})$. If $\bar{v} \notin \interp_2(r)$, then under $\interp_2$, the same execution would still result in~$\error(r, \bar{v})$, so $\bar{v} \in Q_2$. Thus, $\interp_1(r) \cup Q_1 \subseteq \interp_2(r) \cup Q_2$.
\end{proof}

\paragraph{Fixed-Point Construction}
Starting from $\interp_0$ where $\interp_0(r) = \emptyset$ for all predicates $r$, we iteratively compute a sequence~$\interp_0, \interp_1, \interp_2, \ldots$ by $\interp_{i+1} = \ico_p(\interp_i)$.
By Tarski's fixed-point theorem~\cite{tarski-fp}, the monotonicity of $\ico$ guarantees the existence of a least fixed point~$\interp^*$ of $\ico$, which is the limit of this sequence. All assertions of the form $\mathtt{assert}(r(\bar{v}))$ in the program hold under $\interp^*$, as any violating arguments $\bar{v}$ would eventually have been included in $\interp(r)$ by $\ico$. Other assertions, over concrete program properties, not over uninterpreted predicates, may still fail under $\interp^*$.

% The domain of $\ico$ is a complete lattice, where interpretations are mappings $\invs \to \mathcal{P}(D^*)$\ze{Define $D$}, ordered pointwise by $\sqsubseteq$. By Tarski's fixed-point theorem~\cite{tarski-fp}, the monotonicity of $\ico$ guarantees the existence of a least fixed point when iterating from the empty interpretation $\interp_0$ (where $\interp_0(r) = \emptyset$ for all predicates $r$). This least fixed point represents an interpretation $\interp$ where no assertions over uninterpreted predicates fail in the program $p$, as all values causing failures are incrementally added to $\interp$ by $\ico_p$.

% The sequence $\interp_0, \interp_1, \interp_2, \ldots$ is defined as:
% \begin{itemize}
%   \item $\interp_1 = \ico_p(\interp_0)$,
%   \item $\interp_2 = \ico_p(\interp_1)$,
%   \item $\cdots$,
%   \item $\interp_{i} = \ico_p(\interp_{i-1})$,
% \end{itemize}
% and the least fixed point is the limit $\bigcup_{n \in \mathbb{N}} \ico_p^n(\interp_0)$. For example, if $\ico_p$ iteratively adds values $\bar{v}_1, \bar{v}_2, \ldots$ to the interpretation of $r$, the least fixed point is $\{r \mapsto \{\bar{v}_1, \bar{v}_2, \ldots\}\}$.

While our setting involves programs with uninterpreted predicates, the operator $\ico$ mimics the immediate consequence operator for CHCs~\cite{DBLP:journals/tplp/AngelisFGHPP22}, and the least fixed point in our setting corresponds to the least model (or least solution) of a set of CHCs. In practice, programs with uninterpreted predicates can be encoded into CHCs (this is supported by Horn-based model checkers \seahorn~\cite{seahorn,DBLP:conf/ecoop/WesleyCNTWG24} and \tricera~\cite{tricera}), allowing the use of off-the-shelf Horn solvers.

\begin{lemma}[Safety under $\interp^*$]\label{lem:safety-under-i}
  Let $p$ be \alang\ program, and $\interp^*$ be the least fixed point of the immediate consequence operator $\ico_p$. The program~$p$ is \emph{safe} if and only if it is safe under $\interp^*$. Formally,
  \begin{equation}
  \exists I.\ \forall s \in \stacks.\ \tupleget{\bigstepeval(p, s, \eh)}{1} \neq \error(P, \bar{v}) \quad \text{(for any predicate $P$ and $\bar{v}$)}\label{eq:safety-under-i}
  \end{equation}
  if and only if
  \begin{equation}
    \forall s \in \stacks.\ \tupleget{\bigstepevalp{\interp^*}(p, s, \eh)}{1} \neq \error(P, \bar{v}) \quad \text{(for any predicate $P$ and $\bar{v}$).}\label{eq:safety-under-istar}
  \end{equation}
\end{lemma}

\begin{proof}
  ($\Rightarrow$) Assume~\eqref{eq:safety-under-i} for some $\interp$. No $\assertkw$ statement fails in $p$ under $\interp$ for any initial configuration. 
  An $\assertkw$ statement can be (i) over uninterpreted predicates or (ii) over expressions.
  Under the fixed point $\interp^*$, no $\assertkw$ over an uninterpreted predicate can fail, therefore we only need to show that no $\assertkw$ over an expression can fail under $\interp^*$. The evaluation of an $\assertkw$ over an expression can differ only if an $\assumekw$ statement over an uninterpreted predicate passes for different predicate arguments under the two interpretations; however, this is not possible in a deterministic program.
  %
  % An assert in (ii) is of the form $\assertcmd{e}$. Let $s_1$ be the stack at the point of evaluating this command under $\interp$, and $s_2$ under $\interp^*$. By assumption, $\semtwo{e}{s_1} = 1$. For $\semtwo{e}{s_2}$ to evaluate differently, we need $s_1 \neq s_2$. However, this is not possible in a deterministic program.
  Consider the statement $\assumecmd{P(\bar{v})}$. If this statement passes under~$\interp^*$, then under $\interp$ it will pass too and the execution will continue under the same post-state. If the $\assumekw$ fails the result is undefined, and no $\assertkw$ statement can fail.

  ($\Leftarrow$) Assume~\eqref{eq:safety-under-istar}. Choose $\interp$ to be $\interp^*$, and this direction trivially holds.
\end{proof}

%% file: content/encodings.tex
% encodings.tex
\section[The Time-Indexed Heap Invariant Encoding R (EncR)]{The Time-Indexed Heap Invariant Encoding $R$ ($\encR$)}\label{sec:r-encoding}
Given \alang\ program $p$, the $R$ encoding $\encR$ rewrites $p$ into an \emph{equi-safe} \lang\ program~$\encR(p)$ that is free of $\addr$ variables and the heap operations $\readkw$, $\writekw$ and $\allockw$. In $p$ we assume the variable $\inG$ represents the (arbitrarily chosen) program input. Note that we assume $\inG$ includes any additional inputs, such as $\havocvar$,
that were introduced to make the program deterministic 
(see \Cref{par:handling-nondet}). The encoding $\encR$ introduces an
uninterpreted predicate $R$ with the signature
$R : (\mathit{in} : \integer, \mathit{cnt} : \integer, \mathit{obj} : \obj)$
%
% $R : \integer \times \integer \times \tau $
%
where the first argument is always the program input $\inG$, the second argument is the $\cnt$ value of the read, and the third argument is the heap object of type $\obj$.

\subsection[Rewriting p into EncR(p)]{Rewriting $p$ into $\encR(p)$}\label{sec:r-encoding-rules}
$\encR$ rewrites $p$ through the following steps (in order):
\begin{itemize}
\item First, all $\addr$ variables in $p$ are redeclared as $\sortint$ variables with the same names. Casting of $\addr$ values is not needed, because only $\allockw$ and assign ($\assignsym$) statements modify an $\addr$ variable (recall that arithmetic over $\addr$ variables is not permitted). This step is needed, because we model allocation by incrementing an $\sortint$ variable ($\allocctr$) and assigning its value to the allocated variable. This mirrors the semantics of the $\addr$ sort (interpreted as $\mathbb{N}$) and the $\alloc$ operation in the theory of heaps (\Cref{tbl:heap-ops}). The resulting intermediate program is not well-typed if it contains any heap operations, which will be fixed after the final rewriting step.
\item Next, the fresh auxiliary variables $\allocctr$, $\cnt$, $\lastn$ and $\pg$ are introduced, and some of them are initialised by adding the statement in the ``Initialisation'' row of \Cref{tbl:r-encoding-rules} for $\encR$ to the start of the program from the previous step.
The uninitialised variable $\pg$ will be assigned an arbitrary value by the stack in the initial configuration.
\item Finally, rewrite rules are applied to replace occurrences of $\allockw$, $\readkw$, $\writekw$ in $p$ with the code snippets in the middle column of \Cref{tbl:r-encoding-rules}.
% , with $\dethavockw$ defined as the macro}
% \[\dethavoccall{x} \triangleq
% % \cmd{
%   \assigncmd{x}{0};\ 
%   \whilecmd{\havocvar~\text{mod}~2}{
%     \assigncmd{\havocvar}{\havocvar / 2};\
%     \assigncmd{x}{x * 2 + (\havocvar~\text{mod}~2)};\
%     \assigncmd{\havocvar}{\havocvar / 2}
%   }
% % }
% \]
% \textcolor{red}{where all occurrences of $x$ inside the macro is replaced with the argument to $\dethavockw$. The purpose of $\dethavockw$ is to assign a random value to $\readtmp$ based on the value of $\havocvar$.
%, which will be determined by the stack in the initial configuration (i.e., a \emph{deterministic} version of the regular $\havocexpr$ that is commonly available in verification-aware languages). 
% Each time $\dethavockw$ is called, it sets the bits of its argument to an arbitrary value that is determined by the value $\havocvar$ has in the initial stack.
\end{itemize}

% R encoding
\begin{table}
  \caption{Rewrite rules for the R ($\encR$) and RW ($\encRW$) encodings. The rules are applied once to every statement in the input program $p$, after redeclaring all $\addr$ variables as $\sortint$s and adding auxiliary variables.}\label{tbl:r-encoding-rules}
\begin{tabular}{lll}
  $p$ statement & $\encR(p)$ statement & $\encRW(p)$ statement \\\toprule
  Initialisation 
    & $\begin{array}[t]{@{}l@{}}
        \assigncmd{\allocctr}{0};\ \assigncmd{\cnt}{0};\ \assigncmd{\lastn}{\defObj}
      \end{array}$
    & $\begin{array}[t]{@{}l@{}}
        \assigncmd{\allocctr}{0};\ \assigncmd{\cnt}{0};\ \assigncmd{\lastwtcnt}{0};\\
        \assigncmd{\readtmpcnt}{0};\ \assertcmd{W(\inG, 0, \defObj)}
      \end{array}$\\\midrule
  %
  % alloc
  $\alloccmd{p}{e}$  &
  % encR(alloc)
  $ \begin{array}[t]{@{}l@{}}
      \assigncmd{\allocctr}{\allocctr+1} ;\ 
      \assigncmd{p}{\allocctr} ;\\
      \ifshortcmd{\pg = p}
      {\assigncmd{\lastn}{e}}
    \end{array}
  $&
  % encRW(alloc)
  $ \begin{array}[t]{@{}l@{}}
      \assigncmd{\allocctr}{\allocctr+1} ;\ 
      \assigncmd{p}{\allocctr} ;\\
      \assigncmd{\cnt}{\cnt + 1} ;\\
      \assertcmd{W(\inG, cnt, e)};\\
      \ifshortcmd{\pg = p}{\assigncmd{\lastwtcnt}{\cnt}}
    \end{array}
  $
  \\\midrule
  %
  % read
  $\readcmd{x}{p}$  &
  % encR(read)
  $\begin{array}[t]{@{}l@{}}
    \assigncmd{\cnt}{\cnt+1};\\       % cnt++;
    \ifcmd{\pg = p}
      { % then
        \\ \ \ \assertcmd{R(\inG, \cnt, \lastn)};\\
        \ \ \assigncmd{\readresult}{\lastn} \\
      } {\\ % else
         \ \ \dethavoccall{\readresult};\\
         \ \ \assumecmd{R(\inG, \cnt, \readresult)}\\ 
      }
  \end{array}
  $&
  % encRW(read)
  $\begin{array}[t]{@{}l@{}}
    \assigncmd{\cnt}{\cnt+1};\\       % cnt++;
    \ifcmd{\pg = p}
      { % then
        \\ \ \ \assertcmd{R(\inG, \cnt, \lastwtcnt)};\\
        \ \ \assigncmd{\readtmpcnt}{\lastwtcnt}\\
      } {\\ % else
      \ \ \dethavoccall{\readtmpcnt};\\
      \ \ \assumecmd{R(\inG, \cnt, \readtmpcnt)}\\
      };\\
      \dethavoccall{x};\\
      \assumecmd{W(\inG, \readtmpcnt, x)};\\
  \end{array}$
  \\\midrule
  %
  % write
  $\textbf{write}(p,e)$   &
  % encR(write)
  $\begin{array}[t]{@{}l@{}}
    \ifshortcmd{\pg = p \wedge 0 < p \leq \allocctr}
    {
      \\\ \ \assigncmd{\lastn}{e} \\
    }
  \end{array}
  $&
  % encRW(write)
  $\begin{array}[t]{@{}l@{}}
    \assigncmd{\cnt}{\cnt+1};\\ % cnt++;
    \ifshortcmd{0 < p \leq \allocctr}
    {\\
    \ \ \assertcmd{W(\inG, \cnt, e)}\\
    \ \ \ifshortcmd{\pg = p} {\assigncmd{\lastwtcnt}{\cnt}};\\
    }
  \end{array}$\\\bottomrule
\end{tabular}
\end{table}

\subsection[Correctness of EncR]{Correctness of $\encR$}\label{sec:r-encoding-proof}
We show the correctness of $\encR$ by showing that given \alang\ program $p$, $p$ is \emph{equi-safe} with $\encR(p)$. The core observation used in the proof is that the relation~$R$, obtained as the least fixed point of the immediate consequence operator, correctly represents the values read from the heap. 
The detailed proofs and proof sketches of the lemmas in this section are provided in \Cref{app:proof-r}.

We prove the correctness of $\encR$ in multiple steps; the first step is to show that the relation $R$ is a \emph{partial function} that maps the program inputs~$\inG$ and the read count~$\cnt$ to the value that is read:
\begin{lemma}[Functional consistency of $\interp^*(R)$]
  \label{prop:R-partial}
  Let $\interp^*$ be the least fixed point of the immediate consequence operator $\ico_p$ for \alang\ program $p=\encR(q)$ obtained as the $R$ encoding of some program~$q$. Then $\interp^*(R)$ is a partial function from its first two arguments to its third argument:
  \begin{equation}
    \label{eq:functional}
    \forall g, n, v_1, v_2.\ \Bigl( (g, n, v_1) \in \interp^*(R) \wedge (g, n, v_2) \in \interp^*(R) \Bigr)
      \Longrightarrow v_1 = v_2.
  \end{equation}
\end{lemma}

The lemma follows from the shape of the code snippet introduced by $\encR$ for \readkw{} statements and can be proved by induction on the iteration count~$\alpha$ in approximations $\interp_\alpha$ of the least fixed point~$\interp^*$. We use this lemma in the proof of \Cref{lem:encR-preservation}.

As the next step for proving the correctness of $\encR$, we introduce an intermediate encoding~$\encN$, and show that $p$ is equi-safe with $\encN(p)$. The purpose of this encoding is to introduce a counter that is incremented by each encoded statement, which we use in the inductive proof of correctness for $\encR$.

\paragraph{The Encoding $\encN$.}
Given \alang\ program $p$, $\encN(p)$ is obtained by introducing an $\sortint$ variable $\encNc$,
and before every statement $S$ over one of $\{\writekw, \allockw, \readkw\}$ in $p$, inserting
  \[\assigncmd{\encNc}{\encNc - 1};\ \assumecmd{\encNc \ge 0}.\]
Starting from the same stack $s$ with $n = s(c)$, the executions of $p$ and $\encN(p)$ will remain identical up until the $(n+1)$-th evaluation of any statement $S$ over one of $\{\writekw, \allockw, \readkw\}$ (apart from the value of $\encNc$ in $\encN(p)$), after which the evaluation of $\assumekw$ that was inserted right before $S$ will fail and program execution is stopped.
%The result of the execution, as determined by $\bigstepeval$, will be of the form ...

\begin{lemma}[$\encN$ is correct]\label{lem:encN-proof}
  Let $p$ be \alang\ program, and $\interp$ some interpretation, then $p$ and $\encN(p)$ are equi-safe, i.e.,
  % \[ \left( \exists s \in \stacks. \tupleget{\bigstepevalp{\interp}(p, s, \eh)}{1} = \error(P_1, \bar{v}_1) \right) \iff \left( \exists s \in \stacks. \tupleget{\bigstepevalp{\interp}(\encN(p), s, \eh)}{1} = \error(P_2, \bar{v}_2) \right).\]
  $p$ is safe if and only if $\encN(p)$ is safe.
\end{lemma}

Finally, we state the last lemma needed to show the correctness of $\encR$.

\begin{lemma}[Preservation of final states by $\encR$]\label{lem:encR-preservation}
  Let $p$ be \alang\ program, and $\interp^*$ be the least fixed point of $\ico_{p}$. We use abbreviations $\pstar = \encN(p)$ and $\pstarstar = \encR(\pstar)$,
 %
  % Let $s$ be an arbitrary stack, that also assigns values to the auxiliary variables introduced by $\encN$ and $\encR$.
  $\bigstepevalp{\interp^*}(\pstar, s, \eh) = (\sigma_1, s_1, h_1)$, and
  $\bigstepevalp{\interp^*}(\pstarstar, s, \eh) = (\sigma_2, s_2, h_2)$.
  Then the following holds:
  %for all $n,a\in\mathbb{N}$, if $n = s(\encNc)$ and $a = s(\pg)$, the following holds:
  \vspace{-1em}
\begin{equation}
\begin{array}{c@{\quad}l}
\vcenter{\hbox{$\forall s \in \stacks,\, a \in \integer.\; n = s(\encNc) \wedge a = s(\pg) \Rightarrow$}} &
\begin{aligned}
\overbrace{\sigma_1 = \sigma_2}^{\foutcomes} \wedge
\overbrace{\forall v \in \progvars(\pstar).\; s_1(v) = s_2(v)}^{\fstacks} \wedge\\
\underbrace{\rd(h_1, a) = s_2(\lastn)}_{\freads} \wedge
\underbrace{|h_1| = s_2(\allocctr)}_{\fallocs}
\end{aligned}
\end{array}.
\label{eq:r-preservation}
\end{equation}

  % \begin{equation}
  %   \overbrace{\sigma_1 = \error(F, ()) \leftrightarrow \sigma_2 = \error(F, ())}^{\foutcomes} \wedge
  %   \overbrace{\forall v \in \progvars(\pstar).\ s_1(v) = s_2(v)}^{\fstacks} \wedge
  %   \overbrace{\rd(h_1, a) = s_2(\lastn)}^{\freads} \wedge
  %   \overbrace{|h_1| = s_2(\allocctr)}^{\fallocs}\label{eq:r-preservation}
  % \end{equation}
  % \begin{enumerate}
  %   \item[$\foutcomes$:] $\sigma_1 = \error(F,()) \leftrightarrow \sigma_2 = \error(F,())$
  %   \item[$\fstacks$:]   $\forall v \in \progvars(\pstar).\ s_1(v) = s_2(v)$
  %   \item[$\freads$:]    $\rd(h_1, a) = s_2(\lastn)$
  %   \item[$\fallocs$:]   $|h_1| = s_2(\allocctr)$
  % \end{enumerate}
  % \begin{align}
  %   \begin{split}
  %   \forall n, a \in \mathbb{N}.\ n = s(\encNc) \wedge a = s(\pg) \rightarrow &
  %     \overbrace{\sigma_1 = \error(F, ()) \leftrightarrow \sigma_2 = \error(F, ())}^{\foutcomes} \wedge
  %     \overbrace{\forall v \in \progvars(\pstar).\ s_1(v) = s_2(v)}^{\fstacks} \wedge \\&
  %     \underbrace{\rd(h_1, a) = s_2(\lastn)}_{\freads} \wedge
  %     \underbrace{|h_1| = s_2(\allocctr)}_{\fallocs}\label{eq:r-preservation}
  %   \end{split}
  % \end{align}
\end{lemma}

% Since $\pstar$ is deterministic, it has a unique execution for each initial configuration $(s,\eh)$. $\pstarstar$ introduces an $\havocexpr$ expression for each rewritten $\readkw$ statement, and is therefore nondeterministic if $\pstar$ contains at least one $\readkw$. 
The lemma states that, for the same initial configuration $(s,\eh)$, both $\pstar$ and $\pstarstar$ will result in the same outcome ($\foutcomes$), with the same values for all common variables in the final stacks~($\fstacks$), with the variable $\lastn$ holding the same object that is stored at $\pg$ in the final heap $h_1$ of $\pstar$ ($\freads$), and with the value of $\allocctr$ in $\pstarstar$ matching the size of the heap $h_1$ in~$\pstar$~($\fallocs$).
The variable $n$ is free and represents an arbitrary natural number; the lemma is proven by induction on $n$.
We will use this lemma in \Cref{thm:encR-proof} to show that $\pstar$ and $\pstarstar$ are equi-safe (\Cref{def:equisafety}), which is a weaker claim.

\begin{theorem}[$\encR$ is correct]\label{thm:encR-proof}
  Let $p$ be \alang\ program. $p$ and $\encR(p)$ are equi-safe.
  % \[ \left( \exists s \in \stacks. \tupleget{\bigstepevalp{\interp}(p, s, \eh)}{1} = \error(P_1, \bar{v}_1) \right) \iff \left( \exists s \in \stacks. \tupleget{\bigstepevalp{\interp}(\encR(p), s, \eh)}{1} = \error(P_2, \bar{v}_2) \right).\]
\end{theorem}
\begin{proof}
  Let $\pstar = \encN(p)$ and $\pstarstar = \encR(\pstar)$. By \Cref{lem:encN-proof}, $p$ and $\pstar$ are equi-safe, so it suffices to show equi-safety between $\pstar$ and $\pstarstar$.
  By \Cref{lem:safety-under-i}, it suffices to show equi-safety under the least interpretation $\interp^*$.

  By \Cref{lem:encR-preservation}, for all $s$ we have $\tupleget{\bigstepevalp{\interp^*}(\pstar, s, \eh)}{1} = \sigma_1 = \sigma_2 = \tupleget{\bigstepevalp{\interp^*}(\pstarstar, s, \eh)}{1}$. That is, $\pstar$ and $\pstarstar$ always have the same outcome for every initial stack $s$; therefore if one is safe, the other will be safe too. By \Cref{lem:encN-proof} and transitivity, this shows $p$ and $\encR(p)$ are equi-safe.
\end{proof}

\section[The Time-Indexed Heap Invariant Encoding RW (EncRW)]{The Time-Indexed Heap Invariant Encoding $RW$ ($\encRW$)}\label{sec:rw-encoding}
Using the $R$ encoding it is sometimes difficult to express invariants only in terms of the $\cnt$ value when a read happens, making invariant inference more difficult.
The $RW$ encoding ($\encRW$, rightmost column of \Cref{tbl:r-encoding-rules}) introduces an additional uninterpreted predicate $W$ for write operations. Unlike the $R$ encoding, each write is also indexed using $\cnt$, and the $\cnt$ value of the last written object is stored along with the object. As opposed to only keeping track of the last written value to some address (in the $R$ encoding), using the additional $W$ predicate we store the richer object~$(\lastn, \lastwtcnt)$. The part of the encoding involving $\lastwtcnt$ is the same as the $R$ encoding, with~$\lastwtcnt$ replacing $\lastn$. Like the $R$ predicate, the $W$ predicate is functionally consistent; the proof is similar to the proof of \Cref{prop:R-partial}. The $W$ predicate is then used to \emph{look up} the object residing at that index.

Similarly to the $R$ encoding, 
given \alang\ program $p$, the $RW$ encoding $\encRW$ rewrites $p$ into an \emph{equi-safe} \lang\ program $\encRW(p)$ that is free of $\addr$ variables and the heap operations $\readkw$, $\writekw$ and $\allockw$. In $p$ we assume the variable $\inG$ represents the (arbitrary) program input.
The encoding $\encRW$ introduces the two uninterpreted predicates $R : (\mathit{in} : \integer, \mathit{cnt} : \integer, \mathit{cnt}_\mathit{last} : \integer)$ and $W : (\mathit{in} : \integer, \mathit{cnt} : \integer, \mathit{obj} : \obj)$,
%
% $R : \integer \times \integer \times \tau $
%
where the first argument to both predicates is always the program input $\inG$, the second argument is the $\cnt$ value of the read or the write. The third argument in $R$ is the $\cnt$ value of the last write, and in $W$ the last heap object of type $\obj$.

\subsection[Rewriting p into EncRW(p)]{Rewriting $p$ into $\encRW(p)$}\label{sec:rw-encoding-rules}
$\encRW$ rewrites $p$ through the following steps (in order):
\begin{itemize}
\item The first step is the same as in the $R$ encoding (i.e., converting $\addr$es to $\integer$s).
\item Next, the fresh auxiliary variables $\allocctr$, $\cnt$, $\lastwtcnt$, $\pg$ and $\readtmpcnt$ are introduced, and some of them are initialised by adding the statement to the start of the program from the ``Initialisation'' row of \Cref{tbl:r-encoding-rules} for $\encRW$ to the start of the program from the previous step. The $RW$ encoding replaces $\lastn$ with $\lastwtcnt$.
\item Finally, the rewrite rules in the rightmost column of \Cref{tbl:r-encoding-rules} are applied once.
\end{itemize}

\subsection[Correctness of EncRW]{Correctness of $\encRW$}
The correctness of the $\encRW$ encoding follows a similar argument to that of $\encR$. The proof relies on showing that both $R$ and $W$ are partial functions in $\interp^*$. An inductive argument then shows that the two-step lookup process, using $R$ to find a write-counter and $W$ to find the corresponding data, correctly simulates a heap read. 
A detailed proof sketch is provided in \Cref{app:proof-rw}.

%% file: content/extensions.tex
% JayHorn encoding is one instance, but we can create many others.

While the base $\encR$ and $\encRW$ encodings are both sound and complete, we explore several approximations and extensions designed to enhance scalability and solver efficiency, while allowing control over completeness through a controlled abstraction strategy.

\subsection{Adding Supplementary Information to the Uninterpreted Predicates}
\label{sec:supplementary}
One strategy to extend the base encodings ($\encR$ and $\encRW$) is to augment the uninterpreted predicates~$R, W$ with additional arguments. Those arguments can provide, in general, any information that is uniquely determined by the other arguments, but that might be difficult to infer for a verification system automatically. Such supplementary information does not affect the correctness (soundness and completeness) of the encodings; however, it has the potential to significantly simplify invariant inference. We provide two examples.

\paragraph{Adding Meta-Information}
% example - stretch goal
A simple kind of supplementary information that can be added to the uninterpreted predicates is meta-data, for instance the control location of writes and reads to/from the heap. Similar refinements were proposed also in the context of the space invariants encoding to improve precision~\cite{DBLP:conf/lpar/KahsaiKRS17}, where they required, however, a separate static analysis procedure for determining the data-flow from write to read statements. In our framework, meta-data can be added in a more elegant implicit way by extending the time-indexed heap invariants encoding. Adding meta-data has no effect on the precision of our encoding, since the encoding is already complete, but it can make life simpler for the back-end verification tool since simpler relational invariants can be found.

We consider our $R$ encoding, in which the uninterpreted predicate~$R$ initially has the signature $R : (\mathit{in} : \integer, \mathit{cnt} : \integer, \mathit{obj} : \obj)$. Suppose that $\mathcal{C}$ is a finite set representing the control locations of the program to be transformed, such that every \readkw{}, \writekw{}, and \allockw{} statement~$s$ has a unique control location~$s_{\mathit{Loc}} \in \mathcal{C}$. We extend the encoding by redefining the signature of $R$ as $R : (\mathit{in} : \integer, \mathit{cnt} : \integer, \mathit{obj} : \obj, \mathit{write}_{\mathit{Loc}} : \mathcal{C}, \mathit{read}_{\mathit{Loc}} : \mathcal{C})$, adding two arguments of type~$\mathcal{C}$ for specifying the control location of the \writekw{} or \allockw{} that put data on the heap, as well as the control location of the \readkw{} that is reading the data, respectively.

\begin{wrapfigure}{r}{0.41\textwidth}
\vspace{-1.5em}
%\begin{figure}[h]
\begin{minipage}{0.4\textwidth}
\begin{equation*}
    \begin{array}[t]{@{}l@{}}
    \assigncmd{\cnt}{\cnt+1};\\       % cnt++;
    \ifcmd{\pg = p}
      { % then
        \\ \ \ \assertcmd{R(\inG, \cnt, \lastn, \mathit{last}_{\mathit{Loc}}, \mathit{read}_{\mathit{Loc}})};\\
        \ \ \assigncmd{\readresult}{\lastn} \\
      } {\\ % else
         \ \ \dethavoccall{\readresult};\; \dethavoccall{l};\\
         \ \ \assumecmd{R(\inG, \cnt, \readresult, l, \mathit{read}_{\mathit{Loc}})}\\ 
      }
  \end{array}
\end{equation*}
\end{minipage}
%\caption{Caption}\label{fig:enter-label}
%\end{figure}
\vspace{-1em}
\end{wrapfigure}
The transformation~$\encR$ in \Cref{tbl:r-encoding-rules} can be extended accordingly by adding a further variable~$\mathit{last}_{\mathit{Loc}}$ that is updated by the code snippets for \writekw{} or \allockw{}, recording the control location at which the write occurred (in conjunction with updating the $\lastn$ variable), and by adding arguments for the control locations of the write and read to the \assertkw{} and \assumekw{} statements in the encoding of \readkw{}, where $\mathit{read}_{\mathit{Loc}} \in \mathcal{C}$ is the control location of the encoded \readkw{}, as on the right.

The soundness of the augmented encoding follows directly from the fact that write and read locations are uniquely determined by the values of the $\inG$ and $\cnt$ variables; the added information is redundant, but can often help the back-end solver to find simpler invariants~\cite{DBLP:conf/lpar/KahsaiKRS17}. For instance, in the extended encoding, a time-indexed heap invariant could now state that some particular \writekw{} statement only writes data in the range~$[0, 10]$, or that a particular \readkw{} statement can only read data that was produced by certain \writekw{} statements.

In \Cref{sec:evaluation} we call this \emph{the tagging extension}, and denote it by subscripting encoding names with $T$ (e.g., $\mathit{R_T}$).

%In addition to adding variables in scope to the invariants, meta-information such as line numbers or allocation sites of accesses heap locations can be made part of the invariant.  Our approach is already complete; however, similar techniques can be used to aid invariant inference. 

\paragraph{Adding Variables in Scope}
%\todo{this is vague, we should say more precisely what the extended signature of the predicates is and how the encoding has to be changed}
Similarly, the values of global or local program variables can be added as further arguments to the uninterpreted predicates. In the motivating example, for instance, constant values were written to the heap, so the inferred invariants needed only to reason using those constants. Consider a simple loop that writes the loop index to the heap at each iteration.
%
% \[
%     \whilecmd{i < N}{ \writecmd{p}{i}; \assigncmd{i}{i + 1};},
% \]
%
That is, the values written to the heap do not remain constant between the iterations of the loop. It is still possible to derive the assigned value from the program input and the recorded $\cnt$ values; however, with the loop index part of the invariant a much simpler invariant becomes expressible.
% example - stretch goal

\subsection{Controlled Abstraction}
Another strategy is to deliberately use abstractions in exchange for improved solver performance. This can generally be done by simply \emph{removing} arguments of uninterpreted predicates. It is easy to see that removing arguments is a program transformation that is \emph{sound} but, in general, \emph{incomplete.}
For instance, removing the precise counter values tracking heap accesses from the arguments of the predicate~$R$ is an abstraction sacrificing completeness; however, in many cases, it can be sufficient for the verification task to replace the precise identifiers with abstract information such as control locations, or program variables as discussed earlier, giving rise to a systematic strategy for constructing heap encodings: we start from one of the encodings that are sound and complete~($\encR$ or $\encRW$), augment it with supplementary information (\Cref{sec:supplementary}), and finally remove predicate arguments that are deemed unnecessary. Every encoding constructed in this way is sound, while the degree of incompleteness can be controlled depending on how much information is kept. Different encodings presented in the literature (e.g., \cite{DBLP:conf/lpar/KahsaiKRS17}) can be obtained in this way.

% example - stretch goal

\subsection{Tailoring the Encodings to Properties of Interest}
It is possible to tailor the encodings to target different program properties, for instance to only support checking functional safety or memory safety. Consider the encodings $\mathit{RWf}$ and $\mathit{RWm}$, shown in \Cref{tbl:extension-encoding-rules}. Compared to the $\mathit{RW}$ encoding, the $\mathit{RWf}$ encoding omits the $\assertkw$ statement to the $W$ predicate during initialisation, and it does not write any object to newly allocated addresses\footnote{This is assuming C-like semantics where freshly-allocated locations are considered uninitialised.}. This encoding is correct only for \emph{memory-safe} programs: programs in which it has already been shown that all heap reads are from addresses that are allocated and initialised prior to that read (this encoding is suitable for verification tasks in the \texttt{ReachSafety-Heap} category of SV-COMP~\cite{svcomp24}).

$\mathit{RWm}$ is a variation of the $\mathit{RWf}$ encoding, and it is tailored for checking the absence of invalid pointer references. It does this by inserting $\assertkw$ statements into the encoded read and write statements that fail when there are invalid accesses. The original $\assertkw$ statements of the program can also be dropped if the only property of interest is the absence of invalid pointer dereferences. Note that the $\mathit{RW}$-$\mathit{mem}$ encoding cannot detect accesses to newly allocated but uninitialised addresses, such as accessing the result of a \texttt{malloc} in C before writing to it.

Other extensions could, for instance, add support for the \textbf{\texttt{free}} operation, and check more properties related to memory safety, such as lack of memory leaks and double free operations. We believe these are interesting research avenues to explore on their own, and that the equi-safe encodings we provide serve as a framework to build upon.

\begin{table}
  \caption{Rewrite rules for the $\mathit{RWf}$ ($\encRWfun$) and $\mathit{RWm}$ ($\encRWmem$) encodings. The rules are applied similarly to $R$ and $\mathit{RW}$ encodings.}\label{tbl:extension-encoding-rules}
\begin{tabular}{lll}
  $p$ statement & $\encRWfun(p)$ statement & $\encRWmem(p)$ statement \\\toprule
  Initialisation 
    & $\begin{array}[t]{@{}l@{}}
        \assigncmd{\allocctr}{0};\ \assigncmd{\cnt}{0};\\\assigncmd{\lastwtcnt}{0};\ \assigncmd{\readtmpcnt}{0}
      \end{array}$
    & $\begin{array}[t]{@{}l@{}}
        \assigncmd{\allocctr}{0};\ \assigncmd{\cnt}{0};\\\assigncmd{\lastwtcnt}{0};\ \assigncmd{\readtmpcnt}{0}
      \end{array}$\\\midrule
  %
  % alloc
  $\alloccmd{p}{e}$  &
  % encR(alloc)
  $ \begin{array}[t]{@{}l@{}}
      \assigncmd{\allocctr}{\allocctr+1} ;\\
      \assigncmd{p}{\allocctr}
    \end{array}
  $&
  % encRW(alloc)
  $ \begin{array}[t]{@{}l@{}}
      \assigncmd{\allocctr}{\allocctr+1} ;\\
      \assigncmd{p}{\allocctr}
    \end{array}
  $
  \\\midrule
  %
  % read
  $\readcmd{x}{p}$  &
  % encR(read)
  $\begin{array}[t]{@{}l@{}}
    \assigncmd{\cnt}{\cnt+1};\\\\       % cnt++;
    \ifcmd{\pg = p}
      { % then
        \\ \ \ \assertcmd{R(\inG, \cnt, \lastwtcnt)};\\
        \ \ \assigncmd{\readtmpcnt}{\lastwtcnt}\\
      } {\\ % else
      \ \ \dethavoccall{\readtmpcnt};\\
      \ \ \assumecmd{R(\inG, \cnt, \readtmpcnt)}\\
      };\\
      \dethavoccall{x};\\
      \assumecmd{W(\inG, \readtmpcnt, x)};\\
  \end{array}
  $&
  % encRW(read)
  $\begin{array}[t]{@{}l@{}}
    \assigncmd{\cnt}{\cnt+1};\\       % cnt++;
    \assertcmd{0 < p \leq \allocctr}\\
    \ifcmd{\pg = p}
      { % then
        \\ \ \ \assertcmd{R(\inG, \cnt, \lastwtcnt)};\\
        \ \ \assigncmd{\readtmpcnt}{\lastwtcnt}\\
      } {\\ % else
      \ \ \dethavoccall{\readtmpcnt};\\
      \ \ \assumecmd{R(\inG, \cnt, \readtmpcnt)}\\
      };\\
      \dethavoccall{x};\\
      \assumecmd{W(\inG, \readtmpcnt, x)};\\
  \end{array}$
  \\\midrule
  %
  % write
  $\textbf{write}(p,e)$   &
  % encR(write)
  $\begin{array}[t]{@{}l@{}}
    \assigncmd{\cnt}{\cnt+1};\\ % cnt++;
    \ifshortcmd{0 < p \leq \allocctr}
    {\\
    \ \ \assertcmd{W(\inG, \cnt, e)}\\
    \ \ \ifshortcmd{\pg = p} {\assigncmd{\lastwtcnt}{\cnt}};\\
    }
  \end{array}$&
  % encRW(write)
  $\begin{array}[t]{@{}l@{}}
    \assigncmd{\cnt}{\cnt+1};\\ % cnt++;
    \ifcmd{0 < p \leq \allocctr}
    {\\
    \ \ \assertcmd{W(\inG, \cnt, e)}\\
    \ \ \ifshortcmd{\pg = p} {\assigncmd{\lastwtcnt}{\cnt}};\\
    } { %else
    \assertcmd{0}
    }
  \end{array}$\\\bottomrule
\end{tabular}
\end{table}

\subsection{Caching}\label{subsec:caching}
Programs often access the same address repeatedly, which can make verification harder when every access is modelled as a distinct heap operation. In~\cite{DBLP:conf/lpar/KahsaiKRS17} the authors carry out a data-flow based analysis to minimise the number of heap interactions. We instead build a similar optimisation directly into our encodings through a simple caching mechanism, requiring no external analysis. The extension, which we denote using the subscript $C$ in encoding names (e.g., $\mathit{R_C}$), introduces a simple, one-element cache, but can be extended to any number of elements.

\begin{wrapfigure}{r}{0.41\textwidth}
\vspace{-1em}
\begin{minipage}{0.4\textwidth}
\begin{equation*}
    \begin{array}[t]{@{}l@{}}
    \assigncmd{\cnt}{\cnt+1};\\
    \ifcmd{\mathit{lastc_{addr}} = p}{ \\
    \ \ \assigncmd{x}{\mathit{lastc_{data}}} \\
    }{ \\
    \ \
      \textit{[\ldots] // existing code for read}\\
    \ \ \assigncmd{\mathit{lastc_{addr}}}{p};\ \assigncmd{\mathit{lastc_{data}}}{x} \\
    }
    \end{array}
\end{equation*}
\end{minipage}
\vspace{-1em}
\end{wrapfigure}
\
\ % these are needed to make wrapfigure behave
The extension adds two global variables: $\mathit{lastc_{addr}}$ to store the address of the last heap access, and $\mathit{lastc_{data}}$ for the corresponding data. Every heap operation modifies the cache. The transformation for a \readkw{} in the~$R_C$ encoding is shown on the right. If the read address matches~$\mathit{lastc_{addr}}$ (a cache hit), the cached value is returned immediately. Otherwise (a cache miss), the original relational logic is executed, and the cache is updated with the new address and value. The \writekw{}$(p, e)$ and \allockw{}$(p, e)$ operations are similarly extended to update $\mathit{lastc_{addr}}$ and $\mathit{lastc_{data}}$ with $p$ and $e$.

\subsection{Handling Different Heap Regions and Object Types}
Our encodings do not distinguish between different heap regions or types of objects. We carried out our experiments using a single global heap and used algebraic data-types (ADTs) in order to support multiple object types. A possible extension to make the approach more scalable is to assign each heap region and object type an own set of predicates. This can be done through an external static analysis to determine the different heap regions and object types on the heap.
For instance, when such an analysis can determine the type $\tau$ of each pointer, the encoding can introduce a separate set of heap operations for each type~$\tau$ ($\readkw_\tau$, $\writekw_\tau$ etc.), each with its own set of predicates and auxiliary variables.

%% file: content/evaluation.tex
%evaluation.tex

To evaluate the time-indexed heap invariants approach, we conducted experiments using two suites of benchmarks: a set of 22 manually crafted benchmarks, representing challenging heap-manipulating programs, and a set of 132 benchmarks from the SV-COMP \texttt{ReachSafety-Heap} category. 
Our evaluation aims at answering the following research questions:
\begin{enumerate}
    \item[\textbf{RQ1:}] Can time-indexed heap invariants, in combination with off-the-shelf verification tools for heap-free programs with uninterpreted predicates, be used to verify heap-manipulating programs that are beyond the capabilities of existing tools?
    \item[\textbf{RQ2:}] How do the different encoding variants affect verification performance?
    \item[\textbf{RQ3:}] How does time-indexed heap invariants approach compare to state-of-the-art verification tools on the standard SV-COMP benchmarks?
\end{enumerate}

We implemented two distinct pipelines to investigate those questions. We describe the implementation of these pipelines next, followed by details of the experimental setup and the results.

\subsection{Implementation}\label{sec:eval-impl}
\subsubsection{Semi-Automated Pipeline}
\label{sec:eval:semi}
For the crafted benchmark suite, we used a semi-automated process. We first manually normalised the C source code to replace heap operations with corresponding function calls (\texttt{read(p)}, \texttt{write(p, e)}, \texttt{alloc(e)}), and added annotations specifying input variables and the heap object type. Automated scripts then applied the encodings to these normalised programs, outputting C files with uninterpreted predicates that can be handled by the tools \seahorn{} and \tricera{}.

The goal of this pipeline is to demonstrate the viability of our approach using off-the-shelf tools on a curated set of benchmarks consisting of singly-linked lists and trees, while controlling for external factors that influence verification performance. In particular, \tricera{} represents C structs as algebraic data-types (ADTs), so updating a struct field involves reading the whole ADT, creating a new ADT with the modified field and writing it back. These additional reads can make verification harder, and are also problematic when a write happens directly after an allocation, as the read is then to an uninitialised location, breaking the memory-safety assumption of the ``$\mathit{fun}$'' encoding (a variant, detailed in \Cref{tbl:extension-encoding-rules}, that targets functional safety and thus assumes all memory accesses are valid). To avoid this, during normalisation we merge common \texttt{alloc-write}$^*$ patterns (e.g., struct allocation and initialisation) into a single \texttt{alloc}.

\subsubsection{\tricerare{}: Prototype in \tricera{}}
\label{sec:eval:auto}
The second pipeline is a prototypical implementation of our approach on top of the open-source verification tool \tricera{}, called \tricerare{}, replacing the native heap model of \tricera{} with our encodings. \tricerare{} is fully automatic and used for more large-scale experiments with SV-COMP benchmarks. However, at the point of submitting the paper, \tricerare{} is still in a relatively early stage of development and does not include many relevant optimisations, like the merging of \texttt{alloc-write}$^*$ patterns that was applied in the semi-automated pipeline. This is somewhat alleviated by the caching extension discussed in \Cref{subsec:caching}, which we used in a subset of the encodings. We also adjusted the ``$\mathit{fun}$'' variants of the encodings to write an initial object on allocation to make any potential reads after an allocation memory-safe.

\paragraph{Making Input Programs Deterministic.}
Our encodings require deterministic input programs, but many SV-COMP benchmarks contain non-deterministic calls (e.g., \texttt{nondet\_int()}). We implemented a Clang-based preprocessor to replace these calls with unique program input variables. For calls that can be executed multiple times from the same location (e.g., inside a loop), we use mathematical arrays as input variables with a corresponding counter that is incremented at each call. These input variables are then passed as arguments to the time-indexed heap invariants.

\paragraph{Translation to Constrained Horn Clauses.}
CHC-based verification tools such as \seahorn{} and \tricera{} translate input program into a set of Constrained Horn Clauses (CHCs)~\cite{DBLP:conf/birthday/BjornerGMR15}, using uninterpreted predicates to represent program invariants. Our use of uninterpreted predicates in \texttt{assert} and \texttt{assume} statements is already supported by both \seahorn{} and \tricera{}~\cite{DBLP:conf/ecoop/WesleyCNTWG24,tricera}. 
An $\assertcmd{P(\bar{x})}$ statement asserts that $P(\bar{x})$ must hold at that program location, which translates to the predicate appearing in the \emph{head} (the consequent) of a clause. Conversely, an $\assumecmd{P(\bar{x})}$ statement constrains an execution path from a program location, which translates to the predicate~$P(\bar{x})$ appearing in the \emph{body} (the antecedent) of a CHC, in conjunction to the uninterpreted predicate representing that program location. 
This implies the resulting CHCs will be non-linear, i.e., at least two uninterpreted predicates will appear in the body of the CHC where the predicate was assumed, which can make it harder for the CHC solver. Non-linear CHCs can also arise when encoding other program constructs, such as function calls, where the predicates representing function contracts are asserted and assumed at call sites. In our experiments with \seahorn{} and \tricera{} we inline all non-recursive functions.

\subsection{Benchmarks and Experimental Setup}

All experiments for both pipelines were conducted on a Linux machine with an Intel Core i7-7800X CPU @ 3.50GHz and 16 GB of RAM, with a wall-clock timeout of 900 seconds.

\subsubsection{Semi-Automated Pipeline Benchmarks}
The first suite consists of 22 benchmarks (14 safe, 8 unsafe) and was evaluated using the semi-automated pipeline. Among these, six benchmarks are derived from three SV-COMP benchmarks (\texttt{simple\_and\_skiplist\_2lvl-1}, \texttt{simple\_built\_from\_\discretionary{}{}{}end} and \texttt{tree-3}), while the remaining 16 benchmarks are manually crafted, with most of them involving unbounded singly-linked lists. We evaluated the three primary encodings: the base $R$ and~$RW$ encodings given in \Cref{tbl:r-encoding-rules}, and a specialised variant of the $RW$ encoding, $\mathit{RW}$-$\mathit{fun}$, that targets only functional safety properties that is given in \Cref{tbl:extension-encoding-rules}. The $\mathit{None}$ encoding refers to unencoded programs.

\seahorn~(llvm14-nightly, 21-03-2025, using the options ``\texttt{pf --enable-nondet-init --inline}'') and \tricera~(version 0.3.2, using the options ``\texttt{-abstractPO -reachsafety -valid-deref}'') were executed on both encoded and unencoded benchmarks. \cpachecker~(version 4.0) and \predatorhp~(version 3.1415) were evaluated only on the original, unencoded benchmarks since they do not support uninterpreted predicates. The \cpachecker\ and \predatorhp\ tools were configured to only check for explicit assertion failures (i.e., the \texttt{ReachSafety} category of SV-COMP) in both pipelines, using the same options for these tools as in SV-COMP 2025's \texttt{ReachSafety-Heap} category.

\subsubsection{\tricerare{} Pipeline Benchmarks}
The second suite is SV-COMP's \texttt{ReachSafety-Heap} category, which consists of 240 memory-safe benchmarks. Out of the original 240, we excluded those that do not perform heap allocations (e.g., only use stack pointers) or use a fragment of~C unsupported by \tricerare{} (e.g., function pointers), which resulted in 132 benchmarks. We evaluated the unencoded benchmarks using \cpachecker, \predatorhp\ and \tricera\ using the same configuration for the unencoded benchmarks as in the semi-automated pipeline.

Using \tricerare{}, we evaluated a broader range of encodings. These combine the base $R$ and~$RW$ encodings with several extensions, denoted by subscripts: ``$C$'' for \emph{caching} and ``$T$'' for \emph{tagging}. The functional-safety variant is indicated by appending ``$\text{-}\mathit{fun}$''. For example, $\mathit{RWf_{CT}}$ refers to the $RW$ encoding with both caching and tagging extensions applied to the functional-safety variant. The caching extension adds a one-element cache to reduce redundant reads, while the tagging extension incorporates control-flow location identifiers into the time-indexed heap invariants (\Cref{sec:approximations}).

\subsection{Results and Discussion}
\label{sec:results}
We structure the results according to the research questions. Detailed per-benchmark results for the semi-automated and fully-automated pipelines are available in the appendix in \Cref{tbl:per-benchmark-results} and \Cref{tbl:per-benchmark-results-tri}, respectively.

\begin{table}
\caption{Summary of the results for the manually normalised benchmarks. ``Unknown'' represents timeouts, errors and unknown results returned by a tool.}\label{tbl:results-summary}
  \begin{tabular}{llrrrr}
    \toprule
    Tool & Encoding & Safe & Unsafe & Unknown & Total \\ \midrule
    \cpa{} & None & 4 & 8 & 10 & 22 \\
\pred{} & None & 10 & 8 & 4 & 22 \\
\sea{} & None & 3 & 8 & 11 & 22 \\
\tri{} & None & 4 & 8 & 10 & 22 \\
\midrule
\sea{} & $\mathit{R}$ & 7 & 0 & 15 & 22 \\
\tri{} & $\mathit{R}$ & 7 & 8 & 7 & 22 \\
\sea{} & $\mathit{RW}$ & 2 & 0 & 20 & 22 \\
\tri{} & $\mathit{RW}$ & 7 & 8 & 7 & 22 \\
\sea{} & $\mathit{RWf}$ & 12 & 0 & 10 & 22 \\
\tri{} & $\mathit{RWf}$ & 12 & 8 & 2 & 22 \\
    \bottomrule
  \end{tabular}
\end{table}

\begin{table}[h]
  \centering
  \caption{Comparison matrix for the manually normalised benchmarks. Each cell shows the number of (correct) safe/unsafe benchmarks that could be solved by the configuration in the row of that cell, but not by the configuration of its column.}
  \label{tbl:comparison-matrix-old}
%   \begin{tabular}{lcccc|cccccc}
%      & \rotatebox{90}{\cpa{}} & \rotatebox{90}{\pred{}} & \rotatebox{90}{\sea{}} & \rotatebox{90}{\tri{}} & \rotatebox{90}{$\mathit{R}$ \sea{}} & \rotatebox{90}{$\mathit{R}$ \tri{}} & \rotatebox{90}{$\mathit{RW}$ \sea{}} & \rotatebox{90}{$\mathit{RW}$ \tri{}} & \rotatebox{90}{$\mathit{RWf}$ \sea{}} & \rotatebox{90}{$\mathit{RWf}$ \tri{}} \\ \midrule
%     \cpa{} & --- & 0/0 & 1/0 & 1/0 & 3/8 & 1/0 & 3/8 & 1/0 & 0/8 & 0/0 \\
% \pred{} & 6/0 & --- & 7/0 & 6/0 & 5/8 & 5/0 & 8/8 & 5/0 & 0/8 & 0/0 \\
% \sea{} & 0/0 & 0/0 & --- & 0/0 & 2/8 & 0/0 & 2/8 & 0/0 & 0/8 & 0/0 \\
% \tri{} & 1/0 & 0/0 & 1/0 & --- & 3/8 & 0/0 & 3/8 & 0/0 & 0/8 & 0/0 \\
% \midrule
% $\mathit{R}$ \sea{} & 6/0 & 2/0 & 6/0 & 6/0 & --- & 4/0 & 5/0 & 4/0 & 1/0 & 1/0 \\
% $\mathit{R}$ \tri{} & 3/0 & 1/0 & 3/0 & 2/0 & 3/8 & --- & 4/8 & 0/0 & 0/8 & 0/0 \\
% $\mathit{RW}$ \sea{} & 1/0 & 0/0 & 1/0 & 1/0 & 0/0 & 0/0 & --- & 0/0 & 0/0 & 0/0 \\
% $\mathit{RW}$ \tri{} & 4/0 & 2/0 & 4/0 & 3/0 & 4/8 & 1/0 & 5/8 & --- & 1/8 & 0/0 \\
% $\mathit{RWf}$ \sea{} & 8/0 & 2/0 & 9/0 & 8/0 & 6/0 & 6/0 & 10/0 & 6/0 & --- & 1/0 \\
% $\mathit{RWf}$ \tri{} & 8/0 & 2/0 & 9/0 & 8/0 & 6/8 & 6/0 & 10/8 & 5/0 & 1/8 & --- \\
%     \bottomrule
%   \end{tabular}
  \begin{tabular}{lcccc|cccccc}
    % \toprule
     & \rotatebox{90}{\cpa{}} & \rotatebox{90}{\pred{}} & \rotatebox{90}{\sea{}} & \rotatebox{90}{\tri{}} & \rotatebox{90}{$\mathit{R}$ \sea{}} & \rotatebox{90}{$\mathit{R}$ \tri{}} & \rotatebox{90}{$\mathit{RW}$ \sea{}} & \rotatebox{90}{$\mathit{RW}$ \tri{}} & \rotatebox{90}{$\mathit{RWf}$ \sea{}} & \rotatebox{90}{$\mathit{RWf}$ \tri{}} \\ \midrule
    \cpa{} & --- & 0/0 & 1/0 & 1/0 & 3/8 & 1/0 & 3/8 & 1/0 & 0/8 & 0/0 \\
\pred{} & 6/0 & --- & 7/0 & 6/0 & 5/8 & 5/0 & 8/8 & 5/0 & 0/8 & 0/0 \\
\sea{} & 0/0 & 0/0 & --- & 0/0 & 2/8 & 0/0 & 2/8 & 0/0 & 0/8 & 0/0 \\
\tri{} & 1/0 & 0/0 & 1/0 & --- & 3/8 & 0/0 & 3/8 & 0/0 & 0/8 & 0/0 \\
\midrule
$\mathit{R}$ \sea{} & 6/0 & 2/0 & 6/0 & 6/0 & --- & 4/0 & 5/0 & 4/0 & 1/0 & 1/0 \\
$\mathit{R}$ \tri{} & 4/0 & 2/0 & 4/0 & 3/0 & 4/8 & --- & 5/8 & 0/0 & 1/8 & 0/0 \\
$\mathit{RW}$ \sea{} & 1/0 & 0/0 & 1/0 & 1/0 & 0/0 & 0/0 & --- & 0/0 & 0/0 & 0/0 \\
$\mathit{RW}$ \tri{} & 4/0 & 2/0 & 4/0 & 3/0 & 4/8 & 0/0 & 5/8 & --- & 1/8 & 0/0 \\
$\mathit{RWf}$ \sea{} & 8/0 & 2/0 & 9/0 & 8/0 & 6/0 & 6/0 & 10/0 & 6/0 & --- & 1/0 \\
$\mathit{RWf}$ \tri{} & 8/0 & 2/0 & 9/0 & 8/0 & 6/8 & 5/0 & 10/8 & 5/0 & 1/8 & --- \\
    % \bottomrule
  \end{tabular}
\end{table}

\subsubsection{RQ1: Verifying Challenging Heap Programs}
A key question is whether time-indexed heap invariants, combined with off-the-shelf verification tools, can be used to verify heap programs beyond the capabilities of existing tools. The results for the curated benchmarks (\Cref{tbl:results-summary}, \Cref{tbl:comparison-matrix-old}) answer this positively: using a portfolio of our encodings with \tricera{} and \seahorn{}, we successfully verify all 22 crafted benchmarks. These include programs that require identifying complex invariants involving the shape of heap-allocated data structures. In contrast, both \cpachecker{} and the native \tricera{} solve only 12 benchmarks, and \predatorhp{} times out on 4 benchmarks.

\paragraph{\seahorn{} Counterexample Validation.}
We use \seahorn{}'s unbounded verification engine (\texttt{pf}), which may spuriously report that a program is unsafe. We therefore rely on \seahorn{}'s counterexample generation to validate unsafe results, and treat results that cannot be validated as unknown. We use a wrapper script that automatically performs this validation. In the semi-automated pipeline, \seahorn{} is run on both base and encoded benchmarks; counterexample generation works for unencoded benchmarks but currently fails for most encoded benchmarks, leading to 0 unsafe solved instances for \seahorn{} on encoded benchmarks. We have reported this issue to the \seahorn{} developers. When comparing against \tricerare{}, \seahorn{} is run only on the base unencoded benchmarks.

\subsubsection{RQ2: Effect of Encoding Variants}
\Cref{tbl:results-summary} (semi-automated) and \Cref{tbl:results-summary-svcomp} (\tricerare{}) compare the different encoding variants. Among the three encodings evaluated in the semi-automated pipeline ($R$, $RW$, and $\mathit{RWf}$), the $\mathit{RWf}$ encoding achieves the strongest performance with \tricera{}, verifying 20 out of 22 benchmarks compared to 15 for both $R$ and $RW$. The $\mathit{RWf}$ encoding targets memory-safe programs and, as detailed in \Cref{sec:approximations}, omits the initialisation and allocation asserts to $W$, simplifying the invariants required by avoiding reasoning about uninitialised reads.

The base $\mathit{RWf}$ encoding performs significantly worse in the fully-automated \tricerare{} pipeline (39/132) than in the semi-automated pipeline (20/22). This is primarily due to \tricerare{}'s handling of struct updates via ADTs, which introduces read-modify-write patterns that increase heap interactions (see \Cref{sec:eval:auto}). The caching extension addresses this by reducing redundant reads: $\mathit{RWf_{C}}$ solves 68 benchmarks, and with tagging ($\mathit{RWf_{CT}}$), 74 benchmarks. \Cref{fig:cactus-plot} visualises the effect of the different encodings on both solving time and the number of solved instances, and \Cref{tbl:comparison-matrix-svcomp} shows the pairwise comparison between encodings.

\begin{table}[t!]
\caption{Summary of the results for the \tricerare{} pipeline of experiments. ``V. Portfolio'' rows are computed virtually by taking the first result returned by \tricerare{}. ``All - None'' is the result of running all invariant encodings in parallel, and ``All'' in addition runs \tricera{} (the ``None'' encoding).}\label{tbl:results-summary-svcomp}
% old results
%   \begin{tabular}{llrrrrr}
%     \toprule
%     Tool & Encoding & Safe & Unsafe & Incorrect & Unknown & Total \\ \midrule
%     \cpa{} & None & 48 & 37 & 0 & 47 & 132 \\
% \pred{} & None & 71 & 37 & 0 & 24 & 132 \\
% \tri{} & None & 12 & 16 & 0 & 104 & 132 \\\midrule
% \trire{} & $\mathit{R}$ & 16 & 21 & 0 & 95 & 132 \\
% \trire{} & $\mathit{R_{C}}$ & 28 & 26 & 0 & 78 & 132 \\
% \trire{} & $\mathit{R_{T}}$ & 19 & 22 & 0 & 91 & 132 \\
% \trire{} & $\mathit{RW}$ & 6 & 15 & 0 & 111 & 132 \\
% \trire{} & $\mathit{RW_{T}}$ & 7 & 19 & 0 & 106 & 132 \\
% \trire{} & $\mathit{RW_{CT}}$ & 34 & 23 & 0 & 75 & 132 \\
% \trire{} & $\mathit{RW}$ & 13 & 22 & 0 & 97 & 132 \\
% \trire{} & $\mathit{RWf_{C}}$ & 47 & 25 & 0 & 60 & 132 \\
% \trire{} & $\mathit{RWf_{T}}$ & 11 & 20 & 0 & 101 & 132 \\
% \trire{} & $\mathit{RWf_{CT}}$ & 50 & 26 & 0 & 56 & 132 \\\midrule
% V. Portfolio & All - None & 53 & 27 & 0 & 52 & 132 \\
% V. Portfolio & All & 54 & 28 & 0 & 50 & 132 \\
%     \bottomrule
%   \end{tabular}
  \begin{tabular}{llrrrr}
%    \toprule
    Tool & Encoding & Safe & Unsafe & Unknown & Total \\ \midrule
    \cpa{} & None & 48 & 37 & 47 & 132 \\
\pred{} & None & 71 & 37 & 24 & 132 \\
\sea{} & None & 13 & 25 & 94 & 132 \\
\tri{} & None & 13 & 19 & 100 & 132 \\
\midrule
\trire{} & $\mathit{R}$ & 13 & 22 & 97 & 132 \\
\trire{} & $\mathit{R_{C}}$ & 28 & 25 & 79 & 132 \\
\trire{} & $\mathit{R_{T}}$ & 16 & 22 & 94 & 132 \\
\trire{} & $\mathit{RW}$ & 6 & 16 & 110 & 132 \\
\trire{} & $\mathit{RW_{T}}$ & 10 & 20 & 102 & 132 \\
\trire{} & $\mathit{RW_{CT}}$ & 35 & 23 & 74 & 132 \\
\trire{} & $\mathit{RWf}$ & 17 & 22 & 93 & 132 \\
\trire{} & $\mathit{RWf_{C}}$ & 46 & 22 & 64 & 132 \\
\trire{} & $\mathit{RWf_{T}}$ & 11 & 20 & 101 & 132 \\
\trire{} & $\mathit{RWf_{CT}}$ & 50 & 24 & 58 & 132 \\\midrule
V. Portfolio & All - None & 53 & 26 & 53 & 132 \\
V. Portfolio & All & 54 & 27 & 51 & 132 \\
%    \bottomrule
  \end{tabular}

\end{table}

\begin{table}[t!]
  \centering
  \caption{Comparison matrix for the \tricerare{} pipeline of experiments. The first four rows/columns show the results for the named tools, the rest show the results for our encodings implemented in \tricerare{}.}
  \label{tbl:comparison-matrix-svcomp}
  \resizebox{\textwidth}{!}{
  \setlength{\tabcolsep}{2pt}
  \begin{tabular}{lcccc|cccccccccc}
     & \rotatebox{90}{\cpa{}} & \rotatebox{90}{\pred{}} & \rotatebox{90}{\sea{}} & \rotatebox{90}{\tri{}} & \rotatebox{90}{$\mathit{R}$ } & \rotatebox{90}{$\mathit{R_{C}}$ } & \rotatebox{90}{$\mathit{R_{T}}$ } & \rotatebox{90}{$\mathit{RW}$ } & \rotatebox{90}{$\mathit{RW_{T}}$ } & \rotatebox{90}{$\mathit{RW_{CT}}$ } & \rotatebox{90}{$\mathit{RWf}$ } & \rotatebox{90}{$\mathit{RWf_{C}}$ } & \rotatebox{90}{$\mathit{RWf_{T}}$ } & \rotatebox{90}{$\mathit{RWf_{CT}}$ } \\ \midrule
    \cpa{} & --- & 1/1 & 36/12 & 35/18 & 35/15 & 20/12 & 32/15 & 42/21 & 40/17 & 19/14 & 31/15 & 4/15 & 39/17 & 6/13 \\
\pred{} & 24/1 & --- & 59/12 & 58/19 & 59/15 & 44/13 & 56/15 & 66/21 & 63/17 & 40/14 & 55/15 & 28/15 & 62/17 & 27/13 \\
\sea{} & 1/0 & 1/0 & --- & 5/10 & 8/6 & 3/6 & 6/6 & 8/10 & 8/7 & 2/6 & 5/6 & 2/7 & 8/7 & 1/5 \\
\tri{} & 0/0 & 0/1 & 5/4 & --- & 7/3 & 3/1 & 6/3 & 9/7 & 9/3 & 3/2 & 5/2 & 1/3 & 9/3 & 1/2 \\
\midrule
$\mathit{R}$ & 0/0 & 1/0 & 8/3 & 7/6 & --- & 0/0 & 1/0 & 7/6 & 5/2 & 0/0 & 1/1 & 0/1 & 5/2 & 0/0 \\
$\mathit{R_{C}}$  & 0/0 & 1/1 & 18/6 & 18/7 & 15/3 & --- & 12/3 & 22/9 & 20/5 & 2/2 & 12/3 & 0/3 & 20/5 & 1/2 \\
$\mathit{R_{T}}$  & 0/0 & 1/0 & 9/3 & 9/6 & 4/0 & 0/0 & --- & 10/6 & 8/2 & 0/0 & 3/1 & 0/1 & 8/2 & 0/0 \\
$\mathit{RW}$  & 0/0 & 1/0 & 1/1 & 2/4 & 0/0 & 0/0 & 0/0 & --- & 0/0 & 0/0 & 0/0 & 0/1 & 0/0 & 0/0 \\
$\mathit{RW_{T}}$  & 2/0 & 2/0 & 5/2 & 6/4 & 2/0 & 2/0 & 2/0 & 4/4 & --- & 0/0 & 2/0 & 2/1 & 1/0 & 0/0 \\
$\mathit{RW_{CT}}$  & 6/0 & 4/0 & 24/4 & 25/6 & 22/1 & 9/0 & 19/1 & 29/7 & 25/3 & --- & 18/1 & 6/1 & 25/3 & 0/0 \\
$\mathit{RWf}$  & 0/0 & 1/0 & 9/3 & 9/5 & 5/1 & 1/0 & 4/1 & 11/6 & 9/2 & 0/0 & --- & 1/1 & 9/2 & 0/0 \\
$\mathit{RWf_{C}}$  & 2/0 & 3/0 & 35/4 & 34/6 & 33/1 & 18/0 & 30/1 & 40/7 & 38/3 & 17/0 & 30/1 & --- & 37/3 & 3/0 \\
$\mathit{RWf_{T}}$  & 2/0 & 2/0 & 6/2 & 7/4 & 3/0 & 3/0 & 3/0 & 5/4 & 2/0 & 1/0 & 3/0 & 2/1 & --- & 0/0 \\
$\mathit{RWf_{CT}}$  & 8/0 & 6/0 & 38/4 & 38/7 & 37/2 & 23/1 & 34/2 & 44/8 & 40/4 & 15/1 & 33/2 & 7/2 & 39/4 & --- \\
   % \bottomrule
  \end{tabular}
  }
\end{table}

\begin{figure}[t]
  \centering
  \includegraphics[width=.6\textwidth]{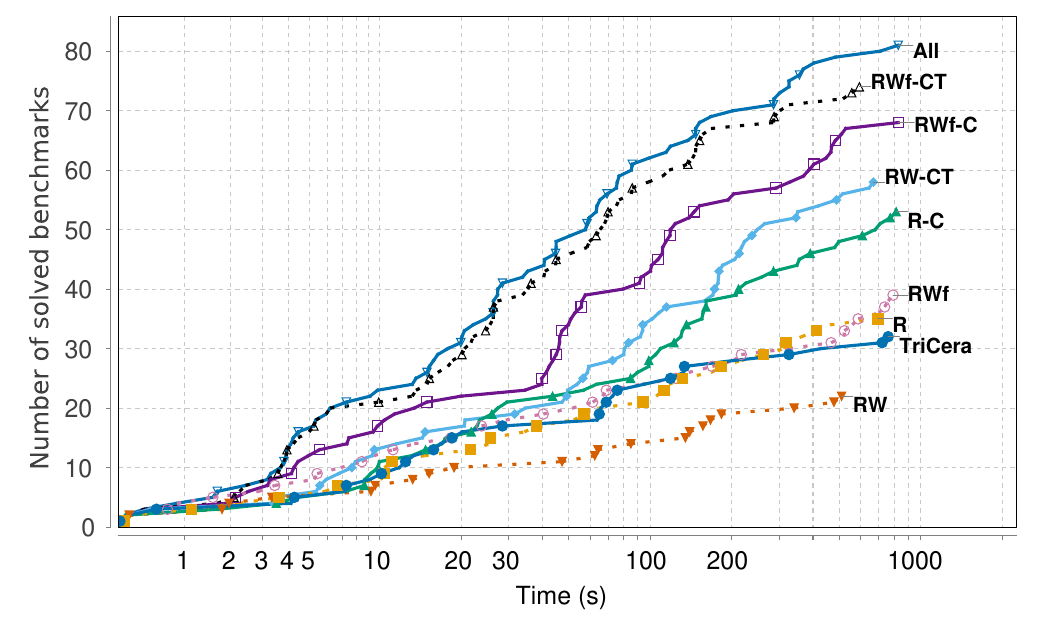}
  \Description{Plot comparing the effect of different encoding variants on the number of benchmarks solved within a given timeout. The caching extension provides the largest improvement.}
  \caption{The plot shows base \tricera{} and a selection of encodings using \tricerare{} to illustrate the impact of encoding variations. For each timeout (x-axis), the y-axis shows how many benchmarks could be solved. "All" is the virtual portfolio of all included encodings and base \tricera{}.}
  \label{fig:cactus-plot}
\end{figure}

\subsubsection{RQ3: Comparison with State-of-the-Art}
\Cref{tbl:results-summary-svcomp} and \Cref{tbl:comparison-matrix-svcomp} present the SV-COMP results. \tricerare{} with the $\mathit{RWf_{CT}}$ encoding verifies 5 safe benchmarks that none of the other tools (\cpachecker{}, \predatorhp{}, \seahorn, \tricera{}) can handle. More specifically, it solves 8 benchmarks that \cpachecker{} cannot verify, 6 that \predatorhp{} cannot verify, 42 that \seahorn{} cannot verify, and 45 that the native \tricera{} cannot verify\footnote{Although not included in our experiments, we note that some of these benchmarks were also solved by other tools participating in SV-COMP, in particular \textsc{2LS}~\cite{DBLP:conf/tacas/MalikMSSVW18}.}. 

In terms of solved instances, the best \tricerare{} configuration ($\mathit{RWf_{CT}}$) solves 50 safe programs, narrowly beating \cpachecker{}'s 48, but behind \predatorhp{}'s 71. For unsafe programs, \tricerare{} (24) trails both \cpachecker{} and \predatorhp{} (37 each). However, a virtual portfolio combining all \tricerare{} encodings with base \tricera{} solves 81 benchmarks (54 safe, 27 unsafe), approaching \cpachecker{}'s 85 (48 safe, 37 unsafe).

%\subsubsection{Discussion}
Our prototype \tricerare{} is at an early stage of development, and several directions could extend its reach. Static analyses could minimise the number of heap interactions, which we observed to have a significant effect when done manually in our crafted benchmarks. The caching extension builds a similar optimisation directly into the encoding (though at a cost in verification complexity), more than doubling the number of solved instances.

\tricerare{} uses mathematical arrays as sources of non-deterministic values retrieved multiple times, e.g., use of $\dethavoccall{x}$ within loops. While this adds array arguments to the time-indexed heap invariants~($R$, $W$), the experiments show that the required formulas about arrays are relatively easy to find for a CHC solver. We believe that this is because the arrays are only read, never modified, and the precise values present in the arrays are therefore immaterial for time-indexed heap invariants.

%Furthermore, the results support our claims of correctness. Although \seahorn\ reported incorrect (unsafe) answers, we anticipate this is unrelated to our encoding, because it also reported incorrect answers for the unencoded programs. Using the same encodings \tricera\ did not produce any incorrect results.

% \predatorhp\ exhibited strong performance on the unencoded benchmarks, consistent with their proven success on similar benchmarks as evidenced by previous SV-COMP results~\cite{svcomp24}. However, our encodings using the backends \seahorn\ and \tricera\ could prove the safety of more benchmarks that \predatorhp\ timed out on, showcasing the potential of our approach.

%% file: content/related-work.tex
% related-work.tex

Verification of heap-manipulating programs has inspired a variety of techniques that trade off precision, automation, and annotation effort. We discuss the most closely related work.

\paragraph{CHC-Based Approaches}
Many CHC-based approaches encode heaps using high-level theories such as the theory of arrays~\cite{DBLP:conf/ifip/McCarthy62} and less-commonly the theory of heaps~\cite{esenTheoryHeapConstrained2021}. The backends we used in our implementation, \seahorn\ and \tricera, are examples of tools using these approaches. Solving arrays typically require finding quantified invariants.
Bj{\o}rner et al. present a method for finding quantified invariants over arrays by guiding Horn solvers through constraints put on the form of the proof~\cite{DBLP:conf/sas/BjornerMR13}. In contrast, our approach encodes the heap using only integers.
Monniaux et al. present a method for transforming programs with arrays into nonlinear Horn clauses over scalar variables~\cite{DBLP:conf/sas/MonniauxG16}. While this approach is sound, it is incomplete.

The closest study (that also inspired our work) is by Kahsai et al. in the context of the verification tool \jayhorn~\cite{DBLP:conf/cav/KahsaiRSS16,DBLP:conf/lpar/KahsaiKRS17}, where uninterpreted predicates are used as \emph{space invariants}, representing the invariants of objects on heap. The space invariants approach itself was inspired by techniques based on refinement types and liquid types~\cite{refinementtypes,liquidtypes}. The space invariants encoding is incomplete, but the authors provide several refinements to improve the precision of the analysis, similar to the extensions we provide in \Cref{sec:approximations}. The space invariants approach can be seen as an overapproximation of our approach.

Faella and Parlato present a CHC-based verification approach for tree-manipulating programs~\cite{DBLP:conf/cav/FaellaP25}: the \emph{knitted-tree} approach that models program executions as tree data structures, followed by encoding the trees as CHCs. Our encoding of heap is fundamentally different from knitted trees; while we provide a general encoding of heap operations, agnostic of data structures, knitted trees are a verification approach specific to trees in which the execution of a program is interpreted as a traversal of the tree, associating intermediate program states with the tree nodes. In a sense, \cite{DBLP:conf/cav/FaellaP25} and our approach have dual characteristics: where \cite{DBLP:conf/cav/FaellaP25} inlines program executions in tree data structures, our approach represents data structures sequentially as part of program executions.

In the context of modular verification through CHCs, Garcia-Contreras et.\ al~\cite{DBLP:conf/sas/Garcia-Contreras22} use a static analysis to infer the size of heap used by a function, then represent this bounded memory using a theory of \emph{finite maps}. 
In~\cite{DBLP:conf/vmcai/SuNGG25} the authors use abstract interpretation~\cite{DBLP:conf/popl/CousotC77,DBLP:conf/birthday/Cousot03} and introduce a composite abstract domain to infer relational object invariants in the presence of field updates where such invariants temporarily do not hold. While such static analysis-based methods can be powerful, they are inherently incomplete. Our experiments show that even without any additional static analyses, our approach can already solve many challenging verification tasks, and could serve as a foundation in such approaches. For instance in~\cite{DBLP:conf/sas/Garcia-Contreras22} when the static analysis fails to infer the size of the heap, it will fall back to using unbounded arrays, for which our general heap encoding provides a quantifier-free alternative.

% A key ingredient in our approach is the use of uninterpreted predicates, supported by CHC . ... SeaHorn~\cite{seahorn} has been extended to allow uninterpreted predicates, called Inductive Predicate Synthesis Modulo Programs~\cite{DBLP:conf/ecoop/WesleyCNTWG24}.
% JayHorn is a model checker for Java that translates programs into CHCs and handles the heap by \ldots

% Mature CHC solvers exist~\cite{eldarica, golem, spacer} with support for high-level theories such as arrays~\cite{DBLP:conf/ifip/McCarthy62}, algebraic data types and even heaps~\cite{esenTheoryHeapConstrained2021}. The tools \seahorn~\cite{seahorn} and \tricera~\cite{tricera} use these theories to implement heaps. Other CHC-based tools such as RustHorn\ze{fix tool names} eliminate the heap from the clauses by using prophecy variables ... \ze{other tools, theta?}

% Although we defined our encodings at program level, the encoding could also be applied as a preprocessing stage at solver level, translating CHCs over heaps into CHCs over integers, benefiting all front-end tools using the theory of heaps.
% \ze{cleanup needed and more related work, add jayhorn, fsen paper?}

\paragraph{Methods Based on Separation Logic and Shape Analysis}
An alternative foundation for heap verification is separation logic~\cite{DBLP:conf/lics/Reynolds02, DBLP:conf/csl/OHearnRY01}, which extends Hoare logic with pointers and functions to enable local reasoning about disjoint heaps. 
A long line of work on separation logic has led to powerful verification tools and approaches such as \viper~\cite{viper} and the flow framework that encodes global heap properties using local flow equations~\cite{DBLP:journals/pacmpl/KrishnaSW18,DBLP:conf/esop/KrishnaSW20,DBLP:conf/tacas/MeyerWW23}. 
Separation logic tools often require user-provided annotations to handle complex structures. In contrast, our approach is fully automatic and avoids the need for manual annotations.

Another major approach is abstract interpretation of heap shapes, as exemplified by shape analysis~\cite{DBLP:conf/popl/JonesM79, DBLP:journals/ftpl/ChangDMRR20} algorithms. Shape analyses automatically infer an over-approximation of all possible heap configurations a program can create (often via graph-based abstractions of memory). Modern shape analysers like \predator~\cite{predator, predatorhp} and \forester~\cite{forester} build on this idea, using automata or graph rewriting to represent unbounded linked lists and trees. These methods are fully automatic and can verify memory safety and some structural invariants without user input. However, the abstractions may lead to false alarms. For instance \predatorhp\ tries to work around this by running an additional instance of \predator\ without heap abstractions. This limitation arises because shape domains often focus on heap topology and handle data content in a coarse way. For instance \predator\ uses an abstract domain for lists, and cannot handle programs with trees. It also has limited support for non-pointer data. Our approach is not specialised to any particular domain.

%% file: content/conclusion.tex
% conclusion.tex

We have introduced a sound and complete time-indexed heap invariants approach for verifying safety properties of programs operating on mutable, heap-allocated data structures. Our encoding replaces heap accesses with access-local invariants that are often easier for automated tools to infer than global heap invariants, providing a complementary approach to existing heap verification techniques. Our formal proofs establish the conditions under which such encodings remain sound and complete, providing a solid theoretical foundation for future research. Through our experimental evaluation, we showed that our approach can verify programs beyond the reach of current state-of-the-art tools. The generality of our framework opens avenues for further investigation into different encodings, abstractions, and optimisations. We also plan to investigate extensions to arrays, pointer arithmetic, and concurrency; we believe these extensions are straightforward in theory, but further research is needed to make them practical.

%% file: content/appendix.tex
\subsection[Correctness of EncR]{Correctness of $\encR$}\label{app:proof-r}
We show the correctness of $\encR$ by showing that given \alang\ program $p$, $p$ is \emph{equi-safe} with $\encR(p)$. The core observation used in the proof is that the relation~$R$, obtained as the least fixed point of the immediate consequence operator, correctly represents the values read from the heap.
We prove this result in multiple steps; the first step is to show that the relation $R$ is a \emph{partial function} that maps the program inputs~$\inG$ and the read count~$\cnt$ to the value that is read:
\begin{lemma}[Functional consistency of $\interp^*(R)$]
  \label{prop:R-partial-app}
  Let $\interp^*$ be the least fixed point of the immediate consequence operator $\ico_p$ for \alang\ program $p=\encR(q)$ obtained as the $R$-encoding of some program~$q$. Then $\interp^*(R)$ is a partial function from its first two arguments to its third argument:
  \begin{equation}
    \label{eq:functional-app}
    \forall g, n, v_1, v_2.\ \Bigl( (g, n, v_1) \in \interp^*(R) \wedge (g, n, v_2) \in \interp^*(R) \Bigr)
      \Longrightarrow v_1 = v_2.
  \end{equation}
\end{lemma}

The lemma follows from the shape of the code snippet introduced by $\encR$ for \readkw{} statements and can be proved by induction on the iteration count~$\alpha$ in approximations $\interp_\alpha$ of the least fixed point~$\interp^*$. For sake of brevity, we only give a proof sketch:
\begin{proof}
  Consider an approximation~$\interp_n = \ico_p^n(\interp_0)$ for $n \in \mathbb{N}$. Due to the \assertkw{} in the code snippet for \readkw{}, the relation~$\interp_n(R)$ will contain the values read during the first~$n$ \readkw{} statements encountered during any execution of $p$; the $(n+1)$-th \readkw{} of an execution will block either due to a failing \assertkw{} in the then-branch of the encoding of \readkw{}, or due to a blocking \assumekw{} in the else-branch. Moving from $\interp_n$ to the next approximation~$\interp_{n+1} = \ico_p(\interp_{n})$, the values that can be read during the $(n+1)$-th \readkw{} will be added to the interpretation of $R$. If $\interp_n(R)$ is a partial function, then also $\interp_{n+1}(R)$ is a partial function, because the tuple added for the $(n+1)$-th \readkw{} cannot contradict any of the tuples already present in $\interp_n(R)$ (since $\cnt$ is strictly increasing), and because at most one tuple can be added to $\interp_{n+1}(R)$ for every value of $\inG$.
  
  More formally, this can be shown inductively by proving that for every ordinal~$\alpha$, the approximation $\interp_\alpha$ has the following properties:
  \begin{itemize}
    \item $\interp_\alpha(R)$ is a partial function in the sense of \eqref{eq:functional-app}.
    \item For every tuple~$(\_, n, \_) \in \interp_\alpha(R)$ it is the case that $n \leq \alpha$.
    \item In every final state $\bigstepevalp{\interp_\alpha}(p, s, \eh) = (\_, s', h)$ of an execution of $p$ for $\interp_\alpha$, it is the case that $s(\pg) = s'(\pg)$ and that the value~$s'(\lastn)$ is uniquely determined by the input $s(\inG)$ and $s(\pg)$.
    \item Whenever an execution of $p$ for $\interp_\alpha$ fails with the result~$\bigstepevalp{\interp_\alpha}(p, s, \eh) = (\error(R, \bar{v}), s', h)$, the final values~$s'(\lastn)$ and $s'(\pg)$ are uniquely determined by the input~$s(\inG)$, and it is the case that $s'(\cnt) = \alpha + 1$.
  \end{itemize}
  The property to be shown follows since $\interp^*$ is the limit of the approximations~$\interp_\alpha$.
\end{proof}

As the next step for proving the correctness of $\encR$, we introduce an intermediate encoding $\encN$, and show that $p$ is equi-safe with $\encN(p)$. The purpose of this encoding is to introduce a counter that is incremented by each encoded statement, which we will use in the inductive proof of correctness for $\encR$.

\paragraph{The Encoding $\encN$.}
Given \alang\ program $p$, $\encN(p)$ is obtained by introducing an $\sortint$ variable $\encNc$,
and before every statement $S$ over one of $\{\writekw, \allockw, \readkw\}$ in $p$, inserting
  \[\assigncmd{\encNc}{\encNc - 1};\ \assumecmd{\encNc \ge 0}.\]
Starting from the same stack $s$ with $n = s(c)$, the executions of $p$ and $\encN(p)$ will remain identical up until the $(n+1)$-th evaluation of any statement $S$ over one of $\{\writekw, \allockw, \readkw\}$ (apart from the value of $\encNc$ in $\encN(p)$), after which the evaluation of $\assumekw$ that was inserted right before $S$ will fail and program execution is stopped.
%The result of the execution, as determined by $\bigstepeval$, will be of the form ...

\begin{lemma}[$\encN$ is correct]\label{lem:encN-proof-app}
  Let $p$ be \alang\ program, and $\interp$ some interpretation, then $p$ and $\encN(p)$ are equi-safe, i.e.,
  % \[ \left( \exists s \in \stacks. \tupleget{\bigstepevalp{\interp}(p, s, \eh)}{1} = \error(P_1, \bar{v}_1) \right) \iff \left( \exists s \in \stacks. \tupleget{\bigstepevalp{\interp}(\encN(p), s, \eh)}{1} = \error(P_2, \bar{v}_2) \right).\]
  $p$ is safe if and only if $\encN(p)$ is safe.
\end{lemma}

\begin{proof}
  We show that $p$ is unsafe if and only if $\encN(p)$ is unsafe, which is equivalent to the definition of equi-safety.

  ($\Rightarrow$)
  Assume $p$ is unsafe, i.e., there is some stack $s$ such that $\tupleget{\bigstepevalp{\interp}(p, s, \eh)}{1} = \error(P_1, \bar{v}_1)$. Let $n$ denote the number of evaluations of any statement $S$ over one of $\{\writekw, \allockw, \readkw\}$ before the failing assertion. In $\encN(p)$, the execution with the initial configuration $(s', \eh)$, where $s'$ is the same as $s$ except that $s'(\encNc) = n$, will also fail the same assertion, because the $\assumekw$ statement in $\cmd{\assigncmd{\encNc}{\encNc - 1};\ \assumecmd{\encNc \ge 0}}$ will always result in $\success$ in the first $n$ evaluations.

  \noindent
  ($\Leftarrow$)
  Assume that $\encN(p)$ is unsafe for some stack $s$. $\encN(p)$ is identical to $p$ except for $\encNc$ and the $\assumekw$ statements added by $\encNc$. An $\assumekw$ statement does not lead to an assertion failure (i.e., $\error(P, \bar{v})$ for some $P$ and $\bar{v}$). Therefore, an execution of $p$ using the stack $s$ will fail the same assertion that failed in $\encN(p)$.
\end{proof}

We now state the last lemma needed to show the correctness of $\encR$.

\begin{lemma}[Preservation of final states by $\encR$]\label{lem:encR-preservation-app}
  Let $p$ be \alang\ program, and $\interp^*$ be the least fixed point of $\ico_{p}$. Let $\pstar = \encN(p)$ and $\pstarstar = \encR(\pstar)$,
 %
  % Let $s$ be an arbitrary stack, that also assigns values to the auxiliary variables introduced by $\encN$ and $\encR$.
  $\bigstepevalp{\interp^*}(\pstar, s, \eh) = (\sigma_1, s_1, h_1)$, and
  $\bigstepevalp{\interp^*}(\pstarstar, s, \eh) = (\sigma_2, s_2, h_2)$.
  Then the following holds:
  %for all $n,a\in\mathbb{N}$, if $n = s(\encNc)$ and $a = s(\pg)$, the following holds:
\begin{equation}
\begin{array}{c@{\quad}l}
\vcenter{\hbox{$\forall s \in \stacks,\, a \in \integer.\; n = s(\encNc) \wedge a = s(\pg) \rightarrow$}} &
\begin{aligned}
\overbrace{\sigma_1 = \sigma_2}^{\foutcomes} \wedge
\overbrace{\forall v \in \progvars(\pstar).\; s_1(v) = s_2(v)}^{\fstacks} \wedge\\
\underbrace{\rd(h_1, a) = s_2(\lastn)}_{\freads} \wedge
\underbrace{|h_1| = s_2(\allocctr)}_{\fallocs}
\end{aligned}
\end{array}.
\label{eq:r-preservation-app}
\end{equation}

  % \begin{equation}
  %   \overbrace{\sigma_1 = \error(F, ()) \leftrightarrow \sigma_2 = \error(F, ())}^{\foutcomes} \wedge
  %   \overbrace{\forall v \in \progvars(\pstar).\ s_1(v) = s_2(v)}^{\fstacks} \wedge
  %   \overbrace{\rd(h_1, a) = s_2(\lastn)}^{\freads} \wedge
  %   \overbrace{|h_1| = s_2(\allocctr)}^{\fallocs}\label{eq:r-preservation-app}
  % \end{equation}
  % \begin{enumerate}
  %   \item[$\foutcomes$:] $\sigma_1 = \error(F,()) \leftrightarrow \sigma_2 = \error(F,())$
  %   \item[$\fstacks$:]   $\forall v \in \progvars(\pstar).\ s_1(v) = s_2(v)$
  %   \item[$\freads$:]    $\rd(h_1, a) = s_2(\lastn)$
  %   \item[$\fallocs$:]   $|h_1| = s_2(\allocctr)$
  % \end{enumerate}
  % \begin{align}
  %   \begin{split}
  %   \forall n, a \in \mathbb{N}.\ n = s(\encNc) \wedge a = s(\pg) \rightarrow &
  %     \overbrace{\sigma_1 = \error(F, ()) \leftrightarrow \sigma_2 = \error(F, ())}^{\foutcomes} \wedge
  %     \overbrace{\forall v \in \progvars(\pstar).\ s_1(v) = s_2(v)}^{\fstacks} \wedge \\&
  %     \underbrace{\rd(h_1, a) = s_2(\lastn)}_{\freads} \wedge
  %     \underbrace{|h_1| = s_2(\allocctr)}_{\fallocs}\label{eq:r-preservation-app}
  %   \end{split}
  % \end{align}
\end{lemma}

% Since $\pstar$ is deterministic, it has a unique execution for each initial configuration $(s,\eh)$. $\pstarstar$ introduces an $\havocexpr$ expression for each rewritten $\readkw$ statement, and is therefore nondeterministic if $\pstar$ contains at least one $\readkw$. 
The lemma states that, for the same initial configuration $(s,\eh)$, both $\pstar$ and $\pstarstar$ will result in the same outcome ($\foutcomes$), with the same values for all common variables in the final stacks ($\fstacks$), with the variable $\lastn$ holding the same object that is stored at $\pg$ in the final heap $h_1$ of $\pstar$ ($\freads$), and with the value of $\allocctr$ in $\pstarstar$ matching the size of the heap $h_1$ in $\pstar$ ($\fallocs$).
We will use this lemma in \Cref{thm:encR-proof-app} to show that $\pstar$ and $\pstarstar$ are equi-safe (\Cref{def:equisafety}), which is a weaker claim.

\begin{proof}
We show that each rewrite preserves~\eqref{eq:r-preservation-app} by induction on $n$.

\paragraph{\textbf{Base Case ($n = 0$):}}
When $n = 0$, no $\readkw$, $\writekw$ or $\allockw$ statements are executed due to the failing $\assumekw$ statement added by $\encN$. Therefore, for the given initial configuration, $\pstarstar$ has a single execution that is identical to the execution of $\pstar$, except for the auxiliary variable initialisation in $\pstarstar$. Both $\foutcomes$ and $\fstacks$ hold, because statements that might affect the outcome ($\assertkw$ and $\assumekw$) are only over expressions that use the common variables of $\pstar$ and $\pstarstar$, and those variables have the same values in both.  

We have $s_2(\lastn) = \defObj$, since $\lastn$ is initialised with $\defObj$ in $\pstarstar$ and has the same value in all executions with $n = 0$ since no $\writekw$s occur. We also have $h_1 = \eh$, and $\rd(h_1, a) = \defObj$ by heap theory semantics, and $\freads$ follows.

Finally, from the heap theory semantics we have $|h_1| = |\eh| = 0$, and we have $s_2(\allocctr) = 0$ due to the initialisation of $\allocctr$ in $\pstarstar$, and $\fallocs$ follows.

\paragraph{\textbf{Successor Case ($n = k + 1$):}}
% Assume that~\eqref{eq:r-preservation-app} holds for some $n = k$ (i.e., $\foutcomes \wedge \fstacks \wedge \freads \wedge \fallocs$ holds if $k = s(\encNc) \wedge a = s(\pg)$ for some $a$), we show that it also holds also for $n = k+1$ (i.e., $\foutcomes' \wedge \fstacks' \wedge \freads' \wedge \fallocs'$ holds if $k+1 = s(\encNc) \wedge a = s(\pg)$ for some $a$).

Assume the induction hypothesis \eqref{eq:r-preservation-app} holds for $n = k$, for some $k \in \mathbb{N}$. 
That is, for the execution of $\pstar$ starting with some initial configuration $(s, \eh)$ and resulting in $(\sigma_1, s_1, h_1)$, there exists some execution of $\pstarstar$ that results in $(\sigma_2, s_2, h_2)$ such that~\eqref{eq:r-preservation-app} holds for $n = k$.
We show that the following (for $n = k + 1$) also holds:
\begin{equation}
\begin{array}{c@{\quad}l}
\vcenter{\hbox{$\forall s \in \stacks, a \in \integer.\; k+1 = s(\encNc) \wedge a = s(\pg) \rightarrow$}} &
\begin{aligned}
\overbrace{\sigma_1' = \sigma_2'}^{\foutcomes'} \wedge
\overbrace{\forall v \in \progvars(\pstar).\; s_1'(v) = s_2'(v)}^{\fstacks'} \wedge\\
\underbrace{\rd(h_1', a) = s_2'(\lastn)}_{\freads'} \wedge
\underbrace{|h_1'| = s_2'(\allocctr)}_{\fallocs'}
\end{aligned}
\end{array}
\label{eq:r-step-app}
\end{equation}

Recall that, due to $\encN$, incrementing $k$ has the effect of evaluating at most one additional rewritten statement $S$ (followed by any number of non-rewritten statements). If $\sigma_1 \neq \success$ due to an earlier evaluation than the last evaluated $S$, the result of that evaluation will be propagated by the semantics of \lang, trivially establishing~\eqref{eq:r-step-app} by~\eqref{eq:r-preservation-app}. Therefore we focus on the case where $\sigma_1 \sigma_2 = \success$, with $(s_1, h_1)$ as the initial state for evaluating $S$, and $(\delta_1', s_1', h_1')$ as the result of that evaluation. In the successor case we show that there exists an execution of $\pstarstar$ that starts with the initial configuration $(s_2, h_2)$ satisfying~\eqref{eq:r-preservation-app}, and results in $(\sigma_2', s_2', h_2')$ satisfying \eqref{eq:r-step-app}. Any statements that are not rewritten will trivially satisfy~\eqref{eq:r-step-app} by~\eqref{eq:r-preservation-app}, so we only focus on rewritten statements $S$ over one of $\{\writekw, \allockw, \readkw\}$ in $\pstar$.

\paragraph{\textbf{Case } $S \equiv \writecmd{x}{e}$:}
The statement $\encR(S)$ is\\
\hspace*{4ex}$\cmd{\ifcmd{\pg = x \wedge 0 < x \leq \allocctr}{\assigncmd{\lastn}{e}}{\skipcmd}}$. 
%There is a single execution of $\pstarstar$ for the initial configuration $(s_2, h_2)$ (there are no $\havocexpr$ expressions in $\encR(S)$), and we show that in that execution~\eqref{eq:r-step-app} holds.
\begin{itemize}
\item $\fstacks'$: The only variable modified by $\encR(S)$ is $\lastn$, which is not in $\progvars(\pstar)$. Therefore $\fstacks'$ holds by~\eqref{eq:r-preservation-app}.
\item $\foutcomes'$: $\encR(S)$ does not contain assertions or assumptions that can change the outcome. Since the remaining executions after $S$ and $\encR(S)$ are over identical statements, we have $\foutcomes' \leftrightarrow \foutcomes$, so $\foutcomes'$ holds.
\item $\fallocs'$: $\wt$ does not modify the size of the heap, so we have $|h_1| = |h_1'|$. The variable $\allocctr$ is also not modified in $\encR(S)$, so we have $s_2(\allocctr) = s_2'(\allocctr)$. By~\eqref{eq:r-preservation-app} we have $|h_1| = s_2(\allocctr)$, thus $\fallocs'$ also holds. 
\item $\freads'$: We need to show $\rd(h_1', a) = s_2'(\lastn)$.
We have $s_2(\allocctr) = s_2'(\allocctr) = |h_1| = |h_1'|$ by $\fallocs'$. Let $o = \semtwo{e}{s_1'} = \semtwo{e}{s_2'}$ (by $\fstacks'$).
  \begin{itemize} 
  \item if $s_2(\pg) = s_2(x) \wedge 0 < s_2(x) \leq s_2(\allocctr)$, the \emph{then} branch of $\encR(S)$ is taken. Substituting $s_2(x)$ for $a$ and $s_2(\allocctr)$ for $|h_1|$, by the theory of heaps semantics, after evaluating $S$ in $\pstar$, for $0 < a \leq |h_1|$, we have $h_1'[a-1] = o$ and $\rd(h_1', a) = o$. In $\encR(S)$, we also assign $e$ to $\lastn$ in the \emph{then} branch, therefore after evaluation we have $s_2'(\lastn) = o$ and $\freads'$ holds.
  \item otherwise, the \emph{else} branch of $\encR(S)$ is taken. In this case, $\encR(S)$ reduces to $\skipcmd$, thus $s_2'(\lastn) = s_2(\lastn)$, and we need to show $\rd(h_1', a) = s_2(\lastn)$
  The condition for the \emph{else} branch is $s_2(\pg) \neq s_2(x) \vee \neg(0 < s_2(x) \leq s_2(\allocctr))$.
  \begin{itemize}
      \item Case $s_2(\pg) \neq s_2(x)$: In $\pstar$, $S = \writecmd{x}{e}$ results in $h_1' = \wt(h_1, s_1(x), o)$. Since $a = s(\pg) = s_2(\pg) \neq s_2(x) = s_1(x)$, the address $a$ is different from the address being written to. Therefore, $\rd(h_1', a) = \rd(\wt(h_1, s_1(x), o), a) = \rd(h_1, a)$ by the theory of heaps semantics.
      \item Case $\neg(0 < s_2(x) \leq s_2(\allocctr))$: This means the write to address $s_2(x)$ in $\pstar$ is invalid. In this case, by the semantics of $\wt$, $h_1' = \wt(h_1, s_1(x), o) = h_1$. Thus, $\rd(h_1', a) = \rd(h_1, a)$.
  \end{itemize}
  In both cases, $\rd(h_1', a) = \rd(h_1, a)$. By the induction hypothesis $\freads$, we have $\rd(h_1, a) = s_2(\lastn)$. Therefore, $\rd(h_1', a) = \rd(h_1, a) = s_2(\lastn) = s_2'(\lastn)$, which establishes $\freads'$.
  \end{itemize}
\end{itemize}

\paragraph{\textbf{Case } $S \equiv \alloccmd{x}{e}$:} Let $o = \semtwo{e}{s_1} = \semtwo{e}{s_2}$ (by $\fstacks$). 
%There is a single execution of $\pstarstar$ for the initial configuration $(s_2, h_2)$ (there are no $\havocexpr$ expressions in $\encR(S)$), and we show that in that execution~\eqref{eq:r-step-app} holds.
\begin{itemize}
  \item $\fallocs'$: We need to show $|h_1'| = s_2'(\allocctr)$.
  In $\pstar$, we have $|h_1'| = |h_1| + 1$ by the semantics of $\alloc$.
  In $\pstarstar$, we also have $s_2'(\allocctr) = s_2(\allocctr) + 1$ due to the statement $\assigncmd{\allocctr}{\allocctr+1}$. By the induction hypothesis $\fallocs$, we have $|h_1| = s_2(\allocctr)$. Therefore, $|h_1'| = |h_1| + 1 = s_2(\allocctr) + 1 = s_2'(\allocctr)$. Hence, $\fallocs'$ holds.
  \item $\fstacks'$: The only modified variable common to $\pstar$ and $\pstarstar$ is $x$, for unmodified common variables $\fstacks'$ is established by the hypothesis~\eqref{eq:r-preservation-app}.
  In $\pstar$, $s_1'(x) = |h_1'|$ by the semantics $\alloc$. 
  In $\pstarstar$, $s_2'(\allocctr) = s_2'(x)$ due to the statement $\assigncmd{x}{\allocctr}$ in $\encR(S)$. Using $\fallocs'$ ($|h_1'| = s_2'(\allocctr)$) we have $s_1'(x) = s_2'(x)$, which establishes $\fstacks'$ also for $x$.
  \item $\foutcomes'$: follows the same reasoning as the proof of $\foutcomes'$ for $\writekw$.
  \item $\freads'$: We need to show $\rd(h_1', a) = s_2'(\lastn)$.
  In $\pstar$, we have $o = \rd(h_1', a)$ due to the semantics of $\alloc$.
  In $\pstarstar$, let $s_2^*$ be the stack after the assignment to $x$, but before the $\ifkw$ command. The proof follows the same reasoning as the proof of $\freads'$ for $\writekw$, except here $s_2^*(x) = s_2^*(\allocctr)$, so we only consider the two cases where $s_2^*(\pg) = s_2^*(x)$ and $s_2^*(\pg) \neq s_2^*(x)$.
  If $s_2^*(\pg) = s_2^*(x)$, the then branch is taken and $\freads'$ holds as $s_2'(\lastn) = o$.
  Otherwise, $s_2^*(\pg) \neq s_2^*(x)$, and the address $a$ is different from the allocated address. Therefore, $\rd(h_1', a) = \rd(h_1, a) = o$ by the theory of heaps semantics. In $\pstarstar$ the value of $\lastn$ remains the same, thus by the hypothesis~\eqref{eq:r-preservation-app} $\rd(h_1', a) = s_2'(\lastn)$. In both cases $\freads'$ holds.
\end{itemize}

\paragraph{\textbf{Case } $S \equiv \readcmd{x}{e}$:}
Let $o = \rd(h_1, \semtwo{e}{s_1})$.
The $\cmd{\dethavoccall{\readresult}}$ in $\encR(S)$ has the effect of assigning an arbitrary value to $\readresult$, but this arbitrary value is derived from the input value $\havocvar$. Since $\havocvar$ is an arbitrary input, after each $\dethavockw$ the value of $\readresult$ can be any $\sortint$ value.
\begin{itemize}
  \item $\fstacks'$: The only variable modified by $\encR(S)$ that is also in $\progvars(\pstar)$ is $x$. For other common variables, $\fstacks'$ holds by the hypothesis~\eqref{eq:r-preservation-app}.

  Let $e_v = \semtwo{e}{s_1} = \semtwo{e}{s_2}$ (by $\fstacks$).
  In $\pstar$, $s_1'(x) = \rd(h_1, e_v)$ by the semantics of $\rd$.
  In $\pstarstar$, we need to show that $s_2'(x) = \rd(h_1, e_v)$.

  Consider the following statement in $\encR(S)$:
  \[
    \ifcmd{\pg = e}
      { % then
        \assertcmd{R(\inG, \cnt, \lastn)};\
        \assigncmd{\readresult}{\lastn}% t = last_n;
      } { % else
        \dethavoccall{\readresult};\ 
        \assumecmd{R(\inG, \cnt, \readresult)}
      }
  \]
\begin{figure}[tbp]
  \centering
  \includegraphics[scale=0.85]{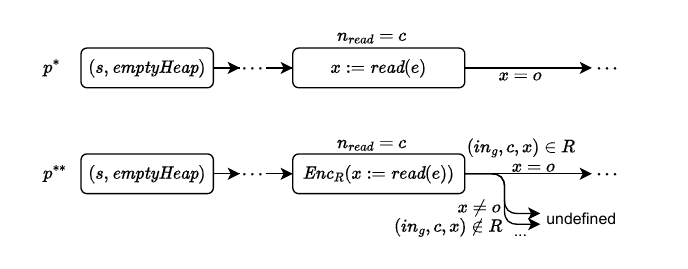}
  \Description{An illustration depicting the executions of p* and p** after a read. In p**, only a single execution survives due to R being a partial function. The diagram illustrates that all other executions are undefined after the read in p**.}
  \caption{A depiction of the executions of $\pstar$ and $\pstarstar$ after a $\readkw$. In $\pstarstar$, only one execution \emph{survives} due to $R$ being a partial function.}\label{fig:pstarstar-illustration-app}
\end{figure}
  %
  % Recall that $s(\pg) = a$ and $\pg$ is not modified in $\pstarstar$. Also recall that in 
  %Recall that $s(\pg) = a$ in~\eqref{eq:r-preservation-app}, and $\pg$ is not modified, therefore $\pg$ has the value $a$ in all stacks.
  \begin{itemize}
      \item If $a = e_v$, the \emph{then} branch is taken. The assertion $\assertcmd{R(\inG, \cnt, \lastn)}$ always passes under $\interp^*$, and due to the assignment $s_2'(\readresult) = s_2(\lastn)$. In this case $a = e_v$, so $\rd(h_1, e_v) = \rd(h_1, a) = s_2(\lastn) = s_2'(\readresult)$, which shows $s_2'(\readresult) = \rd(h_1, e_v)$.

      \item If $a \neq e_v$, the \emph{else} branch is taken, and $\assumecmd{R(\inG, \cnt, \readresult)}$ is executed after \emph{havocing} $\readresult$. Under $\interp^*$, $R$ is a partial function from its first two arguments to its last (by \Cref{prop:R-partial-app}). In the execution where $a = e_v$, the $\assertkw$ statement in the \emph{then} branch defines the partial function $R$ over $(s_2(\inG), s_2(\cnt)+1)$ for the value $s_2(\lastn)$. All executions of $\pstarstar$ use the same two values $(s_2(\inG), s_2(\cnt)+1)$ as the first two arguments of the $\assumekw$ statement, but the third argument, $\readresult$, holds an arbitrary value due to $\cmd{\dethavoccall{\readresult}}$. Since $R$ is a partial function, the $\assumekw$ statement will only pass in executions where after executing $\dethavockw$, the variable $\readresult$ holds the value $\rd(h_1, e_v)$. The result of all other executions is undefined, because the $\assumekw$ statement over $R$ will not pass for other values of $\readresult$ (a depiction is shown in \Cref{fig:pstarstar-illustration-app}). Therefore, at the exit of this branch we have $s_2'(\readresult) = \rd(h_1, e_v)$.
  \end{itemize}
  In both cases, we have $s_2'(\readresult) = \rd(h_1, e_v) = s_1'(readresult)$, which shows $\fstacks'$.

  \item $\foutcomes'$: $\encR(S)$ contains an $\assertkw$ in the \emph{then} branch, and an $\assumekw$ in the \emph{else} branch. Under $\interp^*$, the $\assertkw$ over $R$ always passes and has no effect on the outcome. If $a \neq e_v$, there is an $s$ where $\assumekw$ passes and $\fstacks'$ holds. Any other assertions and assumptions of the program (which are over the common variables of $\pstar$ and $\pstarstar$) will be evaluated under the same stack (due to $\fstacks'$), leading to the same outcome, which shows $\foutcomes'$. 

  \item $\fallocs'$: $\rd$ does not modify the heap, so $h_1' = h_1$ and $|h_1'| = |h_1|$. $\encR(S)$ does not modify $\allocctr$ in any execution, so $s_2'(\allocctr) = s_2(\allocctr)$. By the induction hypothesis $\fallocs$, $|h_1| = s_2(\allocctr)$, thus $|h_1'| = s_2'(\allocctr)$, and $\fallocs'$ holds.

  \item $\freads'$: We need to show $\rd(h_1', a) = s_2'(\lastn)$. Since $h_1' = h_1$, we need to show $\rd(h_1, a) = s_2'(\lastn)$. $\encR(S)$ does not modify $\lastn$, so $s_2'(\lastn) = s_2(\lastn)$. By the induction hypothesis $\freads$, we have $\rd(h_1, a) = s_2(\lastn)$. Therefore, $\rd(h_1', a) = \rd(h_1, a) = s_2(\lastn) = s_2'(\lastn)$, and $\freads'$ holds.
\end{itemize}

\end{proof}

\begin{theorem}[$\encR$ is correct]\label{thm:encR-proof-app}
  Let $p$ be \alang\ program. $p$ and $\encR(p)$ are equi-safe.
  % \[ \left( \exists s \in \stacks. \tupleget{\bigstepevalp{\interp}(p, s, \eh)}{1} = \error(P_1, \bar{v}_1) \right) \iff \left( \exists s \in \stacks. \tupleget{\bigstepevalp{\interp}(\encR(p), s, \eh)}{1} = \error(P_2, \bar{v}_2) \right).\]
\end{theorem}
\begin{proof}
  Let $\pstar = \encN(p)$ and $\pstarstar = \encR(\pstar)$. By \Cref{lem:encN-proof-app}, $p$ and $\pstar$ are equi-safe, so it suffices to show equi-safety between $\pstar$ and $\pstarstar$.
  By \Cref{lem:safety-under-i}, it suffices to show equi-safety under the least interpretation $\interp^*$.

  By \Cref{lem:encR-preservation-app}, for all $s$ we have $\tupleget{\bigstepevalp{\interp^*}(\pstar, s, \eh)}{1} = \sigma_1 = \sigma_2 = \tupleget{\bigstepevalp{\interp^*}(\pstarstar, s, \eh)}{1}$. That is, $\pstar$ and $\pstarstar$ always has the same outcome for every initial stack $s$; therefore if one is safe, the other will be safe too. By \Cref{lem:encN-proof-app} and transitivity, this shows $p$ and $\encR(p)$ are equi-safe.
\end{proof}

%%%%%%%%%%%%%%%%%%%%%%%%%
\subsection[Correctness of EncRW]{Correctness of $\encRW$}\label{app:proof-rw}
We first propose, without proof, that both $R$ and $W$ are partial functions in $\interp^*$. The proofs are similar to the proof of \Cref{prop:R-partial-app}.

% \begin{proposition}[Functional consistency of $\interp^*(W)$]\label{prop:W-partial}
%   Let $\interp^*$ be the least fixed point of the immediate consequence operator $\ico_p$ for \alang\ program $p$. Then,
%   \[
%     \forall g, n, v_1, v_2.\ \Bigl( (g, n, v_1) \in \interp^*(W) \wedge (g, n, v_2) \in \interp^*(W) \Bigr)
%       \Longrightarrow v_1 = v_2.
%   \]
% \end{proposition}

% We do not provide a proof for \Cref{prop:W-partial}, which is similar to the proof of \Cref{prop:R-partial-app}.

\begin{theorem}[Correctness of $\encRW$]\label{thm:encRW-proof-app}
Let $p$ be an \alang\ program. Then $p$ and $\encRW(p)$ are equi-safe.
\end{theorem}

The proof of this theorem largely mirrors the proof of correctness for $\encR(p)$, we provide a sketch:
\begin{proof}
  As one way to show the encoding correct, we first assume an intermediate encoding that increments a counter $c$ that is initialised to zero, and is incremented before every call to the heap operations $\readkw$, $\writekw$ or $\allockw$ (mimicking $\cnt$ of the $RW$ encoding). We also redeclare the heap object type $\obj$ as a tuple $\obj' : (\obj, \integer)$. In the intermediate encoding, every access to a heap object is replaced with an access to the first component of this tuple. At every $\allockw$ and $\writekw$, the variable $c$ is assigned to the second component of the tuple. We claim, without proof, that this intermediate encoding is correct, as it merely introduces a counter and boxes heap objects together with the current count value. Let $\pstar$ be this intermediate encoding of the original program $p$, and $\pstarstar$ be the $RW$ encoding of $\pstar$.

  We tackle the correctness of $\encRW$ in two steps. The first part of the $RW$ encoding is precisely the $R$ encoding, but the $R$ predicate now records the value $\lastwtcnt$ (rather than $\lastn$, the last object stored at $\pg$), corresponding to the $c$ value of the last write to $\pg$. 
  
  Next, let $y = f_W(x)$ be shorthand for $\cmd{\dethavoccall{y};\ \assumecmd{W(\inG, x, y)}}$. We take a shortcut by using the notation $f_W(x)(s)$ to directly represent the value of $y$ in the resulting stack after evaluating this statement under stack $s$.

  Recall \Cref{lem:encR-preservation-app}, we adapt it to the current setting by replacing $\freads$ with $\freadsone \wedge \freadstwo$, where $\freadsone \equiv \tupleget{\rd(h_1,a)}{2} = s_2(\lastwtcnt)$
  and $\freadstwo \equiv \tupleget{\rd(h_1,a)}{1} = f_W(\lastwtcnt)(s_2)$.

  Using the partial function property of $R$ in the $RW$ encoding, it is straightforward to show that the adapted $\freadsone$ should hold, following a similar argument as in the correctness proof for $\encR$. By using the functional consistency of $W$, showing the preservation of $\freadstwo$ is also straightforward. Intuitively, $W$ indexes the heap objects of the original program using the count value associated with each write.

  With the adapted \Cref{lem:encR-preservation-app} (whose proof we omit, which largely parallels the original proof), and using the argument in \Cref{thm:encR-proof-app} adapted to this encoding, the correctness of $\encRW$ follows.
\end{proof}

%%%%%%%%%%%%%%%%%%%%%%%%%

\subsection{Detailed Evaluation Results}\label{subsec:appendix-detailed-results}

\begin{longtable}{lcccc|cccccc}
\caption{Per benchmark results for the manually normalised benchmarks. ("\textbf{T}": safe, "\textbf{F}": unsafe, "U": unknown)}\label{tbl:per-benchmark-results}\\
Benchmark & \rotatebox{90}{None \cpa{}} & \rotatebox{90}{None \pred{}} & \rotatebox{90}{None \sea{}} & \rotatebox{90}{None \tri{}} & \rotatebox{90}{$\mathit{R}$ \sea{}} & \rotatebox{90}{$\mathit{R}$ \tri{}} & \rotatebox{90}{$\mathit{RW}$ \sea{}} & \rotatebox{90}{$\mathit{RW}$ \tri{}} & \rotatebox{90}{$\mathit{RWf}$ \sea{}} & \rotatebox{90}{$\mathit{RWf}$ \tri{}} \\ \toprule
\endfirsthead
Benchmark & \rotatebox{90}{None \cpa{}} & \rotatebox{90}{None \pred{}} & \rotatebox{90}{None \sea{}} & \rotatebox{90}{None \tri{}} & \rotatebox{90}{$\mathit{R}$ \sea{}} & \rotatebox{90}{$\mathit{R}$ \tri{}} & \rotatebox{90}{$\mathit{RW}$ \sea{}} & \rotatebox{90}{$\mathit{RW}$ \tri{}} & \rotatebox{90}{$\mathit{RWf}$ \sea{}} & \rotatebox{90}{$\mathit{RWf}$ \tri{}} \\ \toprule
\endhead
\bottomrule
\endlastfoot
listing-2-safe & U & \textbf{T} & U & U & \textbf{T} & U & U & U & \textbf{T} & \textbf{T} \\
listing-2-unsafe & \textbf{F} & \textbf{F} & \textbf{F} & \textbf{F} & U & \textbf{F} & U & \textbf{F} & U & \textbf{F} \\
simple\_and\_skiplist\_2lvl-1-safe & U & \textbf{T} & U & U & \textbf{T} & U & U & U & \textbf{T} & \textbf{T} \\
simple\_and\_skiplist\_2lvl-1-unsafe & \textbf{F} & \textbf{F} & \textbf{F} & \textbf{F} & U & \textbf{F} & U & \textbf{F} & U & \textbf{F} \\
simple\_built\_from\_end-safe & \textbf{T} & \textbf{T} & U & U & U & U & U & U & \textbf{T} & \textbf{T} \\
simple\_built\_from\_end-unsafe & \textbf{F} & \textbf{F} & \textbf{F} & \textbf{F} & U & \textbf{F} & U & \textbf{F} & U & \textbf{F} \\
sll-constant-no-loop-1-safe & \textbf{T} & \textbf{T} & \textbf{T} & \textbf{T} & U & \textbf{T} & U & \textbf{T} & \textbf{T} & \textbf{T} \\
sll-constant-no-loop-1-unsafe & \textbf{F} & \textbf{F} & \textbf{F} & \textbf{F} & U & \textbf{F} & U & \textbf{F} & U & \textbf{F} \\
sll-constant-no-loop-no-struct-safe & \textbf{T} & \textbf{T} & \textbf{T} & \textbf{T} & U & \textbf{T} & U & \textbf{T} & \textbf{T} & \textbf{T} \\
sll-constant-no-loop-no-struct-unsafe & \textbf{F} & \textbf{F} & \textbf{F} & \textbf{F} & U & \textbf{F} & U & \textbf{F} & U & \textbf{F} \\
sll-constant-unbounded-1-safe & U & \textbf{T} & U & U & \textbf{T} & U & U & U & \textbf{T} & \textbf{T} \\
sll-constant-unbounded-1-unsafe & \textbf{F} & \textbf{F} & \textbf{F} & \textbf{F} & U & \textbf{F} & U & \textbf{F} & U & \textbf{F} \\
sll-constant-unbounded-2-safe & U & \textbf{T} & U & U & \textbf{T} & \textbf{T} & \textbf{T} & \textbf{T} & \textbf{T} & \textbf{T} \\
sll-constant-unbounded-3-safe & U & \textbf{T} & U & \textbf{T} & U & \textbf{T} & U & \textbf{T} & \textbf{T} & \textbf{T} \\
sll-constant-unbounded-4-safe & U & \textbf{T} & U & U & U & U & U & U & \textbf{T} & \textbf{T} \\
sll-variable-unbounded-1-safe & U & U & U & U & \textbf{T} & \textbf{T} & U & \textbf{T} & \textbf{T} & \textbf{T} \\
sll-variable-unbounded-1-unsafe & \textbf{F} & \textbf{F} & \textbf{F} & \textbf{F} & U & \textbf{F} & U & \textbf{F} & U & \textbf{F} \\
sll-variable-unbounded-2-safe & U & U & U & U & \textbf{T} & U & U & U & U & U \\
sll-variable-unbounded-2-unsafe & \textbf{F} & \textbf{F} & \textbf{F} & \textbf{F} & U & \textbf{F} & U & \textbf{F} & U & \textbf{F} \\
sll-variable-unbounded-3-safe & \textbf{T} & \textbf{T} & \textbf{T} & \textbf{T} & \textbf{T} & \textbf{T} & \textbf{T} & \textbf{T} & \textbf{T} & \textbf{T} \\
tree-3-safe & U & U & U & U & U & U & U & U & \textbf{T} & U \\
tree-3-simpler-safe & U & U & U & U & U & \textbf{T} & U & \textbf{T} & U & \textbf{T} \\
\end{longtable}

%%%%%%%%%%5
%%%%%%%%%%%

\begin{longtable}{lcccc|cccccccccc}
\caption{Per benchmark results for the benchmarks from SV-COMP's \texttt{ReachSafety-Heap} category. The encoding columns use \tricerare. ("\textbf{T}": safe, "\textbf{F}": unsafe, "U": unknown)}\label{tbl:per-benchmark-results-tri}\\
Benchmark & \rotatebox{90}{None \cpa{}} & \rotatebox{90}{None \pred{}} & \rotatebox{90}{None \sea{}} & \rotatebox{90}{None \tri{}} & \rotatebox{90}{$\mathit{R}$} & \rotatebox{90}{$\mathit{R_{C}}$} & \rotatebox{90}{$\mathit{R_{T}}$} & \rotatebox{90}{$\mathit{RW}$} & \rotatebox{90}{$\mathit{RW_{T}}$} & \rotatebox{90}{$\mathit{RW_{CT}}$} & \rotatebox{90}{$\mathit{RWf}$} & \rotatebox{90}{$\mathit{RWf_{C}}$} & \rotatebox{90}{$\mathit{RWf_{T}}$} & \rotatebox{90}{$\mathit{RWf_{CT}}$} \\ \toprule
\endfirsthead
Benchmark & \rotatebox{90}{None \cpa{}} & \rotatebox{90}{None \pred{}} & \rotatebox{90}{None \sea{}} & \rotatebox{90}{None \tri{}} & \rotatebox{90}{$\mathit{R}$} & \rotatebox{90}{$\mathit{R_{C}}$} & \rotatebox{90}{$\mathit{R_{T}}$} & \rotatebox{90}{$\mathit{RW}$} & \rotatebox{90}{$\mathit{RW_{T}}$} & \rotatebox{90}{$\mathit{RW_{CT}}$} & \rotatebox{90}{$\mathit{RWf}$} & \rotatebox{90}{$\mathit{RWf_{C}}$} & \rotatebox{90}{$\mathit{RWf_{T}}$} & \rotatebox{90}{$\mathit{RWf_{CT}}$} \\ \toprule
\endhead
\bottomrule
\endlastfoot
alternating\_list-1-safe & U & U & U & U & U & U & U & U & U & U & U & U & U & U \\
alternating\_list-2-unsafe & \textbf{F} & \textbf{F} & \textbf{F} & \textbf{F} & \textbf{F} & \textbf{F} & \textbf{F} & \textbf{F} & \textbf{F} & \textbf{F} & \textbf{F} & \textbf{F} & \textbf{F} & \textbf{F} \\
calendar-safe & U & U & U & U & U & U & U & U & U & \textbf{T} & U & \textbf{T} & U & \textbf{T} \\
cart-safe & U & U & U & U & U & U & U & U & U & U & U & \textbf{T} & U & \textbf{T} \\
dancing-safe & U & \textbf{T} & U & U & U & U & U & U & U & U & U & U & U & U \\
dll-01-1-unsafe & \textbf{F} & \textbf{F} & \textbf{F} & \textbf{F} & \textbf{F} & \textbf{F} & \textbf{F} & U & \textbf{F} & \textbf{F} & \textbf{F} & \textbf{F} & \textbf{F} & \textbf{F} \\
dll-01-2-safe & U & U & U & U & U & U & U & U & U & U & U & U & U & U \\
dll-circular-1-unsafe & \textbf{F} & \textbf{F} & \textbf{F} & \textbf{F} & \textbf{F} & \textbf{F} & \textbf{F} & \textbf{F} & \textbf{F} & \textbf{F} & \textbf{F} & \textbf{F} & \textbf{F} & \textbf{F} \\
dll-circular-2-safe & U & \textbf{T} & U & U & U & U & U & U & U & U & U & U & U & U \\
dll-optional-1-safe & U & U & U & U & U & U & U & U & U & U & U & U & U & U \\
dll-optional-2-unsafe & \textbf{F} & \textbf{F} & \textbf{F} & \textbf{F} & \textbf{F} & \textbf{F} & \textbf{F} & \textbf{F} & \textbf{F} & \textbf{F} & \textbf{F} & \textbf{F} & \textbf{F} & \textbf{F} \\
dll-queue-1-safe & U & \textbf{T} & U & U & U & U & U & U & U & U & U & U & U & U \\
dll-queue-2-unsafe & \textbf{F} & \textbf{F} & U & U & U & U & U & U & U & U & U & U & U & U \\
dll-rb-cnstr\_1-1-safe & U & U & U & U & U & U & U & U & U & U & U & U & U & U \\
dll-rb-cnstr\_1-2-unsafe & \textbf{F} & \textbf{F} & \textbf{F} & \textbf{F} & \textbf{F} & \textbf{F} & \textbf{F} & \textbf{F} & \textbf{F} & \textbf{F} & \textbf{F} & \textbf{F} & \textbf{F} & \textbf{F} \\
dll-rb-sentinel-1-unsafe & \textbf{F} & \textbf{F} & \textbf{F} & U & U & U & U & U & U & U & U & U & U & U \\
dll-rb-sentinel-2-safe & U & U & U & U & U & U & U & U & U & U & U & U & U & U \\
dll-reverse-safe & U & \textbf{T} & U & U & U & U & U & U & U & U & U & U & U & U \\
dll-simple-white-blue-1-unsafe & \textbf{F} & \textbf{F} & U & \textbf{F} & \textbf{F} & \textbf{F} & \textbf{F} & U & \textbf{F} & \textbf{F} & \textbf{F} & \textbf{F} & \textbf{F} & \textbf{F} \\
dll-simple-white-blue-2-safe & U & \textbf{T} & U & U & U & U & U & U & U & U & U & U & U & U \\
dll-sorted-1-unsafe & \textbf{F} & \textbf{F} & U & U & U & \textbf{F} & U & U & U & U & U & U & U & U \\
dll-sorted-2-safe & U & \textbf{T} & U & U & U & U & U & U & U & U & U & U & U & U \\
dll-token-1-safe & U & \textbf{T} & U & U & U & U & U & U & U & U & U & U & U & U \\
dll-token-2-unsafe & \textbf{F} & \textbf{F} & \textbf{F} & U & \textbf{F} & \textbf{F} & \textbf{F} & U & U & \textbf{F} & \textbf{F} & \textbf{F} & U & \textbf{F} \\
dll2c\_append\_equal-safe & \textbf{T} & \textbf{T} & U & \textbf{T} & U & U & U & U & U & U & U & \textbf{T} & U & \textbf{T} \\
dll2c\_append\_unequal-safe & \textbf{T} & \textbf{T} & U & \textbf{T} & U & U & U & U & U & U & U & \textbf{T} & U & \textbf{T} \\
dll2c\_insert\_equal-safe & \textbf{T} & \textbf{T} & U & U & U & U & U & U & U & U & U & \textbf{T} & U & \textbf{T} \\
dll2c\_insert\_unequal-safe & \textbf{T} & \textbf{T} & \textbf{T} & U & U & U & U & U & U & U & U & \textbf{T} & U & \textbf{T} \\
dll2c\_prepend\_equal-safe & \textbf{T} & \textbf{T} & U & U & U & U & U & U & U & U & U & \textbf{T} & U & \textbf{T} \\
dll2c\_prepend\_unequal-safe & \textbf{T} & \textbf{T} & U & U & U & U & U & U & U & U & U & \textbf{T} & U & U \\
dll2c\_remove\_all-safe & \textbf{T} & \textbf{T} & \textbf{T} & \textbf{T} & U & \textbf{T} & \textbf{T} & U & U & \textbf{T} & \textbf{T} & \textbf{T} & U & \textbf{T} \\
dll2c\_remove\_all\_reverse-safe & \textbf{T} & \textbf{T} & \textbf{T} & \textbf{T} & U & \textbf{T} & \textbf{T} & U & U & \textbf{T} & \textbf{T} & \textbf{T} & U & \textbf{T} \\
dll2c\_update\_all-safe & \textbf{T} & \textbf{T} & U & U & U & \textbf{T} & U & U & U & \textbf{T} & \textbf{T} & \textbf{T} & U & \textbf{T} \\
dll2c\_update\_all\_reverse-safe & \textbf{T} & \textbf{T} & U & U & U & \textbf{T} & U & U & U & \textbf{T} & \textbf{T} & \textbf{T} & U & \textbf{T} \\
dll2n\_append\_equal-safe & \textbf{T} & \textbf{T} & U & U & U & U & U & U & U & U & U & \textbf{T} & U & \textbf{T} \\
dll2n\_append\_unequal-safe & \textbf{T} & \textbf{T} & U & U & U & U & U & U & U & U & U & \textbf{T} & U & \textbf{T} \\
dll2n\_insert\_equal-safe & \textbf{T} & \textbf{T} & U & U & U & U & U & U & U & U & U & \textbf{T} & U & \textbf{T} \\
dll2n\_insert\_unequal-safe & \textbf{T} & \textbf{T} & U & U & U & U & U & U & U & U & U & \textbf{T} & U & \textbf{T} \\
dll2n\_prepend\_equal-safe & \textbf{T} & \textbf{T} & U & U & U & U & U & U & U & \textbf{T} & U & \textbf{T} & U & \textbf{T} \\
dll2n\_prepend\_unequal-safe & \textbf{T} & \textbf{T} & U & U & U & \textbf{T} & U & U & U & \textbf{T} & U & \textbf{T} & U & \textbf{T} \\
dll2n\_remove\_all-safe & \textbf{T} & \textbf{T} & U & \textbf{T} & \textbf{T} & \textbf{T} & \textbf{T} & U & U & \textbf{T} & \textbf{T} & \textbf{T} & U & \textbf{T} \\
dll2n\_remove\_all\_reverse-safe & \textbf{T} & \textbf{T} & \textbf{T} & \textbf{T} & U & \textbf{T} & U & U & U & \textbf{T} & U & \textbf{T} & U & \textbf{T} \\
dll2n\_update\_all-safe & \textbf{T} & \textbf{T} & U & U & U & \textbf{T} & U & U & U & \textbf{T} & U & \textbf{T} & U & \textbf{T} \\
dll2n\_update\_all\_reverse-safe & \textbf{T} & \textbf{T} & \textbf{T} & \textbf{T} & \textbf{T} & \textbf{T} & \textbf{T} & \textbf{T} & \textbf{T} & \textbf{T} & \textbf{T} & \textbf{T} & \textbf{T} & \textbf{T} \\
dll\_circular\_traversal-2-safe & \textbf{T} & \textbf{T} & U & U & U & U & U & U & U & U & U & U & U & U \\
dll\_nullified-1-unsafe & \textbf{F} & \textbf{F} & U & U & U & U & U & U & U & U & U & U & U & U \\
dll\_nullified-2-safe & \textbf{T} & \textbf{T} & U & U & U & U & U & U & U & U & U & \textbf{T} & \textbf{T} & \textbf{T} \\
hash\_fun-safe & U & U & U & U & U & U & U & U & \textbf{T} & \textbf{T} & U & U & \textbf{T} & \textbf{T} \\
list-1-safe & U & \textbf{T} & U & U & U & U & U & U & U & U & U & U & U & U \\
list-2-unsafe & \textbf{F} & \textbf{F} & U & U & \textbf{F} & \textbf{F} & \textbf{F} & U & U & \textbf{F} & U & \textbf{F} & U & \textbf{F} \\
list-ext-unsafe & \textbf{F} & \textbf{F} & \textbf{F} & \textbf{F} & \textbf{F} & \textbf{F} & \textbf{F} & \textbf{F} & \textbf{F} & \textbf{F} & \textbf{F} & \textbf{F} & \textbf{F} & \textbf{F} \\
list\_and\_tree\_cnstr-1-unsafe & \textbf{F} & \textbf{F} & U & U & U & U & U & U & U & U & U & U & U & U \\
list\_and\_tree\_cnstr-2-safe & U & U & U & U & U & U & U & U & U & U & U & U & U & U \\
list\_flag-1-unsafe & \textbf{F} & \textbf{F} & \textbf{F} & \textbf{F} & \textbf{F} & \textbf{F} & \textbf{F} & \textbf{F} & \textbf{F} & \textbf{F} & \textbf{F} & \textbf{F} & \textbf{F} & \textbf{F} \\
list\_flag-2-safe & U & \textbf{T} & U & U & U & U & U & U & U & \textbf{T} & U & U & U & \textbf{T} \\
list\_search-1-unsafe & \textbf{F} & \textbf{F} & U & \textbf{F} & U & \textbf{F} & U & U & U & \textbf{F} & \textbf{F} & \textbf{F} & U & \textbf{F} \\
list\_search-2-safe & \textbf{T} & \textbf{T} & \textbf{T} & U & U & U & U & U & U & \textbf{T} & \textbf{T} & U & U & \textbf{T} \\
merge\_sort-1-unsafe & \textbf{F} & \textbf{F} & \textbf{F} & U & U & U & U & U & U & U & U & U & U & U \\
merge\_sort-2-safe & U & \textbf{T} & U & U & U & U & U & U & U & U & U & U & U & U \\
min\_max-safe & U & U & U & U & U & U & U & U & U & U & U & U & U & U \\
packet\_filter-safe & U & U & U & U & U & U & U & U & U & U & U & U & U & U \\
process\_queue-safe & U & U & U & U & U & U & U & U & U & U & U & U & U & U \\
quick\_sort\_split-safe & \textbf{T} & U & \textbf{T} & U & \textbf{T} & \textbf{T} & \textbf{T} & \textbf{T} & \textbf{T} & \textbf{T} & \textbf{T} & \textbf{T} & \textbf{T} & \textbf{T} \\
rule60\_list-safe & \textbf{T} & \textbf{T} & \textbf{T} & \textbf{T} & \textbf{T} & \textbf{T} & \textbf{T} & \textbf{T} & \textbf{T} & \textbf{T} & \textbf{T} & \textbf{T} & \textbf{T} & \textbf{T} \\
running\_example-safe & U & U & U & U & U & U & U & U & U & U & U & U & U & \textbf{T} \\
shared\_mem1-safe & U & \textbf{T} & U & U & U & U & U & U & U & U & U & U & U & U \\
shared\_mem2-safe & U & \textbf{T} & \textbf{T} & U & U & U & U & U & U & U & U & U & U & U \\
simple-1-unsafe & \textbf{F} & \textbf{F} & \textbf{F} & U & \textbf{F} & \textbf{F} & \textbf{F} & \textbf{F} & \textbf{F} & \textbf{F} & \textbf{F} & \textbf{F} & \textbf{F} & \textbf{F} \\
simple-2-safe & U & \textbf{T} & U & U & U & U & U & U & U & \textbf{T} & U & U & U & \textbf{T} \\
simple-ext-unsafe & \textbf{F} & \textbf{F} & \textbf{F} & \textbf{F} & \textbf{F} & \textbf{F} & \textbf{F} & \textbf{F} & \textbf{F} & \textbf{F} & \textbf{F} & \textbf{F} & \textbf{F} & \textbf{F} \\
simple\_and\_skiplist\_2lvl-1-unsafe & \textbf{F} & \textbf{F} & \textbf{F} & U & \textbf{F} & \textbf{F} & \textbf{F} & \textbf{F} & \textbf{F} & \textbf{F} & \textbf{F} & \textbf{F} & \textbf{F} & \textbf{F} \\
simple\_and\_skiplist\_2lvl-2-safe & U & U & U & U & U & U & U & U & U & \textbf{T} & U & U & U & \textbf{T} \\
simple\_built\_from\_end-safe & U & \textbf{T} & U & U & U & U & U & U & \textbf{T} & \textbf{T} & U & U & \textbf{T} & \textbf{T} \\
simple\_search\_value-1-safe & U & U & U & U & U & U & U & U & U & U & U & U & U & U \\
simple\_search\_value-2-unsafe & \textbf{F} & \textbf{F} & U & U & U & U & U & U & U & U & U & U & U & U \\
sll-01-1-safe & U & U & U & U & U & U & U & U & U & U & U & U & U & U \\
sll-01-2-unsafe & \textbf{F} & U & U & \textbf{F} & U & \textbf{F} & U & U & U & U & U & U & U & U \\
sll-buckets-1-safe & U & \textbf{T} & U & U & U & U & U & U & U & U & U & U & U & U \\
sll-buckets-2-unsafe & \textbf{F} & \textbf{F} & \textbf{F} & U & U & U & U & U & U & U & U & U & U & U \\
sll-circular-1-safe & U & \textbf{T} & U & U & U & U & U & U & U & U & U & U & U & U \\
sll-circular-2-unsafe & \textbf{F} & \textbf{F} & \textbf{F} & \textbf{F} & \textbf{F} & \textbf{F} & \textbf{F} & \textbf{F} & \textbf{F} & \textbf{F} & \textbf{F} & \textbf{F} & \textbf{F} & \textbf{F} \\
sll-optional-1-safe & U & U & U & U & U & U & U & U & U & U & U & U & U & U \\
sll-optional-2-unsafe & \textbf{F} & \textbf{F} & \textbf{F} & \textbf{F} & \textbf{F} & \textbf{F} & \textbf{F} & \textbf{F} & \textbf{F} & \textbf{F} & \textbf{F} & \textbf{F} & \textbf{F} & \textbf{F} \\
sll-queue-1-safe & U & \textbf{T} & U & U & U & U & U & U & U & U & U & U & U & U \\
sll-queue-2-unsafe & \textbf{F} & \textbf{F} & U & U & U & U & U & U & U & U & U & U & U & U \\
sll-rb-cnstr\_1-1-safe & U & U & U & U & U & U & U & U & U & U & U & U & U & U \\
sll-rb-cnstr\_1-2-unsafe & \textbf{F} & \textbf{F} & \textbf{F} & \textbf{F} & \textbf{F} & \textbf{F} & \textbf{F} & \textbf{F} & \textbf{F} & \textbf{F} & \textbf{F} & \textbf{F} & \textbf{F} & \textbf{F} \\
sll-rb-sentinel-1-safe & U & U & U & U & U & U & U & U & U & U & U & U & U & U \\
sll-rb-sentinel-2-unsafe & \textbf{F} & \textbf{F} & \textbf{F} & U & U & U & U & U & U & U & U & U & U & \textbf{F} \\
sll-reverse\_simple-safe & U & \textbf{T} & U & U & U & U & U & U & U & U & U & U & U & U \\
sll-simple-white-blue-1-safe & U & \textbf{T} & U & U & U & U & U & U & U & U & U & U & U & U \\
sll-simple-white-blue-2-unsafe & \textbf{F} & \textbf{F} & \textbf{F} & \textbf{F} & \textbf{F} & \textbf{F} & \textbf{F} & U & \textbf{F} & \textbf{F} & \textbf{F} & \textbf{F} & \textbf{F} & \textbf{F} \\
sll-sorted-1-unsafe & \textbf{F} & \textbf{F} & U & U & U & U & U & U & U & U & U & U & U & U \\
sll-sorted-2-safe & U & \textbf{T} & U & U & U & U & U & U & U & U & U & U & U & U \\
sll-token-1-unsafe & \textbf{F} & \textbf{F} & \textbf{F} & U & \textbf{F} & \textbf{F} & \textbf{F} & \textbf{F} & \textbf{F} & \textbf{F} & \textbf{F} & \textbf{F} & \textbf{F} & \textbf{F} \\
sll-token-2-safe & U & \textbf{T} & U & U & U & U & U & U & U & U & U & U & U & U \\
sll2c\_append\_equal-safe & \textbf{T} & \textbf{T} & U & U & U & U & U & U & U & U & U & \textbf{T} & U & \textbf{T} \\
sll2c\_append\_unequal-safe & \textbf{T} & \textbf{T} & U & \textbf{T} & U & U & U & U & U & U & U & U & U & U \\
sll2c\_insert\_equal-safe & \textbf{T} & \textbf{T} & U & U & U & \textbf{T} & U & U & U & \textbf{T} & U & \textbf{T} & U & \textbf{T} \\
sll2c\_insert\_unequal-safe & \textbf{T} & \textbf{T} & \textbf{T} & U & U & \textbf{T} & U & U & U & \textbf{T} & U & \textbf{T} & U & \textbf{T} \\
sll2c\_prepend\_equal-safe & \textbf{T} & \textbf{T} & U & U & U & U & U & U & U & U & U & \textbf{T} & U & \textbf{T} \\
sll2c\_prepend\_unequal-safe & \textbf{T} & \textbf{T} & U & U & U & U & U & U & U & U & U & \textbf{T} & U & U \\
sll2c\_remove\_all-safe & \textbf{T} & \textbf{T} & U & \textbf{T} & \textbf{T} & \textbf{T} & U & U & U & \textbf{T} & \textbf{T} & \textbf{T} & U & \textbf{T} \\
sll2c\_remove\_all\_reverse-safe & \textbf{T} & \textbf{T} & \textbf{T} & \textbf{T} & U & \textbf{T} & U & U & U & \textbf{T} & U & \textbf{T} & U & \textbf{T} \\
sll2c\_update\_all-safe & \textbf{T} & \textbf{T} & U & U & \textbf{T} & \textbf{T} & \textbf{T} & U & \textbf{T} & \textbf{T} & \textbf{T} & \textbf{T} & U & \textbf{T} \\
sll2c\_update\_all\_reverse-safe & \textbf{T} & \textbf{T} & U & U & \textbf{T} & \textbf{T} & \textbf{T} & U & \textbf{T} & \textbf{T} & \textbf{T} & \textbf{T} & \textbf{T} & \textbf{T} \\
sll2n\_append\_equal-safe & \textbf{T} & \textbf{T} & U & U & U & \textbf{T} & U & U & U & \textbf{T} & U & \textbf{T} & U & \textbf{T} \\
sll2n\_append\_unequal-safe & \textbf{T} & \textbf{T} & U & U & U & \textbf{T} & U & U & U & U & U & \textbf{T} & U & U \\
sll2n\_insert\_equal-safe & \textbf{T} & \textbf{T} & U & U & U & \textbf{T} & U & U & U & U & U & \textbf{T} & U & \textbf{T} \\
sll2n\_insert\_unequal-safe & \textbf{T} & \textbf{T} & U & U & U & U & U & U & U & \textbf{T} & U & \textbf{T} & U & \textbf{T} \\
sll2n\_prepend\_equal-safe & \textbf{T} & \textbf{T} & U & U & U & \textbf{T} & \textbf{T} & U & U & \textbf{T} & U & \textbf{T} & U & \textbf{T} \\
sll2n\_prepend\_unequal-safe & \textbf{T} & \textbf{T} & U & U & U & \textbf{T} & \textbf{T} & U & U & \textbf{T} & U & \textbf{T} & U & \textbf{T} \\
sll2n\_remove\_all-safe & \textbf{T} & \textbf{T} & U & U & \textbf{T} & \textbf{T} & \textbf{T} & \textbf{T} & \textbf{T} & \textbf{T} & \textbf{T} & \textbf{T} & \textbf{T} & \textbf{T} \\
sll2n\_remove\_all\_reverse-safe & \textbf{T} & \textbf{T} & U & U & \textbf{T} & \textbf{T} & \textbf{T} & U & U & \textbf{T} & \textbf{T} & \textbf{T} & \textbf{T} & \textbf{T} \\
sll2n\_update\_all-safe & \textbf{T} & \textbf{T} & U & U & \textbf{T} & \textbf{T} & \textbf{T} & U & U & \textbf{T} & U & \textbf{T} & U & \textbf{T} \\
sll2n\_update\_all\_reverse-safe & \textbf{T} & \textbf{T} & U & U & \textbf{T} & \textbf{T} & \textbf{T} & U & U & \textbf{T} & \textbf{T} & \textbf{T} & U & \textbf{T} \\
sll\_circular\_traversal-1-safe & \textbf{T} & \textbf{T} & U & U & U & U & U & U & U & U & U & U & U & U \\
sll\_length\_check-1-unsafe & \textbf{F} & \textbf{F} & \textbf{F} & \textbf{F} & \textbf{F} & \textbf{F} & \textbf{F} & \textbf{F} & \textbf{F} & \textbf{F} & \textbf{F} & U & \textbf{F} & \textbf{F} \\
sll\_length\_check-2-safe & U & \textbf{T} & U & U & U & U & U & U & U & U & U & U & U & U \\
sll\_nondet\_insert-1-unsafe & \textbf{F} & \textbf{F} & \textbf{F} & U & \textbf{F} & \textbf{F} & \textbf{F} & \textbf{F} & \textbf{F} & \textbf{F} & \textbf{F} & \textbf{F} & \textbf{F} & \textbf{F} \\
sll\_nondet\_insert-2-safe & U & \textbf{T} & U & U & U & U & U & U & U & U & U & U & U & U \\
sll\_of\_sll\_nondet\_append-1-safe & U & U & U & U & U & U & U & U & U & U & U & U & U & U \\
sll\_of\_sll\_nondet\_append-2-unsafe & U & \textbf{F} & U & U & U & U & U & U & U & U & U & U & U & U \\
sll\_to\_dll\_rev-1-unsafe & \textbf{F} & \textbf{F} & U & \textbf{F} & \textbf{F} & \textbf{F} & \textbf{F} & \textbf{F} & \textbf{F} & \textbf{F} & \textbf{F} & \textbf{F} & \textbf{F} & \textbf{F} \\
sll\_to\_dll\_rev-2-safe & U & \textbf{T} & U & U & U & U & U & U & U & U & U & U & U & U \\
splice-1-unsafe & \textbf{F} & \textbf{F} & \textbf{F} & \textbf{F} & \textbf{F} & \textbf{F} & \textbf{F} & U & \textbf{F} & \textbf{F} & \textbf{F} & \textbf{F} & \textbf{F} & \textbf{F} \\
splice-2-safe & U & U & U & U & U & U & U & U & U & U & U & U & U & U \\
test\_malloc-1-safe & \textbf{T} & \textbf{T} & \textbf{T} & \textbf{T} & \textbf{T} & \textbf{T} & \textbf{T} & \textbf{T} & \textbf{T} & \textbf{T} & \textbf{T} & \textbf{T} & \textbf{T} & \textbf{T} \\
test\_malloc-2-safe & \textbf{T} & \textbf{T} & \textbf{T} & \textbf{T} & \textbf{T} & \textbf{T} & \textbf{T} & \textbf{T} & \textbf{T} & \textbf{T} & \textbf{T} & \textbf{T} & \textbf{T} & \textbf{T} \\
tree-2-unsafe & \textbf{F} & \textbf{F} & \textbf{F} & \textbf{F} & U & U & U & U & U & U & U & U & U & U \\
tree-3-safe & U & U & U & U & U & U & U & U & U & U & U & U & U & U \\
tree-4-unsafe & \textbf{F} & \textbf{F} & \textbf{F} & U & U & U & U & U & U & U & U & U & U & U \\
\end{longtable}

%% file: refs.bib
@String{Computing = "Computing" }

@String{Computer = "{IEEE} Computer" }

@String{Springer = "Springer-Verlag" }

@article{esenTheoryHeapConstrained2021,
  author       = {Zafer Esen and
                  Philipp R{\"{u}}mmer},
  title        = {A Theory of Heap for Constrained Horn Clauses (Extended Technical
                  Report)},
  journal      = {CoRR},
  volume       = {abs/2104.04224},
  year         = {2021},
  url          = {https://arxiv.org/abs/2104.04224},
  eprinttype    = {arXiv},
  eprint       = {2104.04224},
  bibsource    = {dblp computer science bibliography, https://dblp.org}
}

@inproceedings{DBLP:conf/smt/EsenR22,
  author       = {Zafer Esen and
                  Philipp R{\"{u}}mmer},
  editor       = {David D{\'{e}}harbe and
                  Antti E. J. Hyv{\"{a}}rinen},
  title        = {An {SMT-LIB} Theory of Heaps},
  booktitle    = {Proceedings of the 20th Internal Workshop on Satisfiability Modulo
                  Theories co-located with the 11th International Joint Conference on
                  Automated Reasoning {(IJCAR} 2022) part of the 8th Federated Logic
                  Conference (FLoC 2022), Haifa, Israel, August 11-12, 2022},
  series       = {{CEUR} Workshop Proceedings},
  volume       = {3185},
  pages        = {38--53},
  publisher    = {CEUR-WS.org},
  year         = {2022},
  url          = {https://ceur-ws.org/Vol-3185/paper1180.pdf},
  bibsource    = {dblp computer science bibliography, https://dblp.org}
}

@inproceedings{tricera,
  author    = {Zafer Esen and
               Philipp R{\"{u}}mmer},
  editor    = {Alberto Griggio and
               Neha Rungta},
  title     = {Tricera: Verifying {C} Programs Using the Theory of Heaps},
  booktitle = {22nd Formal Methods in Computer-Aided Design, {FMCAD} 2022, Trento,
               Italy, October 17-21, 2022},
  pages     = {380--391},
  publisher = {{IEEE}},
  year      = {2022},
  url       = {https://doi.org/10.34727/2022/isbn.978-3-85448-053-2\_45},
  doi       = {10.34727/2022/isbn.978-3-85448-053-2\_45},
  bibsource = {dblp computer science bibliography, https://dblp.org}
}

@inproceedings{DBLP:conf/birthday/BjornerGMR15,
  author    = {Nikolaj Bj{\o}rner and
               Arie Gurfinkel and
               Kenneth L. McMillan and
               Andrey Rybalchenko},
  editor    = {Lev D. Beklemishev and
               Andreas Blass and
               Nachum Dershowitz and
               Bernd Finkbeiner and
               Wolfram Schulte},
  title     = {Horn Clause Solvers for Program Verification},
  booktitle = {Fields of Logic and Computation {II} - Essays Dedicated to Yuri Gurevich
               on the Occasion of His 75th Birthday},
  series    = {Lecture Notes in Computer Science},
  volume    = {9300},
  pages     = {24--51},
  publisher = {Springer},
  year      = {2015},
  url       = {https://doi.org/10.1007/978-3-319-23534-9\_2},
  doi       = {10.1007/978-3-319-23534-9\_2},
  bibsource = {dblp computer science bibliography, https://dblp.org}
}

@inproceedings{DBLP:conf/popl/CousotC77,
  author    = {Patrick Cousot and
               Radhia Cousot},
  editor    = {Robert M. Graham and
               Michael A. Harrison and
               Ravi Sethi},
  title     = {Abstract Interpretation: {A} Unified Lattice Model for Static Analysis
               of Programs by Construction or Approximation of Fixpoints},
  booktitle = {Conference Record of the Fourth {ACM} Symposium on Principles of Programming
               Languages, Los Angeles, California, USA, January 1977},
  pages     = {238--252},
  publisher = {{ACM}},
  year      = {1977},
  url       = {https://doi.org/10.1145/512950.512973},
  doi       = {10.1145/512950.512973},
  bibsource = {dblp computer science bibliography, https://dblp.org}
}

@inproceedings{DBLP:conf/birthday/Cousot03,
  author    = {Patrick Cousot},
  editor    = {Nachum Dershowitz},
  title     = {Verification by Abstract Interpretation},
  booktitle = {Verification: Theory and Practice, Essays Dedicated to Zohar Manna
               on the Occasion of His 64th Birthday},
  series    = {Lecture Notes in Computer Science},
  volume    = {2772},
  pages     = {243--268},
  publisher = {Springer},
  year      = {2003},
  url       = {https://doi.org/10.1007/978-3-540-39910-0\_11},
  doi       = {10.1007/978-3-540-39910-0\_11},
  bibsource = {dblp computer science bibliography, https://dblp.org}
}

@inproceedings{DBLP:conf/cav/KahsaiRSS16,
  author       = {Temesghen Kahsai and
                  Philipp R{\"{u}}mmer and
                  Huascar Sanchez and
                  Martin Sch{\"{a}}f},
  editor       = {Swarat Chaudhuri and
                  Azadeh Farzan},
  title        = {{JayHorn}: {A} Framework for Verifying {Java} programs},
  booktitle    = {Computer Aided Verification - 28th International Conference, {CAV}
                  2016, Toronto, ON, Canada, July 17-23, 2016, Proceedings, Part {I}},
  series       = {Lecture Notes in Computer Science},
  volume       = {9779},
  pages        = {352--358},
  publisher    = {Springer},
  year         = {2016},
  url          = {https://doi.org/10.1007/978-3-319-41528-4\_19},
  doi          = {10.1007/978-3-319-41528-4\_19},
  bibsource    = {dblp computer science bibliography, https://dblp.org}
}

@inproceedings{DBLP:conf/lpar/KahsaiKRS17,
  author    = {Temesghen Kahsai and
               Rody Kersten and
               Philipp R{\"{u}}mmer and
               Martin Sch{\"{a}}f},
  editor    = {Thomas Eiter and
               David Sands},
  title     = {Quantified Heap Invariants for Object-Oriented Programs},
  booktitle = {LPAR-21, 21st International Conference on Logic for Programming, Artificial
               Intelligence and Reasoning, Maun, Botswana, May 7-12, 2017},
  series    = {EPiC Series in Computing},
  volume    = {46},
  pages     = {368--384},
  publisher = {EasyChair},
  year      = {2017},
  url       = {https://doi.org/10.29007/zrct},
  doi       = {10.29007/zrct},
  bibsource = {dblp computer science bibliography, https://dblp.org}
}

@inproceedings{DBLP:conf/cade/BohmeM11,
  author    = {Sascha B{\"{o}}hme and
               Michal Moskal},
  editor    = {Nikolaj S. Bj{\o}rner and
               Viorica Sofronie{-}Stokkermans},
  title     = {Heaps and Data Structures: {A} Challenge for Automated Provers},
  booktitle = {Automated Deduction - {CADE-23} - 23rd International Conference on
               Automated Deduction, Wroclaw, Poland, July 31 - August 5, 2011. Proceedings},
  series    = {Lecture Notes in Computer Science},
  volume    = {6803},
  pages     = {177--191},
  publisher = {Springer},
  year      = {2011},
  url       = {https://doi.org/10.1007/978-3-642-22438-6\_15},
  doi       = {10.1007/978-3-642-22438-6\_15},
  bibsource = {dblp computer science bibliography, https://dblp.org}
}

@article{DBLP:journals/tplp/AngelisFGHPP22,
  author    = {Emanuele {De Angelis} and
               Fabio Fioravanti and
               John P. Gallagher and
               Manuel V. Hermenegildo and
               Alberto Pettorossi and
               Maurizio Proietti},
  title     = {Analysis and Transformation of Constrained Horn Clauses for Program
               Verification},
  journal   = {Theory Pract. Log. Program.},
  volume    = {22},
  number    = {6},
  pages     = {974--1042},
  year      = {2022},
  url       = {https://doi.org/10.1017/S1471068421000211},
  doi       = {10.1017/S1471068421000211},
  bibsource = {dblp computer science bibliography, https://dblp.org}
}

@article{tarski-fp,
author = {Alfred Tarski},
title = {{A lattice-theoretical fixpoint theorem and its applications.}},
volume = {5},
journal = {Pacific Journal of Mathematics},
number = {2},
publisher = {Pacific Journal of Mathematics, A Non-profit Corporation},
pages = {285 -- 309},
year = {1955},
url = {https://doi.org/10.2140/pjm.1955.5.285}
}

@inproceedings{DBLP:conf/ifip/McCarthy62,
  author       = {John McCarthy},
  title        = {Towards a Mathematical Science of Computation},
  booktitle    = {Information Processing, Proceedings of the 2nd {IFIP} Congress 1962,
                  Munich, Germany, August 27 - September 1, 1962},
  pages        = {21--28},
  publisher    = {North-Holland},
  year         = {1962},
  bibsource    = {dblp computer science bibliography, https://dblp.org}
}

@inproceedings{DBLP:conf/ecoop/WesleyCNTWG24,
  author       = {Scott Wesley and
                  Maria Christakis and
                  Jorge A. Navas and
                  Richard J. Trefler and
                  Valentin W{\"{u}}stholz and
                  Arie Gurfinkel},
  editor       = {Jonathan Aldrich and
                  Guido Salvaneschi},
  title        = {Inductive Predicate Synthesis Modulo Programs},
  booktitle    = {38th European Conference on Object-Oriented Programming, {ECOOP} 2024,
                  September 16-20, 2024, Vienna, Austria},
  series       = {LIPIcs},
  volume       = {313},
  pages        = {43:1--43:30},
  publisher    = {Schloss Dagstuhl - Leibniz-Zentrum f{\"{u}}r Informatik},
  year         = {2024},
  url          = {https://doi.org/10.4230/LIPIcs.ECOOP.2024.43},
  doi          = {10.4230/LIPICS.ECOOP.2024.43},
  bibsource    = {dblp computer science bibliography, https://dblp.org}
}

@article{DBLP:journals/tcs/AbadiL91,
  author       = {Mart{\'{\i}}n Abadi and
                  Leslie Lamport},
  title        = {The Existence of Refinement Mappings},
  journal      = {Theor. Comput. Sci.},
  volume       = {82},
  number       = {2},
  pages        = {253--284},
  year         = {1991},
  url          = {https://doi.org/10.1016/0304-3975(91)90224-P},
  doi          = {10.1016/0304-3975(91)90224-P},
  bibsource    = {dblp computer science bibliography, https://dblp.org}
}

@inproceedings{cpachecker1,
  author    = {Dirk Beyer and
               M. Erkan Keremoglu},
  editor    = {Ganesh Gopalakrishnan and
               Shaz Qadeer},
  title     = {{CPAchecker}: {A} Tool for Configurable Software Verification},
  booktitle = {Computer Aided Verification - 23rd International Conference, {CAV}
               2011, Snowbird, UT, USA, July 14-20, 2011. Proceedings},
  series    = {Lecture Notes in Computer Science},
  volume    = {6806},
  pages     = {184--190},
  publisher = {Springer},
  year      = {2011},
  url       = {https://doi.org/10.1007/978-3-642-22110-1\_16},
  doi       = {10.1007/978-3-642-22110-1\_16},
  bibsource = {dblp computer science bibliography, https://dblp.org}
}

@inproceedings{cpachecker2,
  author       = {Daniel Baier and
                  Dirk Beyer and
                  Po{-}Chun Chien and
                  Marek Jankola and
                  Matthias Kettl and
                  Nian{-}Ze Lee and
                  Thomas Lemberger and
                  Marian Lingsch Rosenfeld and
                  Martin Spiessl and
                  Henrik Wachowitz and
                  Philipp Wendler},
  editor       = {Bernd Finkbeiner and
                  Laura Kov{\'{a}}cs},
  title        = {CPAchecker 2.3 with Strategy Selection - (Competition Contribution)},
  booktitle    = {Tools and Algorithms for the Construction and Analysis of Systems
                  - 30th International Conference, {TACAS} 2024, Held as Part of the
                  European Joint Conferences on Theory and Practice of Software, {ETAPS}
                  2024, Luxembourg City, Luxembourg, April 6-11, 2024, Proceedings,
                  Part {III}},
  series       = {Lecture Notes in Computer Science},
  volume       = {14572},
  pages        = {359--364},
  publisher    = {Springer},
  year         = {2024},
  url          = {https://doi.org/10.1007/978-3-031-57256-2\_21},
  doi          = {10.1007/978-3-031-57256-2\_21},
  bibsource    = {dblp computer science bibliography, https://dblp.org}
}

@inproceedings{svcomp24,
        author = {D.~Beyer},
        title = {State of the Art in Software Verification and Witness Validation: {SV-COMP 2024}},
        booktitle = {Proc.\ TACAS~(3)},
        pages = {299-329},
        year = {2024},
        series = {LNCS~14572},
        publisher = {Springer},
        doi = {10.1007/978-3-031-57256-2_15},
}

@inproceedings{predator,
  author       = {Luk{\'{a}}s Hol{\'{\i}}k and
                  Michal Kotoun and
                  Petr Peringer and
                  Veronika Sokov{\'{a}} and
                  Marek Trt{\'{\i}}k and
                  Tom{\'{a}}s Vojnar},
  editor       = {Roderick Bloem and
                  Eli Arbel},
  title        = {Predator Shape Analysis Tool Suite},
  booktitle    = {Hardware and Software: Verification and Testing - 12th International
                  Haifa Verification Conference, {HVC} 2016, Haifa, Israel, November
                  14-17, 2016, Proceedings},
  series       = {Lecture Notes in Computer Science},
  volume       = {10028},
  pages        = {202--209},
  year         = {2016},
  url          = {https://doi.org/10.1007/978-3-319-49052-6\_13},
  doi          = {10.1007/978-3-319-49052-6\_13},
  bibsource    = {dblp computer science bibliography, https://dblp.org}
}

@inproceedings{predatorhp,
  author       = {Petr Peringer and
                  Veronika Sokov{\'{a}} and
                  Tom{\'{a}}s Vojnar},
  editor       = {Armin Biere and
                  David Parker},
  title        = {PredatorHP Revamped (Not Only) for Interval-Sized Memory Regions and
                  Memory Reallocation (Competition Contribution)},
  booktitle    = {Tools and Algorithms for the Construction and Analysis of Systems
                  - 26th International Conference, {TACAS} 2020, Held as Part of the
                  European Joint Conferences on Theory and Practice of Software, {ETAPS}
                  2020, Dublin, Ireland, April 25-30, 2020, Proceedings, Part {II}},
  series       = {Lecture Notes in Computer Science},
  volume       = {12079},
  pages        = {408--412},
  publisher    = {Springer},
  year         = {2020},
  url          = {https://doi.org/10.1007/978-3-030-45237-7\_30},
  doi          = {10.1007/978-3-030-45237-7\_30},
  bibsource    = {dblp computer science bibliography, https://dblp.org}
}

@inproceedings{seahorn,
  author    = {Arie Gurfinkel and
               Temesghen Kahsai and
               Anvesh Komuravelli and
               Jorge A. Navas},
  editor    = {Daniel Kroening and
               Corina S. Pasareanu},
  title     = {{The SeaHorn Verification Framework}},
  booktitle = {Computer Aided Verification - 27th International Conference, {CAV}
               2015, San Francisco, CA, USA, July 18-24, 2015, Proceedings, Part
               {I}},
  series    = {Lecture Notes in Computer Science},
  volume    = {9206},
  pages     = {343--361},
  publisher = {Springer},
  year      = {2015},
  url       = {https://doi.org/10.1007/978-3-319-21690-4\_20},
  doi       = {10.1007/978-3-319-21690-4\_20},
  bibsource = {dblp computer science bibliography, https://dblp.org}
}

@inproceedings{flanagan01popl,
  author = {Flanagan, Cormac and Saxe, James},
  year = {2001},
  month = {03},
  pages = {193-205},
  title = {Avoiding exponential explosion: Generating compact verification conditions},
  volume = {36},
  booktitle = {Conference Record of the Annual {ACM} Symposium on Principles of Programming Languages},
  doi = {10.1145/373243.360220}
}

@phdthesis{DBLP:phd/us/Owicki75,
  author       = {Susan S. Owicki},
  title        = {Axiomatic Proof Techniques for Parallel Programs},
  school       = {Cornell University, {USA}},
  year         = {1975},
  bibsource    = {dblp computer science bibliography, https://dblp.org}
}

@inproceedings{DBLP:conf/lopstr/EsenR20,
  author       = {Zafer Esen and
                  Philipp R{\"{u}}mmer},
  editor       = {Maribel Fern{\'{a}}ndez},
  title        = {Reasoning in the Theory of Heap: Satisfiability and Interpolation},
  booktitle    = {Logic-Based Program Synthesis and Transformation - 30th International
                  Symposium, {LOPSTR} 2020, Bologna, Italy, September 7-9, 2020, Proceedings},
  series       = {Lecture Notes in Computer Science},
  volume       = {12561},
  pages        = {173--191},
  publisher    = {Springer},
  year         = {2020},
  doi          = {10.1007/978-3-030-68446-4\_9},
  bibsource    = {dblp computer science bibliography, https://dblp.org}
}

@inproceedings{DBLP:conf/lics/Reynolds02,
  author       = {John C. Reynolds},
  title        = {Separation Logic: {A} Logic for Shared Mutable Data Structures},
  booktitle    = {17th {IEEE} Symposium on Logic in Computer Science {(LICS} 2002),
                  22-25 July 2002, Copenhagen, Denmark, Proceedings},
  pages        = {55--74},
  publisher    = {{IEEE} Computer Society},
  year         = {2002},
  url          = {https://doi.org/10.1109/LICS.2002.1029817},
  doi          = {10.1109/LICS.2002.1029817},
  bibsource    = {dblp computer science bibliography, https://dblp.org}
}

@inproceedings{DBLP:conf/fmcad/KomuravelliBGM15,
  author       = {Anvesh Komuravelli and
                  Nikolaj S. Bj{\o}rner and
                  Arie Gurfinkel and
                  Kenneth L. McMillan},
  editor       = {Roope Kaivola and
                  Thomas Wahl},
  title        = {Compositional Verification of Procedural Programs using Horn Clauses
                  over Integers and Arrays},
  booktitle    = {Formal Methods in Computer-Aided Design, {FMCAD} 2015, Austin, Texas,
                  USA, September 27-30, 2015},
  pages        = {89--96},
  publisher    = {{IEEE}},
  year         = {2015},
  url          = {https://doi.org/10.1109/FMCAD.2015.7542257},
  doi          = {10.1109/FMCAD.2015.7542257},
  bibsource    = {dblp computer science bibliography, https://dblp.org}
}

@article{DBLP:journals/fuin/AngelisFPP17a,
  author    = {Emanuele {De Angelis} and
               Fabio Fioravanti and
               Alberto Pettorossi and
               Maurizio Proietti},
  title     = {Program Verification using Constraint Handling Rules and Array Constraint
               Generalizations},
  journal   = {Fundam. Inform.},
  volume    = {150},
  number    = {1},
  pages     = {73--117},
  year      = {2017},
  url       = {https://doi.org/10.3233/FI-2017-1461},
  doi       = {10.3233/FI-2017-1461},
  bibsource = {dblp computer science bibliography, https://dblp.org}
}

@inproceedings{DBLP:conf/tacas/DistefanoOY06,
  author       = {Dino Distefano and
                  Peter W. O'Hearn and
                  Hongseok Yang},
  editor       = {Holger Hermanns and
                  Jens Palsberg},
  title        = {A Local Shape Analysis Based on Separation Logic},
  booktitle    = {Tools and Algorithms for the Construction and Analysis of Systems,
                  12th International Conference, {TACAS} 2006 Held as Part of the Joint
                  European Conferences on Theory and Practice of Software, {ETAPS} 2006,
                  Vienna, Austria, March 25 - April 2, 2006, Proceedings},
  series       = {Lecture Notes in Computer Science},
  volume       = {3920},
  pages        = {287--302},
  publisher    = {Springer},
  year         = {2006},
  url          = {https://doi.org/10.1007/11691372\_19},
  doi          = {10.1007/11691372\_19},
  bibsource    = {dblp computer science bibliography, https://dblp.org}
}

@inproceedings{DBLP:conf/csl/OHearnRY01,
  author       = {Peter W. O'Hearn and
                  John C. Reynolds and
                  Hongseok Yang},
  editor       = {Laurent Fribourg},
  title        = {Local Reasoning about Programs that Alter Data Structures},
  booktitle    = {Computer Science Logic, 15th International Workshop, {CSL} 2001. 10th
                  Annual Conference of the EACSL, Paris, France, September 10-13, 2001,
                  Proceedings},
  series       = {Lecture Notes in Computer Science},
  volume       = {2142},
  pages        = {1--19},
  publisher    = {Springer},
  year         = {2001},
  url          = {https://doi.org/10.1007/3-540-44802-0\_1},
  doi          = {10.1007/3-540-44802-0\_1},
  bibsource    = {dblp computer science bibliography, https://dblp.org}
}

@inproceedings{DBLP:conf/cav/DudkaPV11,
  author       = {Kamil Dudka and
                  Petr Peringer and
                  Tom{\'{a}}s Vojnar},
  editor       = {Ganesh Gopalakrishnan and
                  Shaz Qadeer},
  title        = {Predator: {A} Practical Tool for Checking Manipulation of Dynamic
                  Data Structures Using Separation Logic},
  booktitle    = {Computer Aided Verification - 23rd International Conference, {CAV}
                  2011, Snowbird, UT, USA, July 14-20, 2011. Proceedings},
  series       = {Lecture Notes in Computer Science},
  volume       = {6806},
  pages        = {372--378},
  publisher    = {Springer},
  year         = {2011},
  url          = {https://doi.org/10.1007/978-3-642-22110-1\_29},
  doi          = {10.1007/978-3-642-22110-1\_29},
  bibsource    = {dblp computer science bibliography, https://dblp.org}
}

@inproceedings{DBLP:conf/sas/ChangRN07,
  author       = {Bor{-}Yuh Evan Chang and
                  Xavier Rival and
                  George C. Necula},
  editor       = {Hanne Riis Nielson and
                  Gilberto Fil{\'{e}}},
  title        = {Shape Analysis with Structural Invariant Checkers},
  booktitle    = {Static Analysis, 14th International Symposium, {SAS} 2007, Kongens
                  Lyngby, Denmark, August 22-24, 2007, Proceedings},
  series       = {Lecture Notes in Computer Science},
  volume       = {4634},
  pages        = {384--401},
  publisher    = {Springer},
  year         = {2007},
  url          = {https://doi.org/10.1007/978-3-540-74061-2\_24},
  doi          = {10.1007/978-3-540-74061-2\_24},
  bibsource    = {dblp computer science bibliography, https://dblp.org}
}

@inproceedings{DBLP:conf/cav/BerdineCCDOWY07,
  author       = {Josh Berdine and
                  Cristiano Calcagno and
                  Byron Cook and
                  Dino Distefano and
                  Peter W. O'Hearn and
                  Thomas Wies and
                  Hongseok Yang},
  editor       = {Werner Damm and
                  Holger Hermanns},
  title        = {Shape Analysis for Composite Data Structures},
  booktitle    = {Computer Aided Verification, 19th International Conference, {CAV}
                  2007, Berlin, Germany, July 3-7, 2007, Proceedings},
  series       = {Lecture Notes in Computer Science},
  volume       = {4590},
  pages        = {178--192},
  publisher    = {Springer},
  year         = {2007},
  url          = {https://doi.org/10.1007/978-3-540-73368-3\_22},
  doi          = {10.1007/978-3-540-73368-3\_22},
  bibsource    = {dblp computer science bibliography, https://dblp.org}
}

@inproceedings{refinementtypes,
  author       = {Timothy S. Freeman and
                  Frank Pfenning},
  editor       = {David S. Wise},
  title        = {Refinement Types for {ML}},
  booktitle    = {Proceedings of the {ACM} SIGPLAN'91 Conference on Programming Language
                  Design and Implementation (PLDI), Toronto, Ontario, Canada, June 26-28,
                  1991},
  pages        = {268--277},
  publisher    = {{ACM}},
  year         = {1991},
  url          = {https://doi.org/10.1145/113445.113468},
  doi          = {10.1145/113445.113468},
  bibsource    = {dblp computer science bibliography, https://dblp.org}
}

@inproceedings{liquidtypes,
  author       = {Patrick Maxim Rondon and
                  Ming Kawaguchi and
                  Ranjit Jhala},
  editor       = {Rajiv Gupta and
                  Saman P. Amarasinghe},
  title        = {Liquid types},
  booktitle    = {Proceedings of the {ACM} {SIGPLAN} 2008 Conference on Programming
                  Language Design and Implementation, Tucson, AZ, USA, June 7-13, 2008},
  pages        = {159--169},
  publisher    = {{ACM}},
  year         = {2008},
  url          = {https://doi.org/10.1145/1375581.1375602},
  doi          = {10.1145/1375581.1375602},
  bibsource    = {dblp computer science bibliography, https://dblp.org}
}

@inproceedings{DBLP:conf/csfw/BengtsonBFGM08,
  author       = {Jesper Bengtson and
                  Karthikeyan Bhargavan and
                  C{\'{e}}dric Fournet and
                  Andrew D. Gordon and
                  Sergio Maffeis},
  title        = {Refinement Types for Secure Implementations},
  booktitle    = {Proceedings of the 21st {IEEE} Computer Security Foundations Symposium,
                  {CSF} 2008, Pittsburgh, Pennsylvania, USA, 23-25 June 2008},
  pages        = {17--32},
  publisher    = {{IEEE} Computer Society},
  year         = {2008},
  url          = {https://doi.org/10.1109/CSF.2008.27},
  doi          = {10.1109/CSF.2008.27},
  bibsource    = {dblp computer science bibliography, https://dblp.org}
}

@article{DBLP:journals/toplas/KnowlesF10,
  author       = {Kenneth L. Knowles and
                  Cormac Flanagan},
  title        = {Hybrid type checking},
  journal      = {{ACM} Trans. Program. Lang. Syst.},
  volume       = {32},
  number       = {2},
  pages        = {6:1--6:34},
  year         = {2010},
  url          = {https://doi.org/10.1145/1667048.1667051},
  doi          = {10.1145/1667048.1667051},
  bibsource    = {dblp computer science bibliography, https://dblp.org}
}

@inproceedings{DBLP:conf/sas/BjornerMR13,
  author       = {Nikolaj S. Bj{\o}rner and
                  Kenneth L. McMillan and
                  Andrey Rybalchenko},
  editor       = {Francesco Logozzo and
                  Manuel F{\"{a}}hndrich},
  title        = {On Solving Universally Quantified Horn Clauses},
  booktitle    = {Static Analysis - 20th International Symposium, {SAS} 2013, Seattle,
                  WA, USA, June 20-22, 2013. Proceedings},
  series       = {Lecture Notes in Computer Science},
  volume       = {7935},
  pages        = {105--125},
  publisher    = {Springer},
  year         = {2013},
  url          = {https://doi.org/10.1007/978-3-642-38856-9\_8},
  doi          = {10.1007/978-3-642-38856-9\_8},
  bibsource    = {dblp computer science bibliography, https://dblp.org}
}

@inproceedings{DBLP:conf/sas/MonniauxG16,
  author       = {David Monniaux and
                  Laure Gonnord},
  editor       = {Xavier Rival},
  title        = {Cell Morphing: From Array Programs to Array-Free Horn Clauses},
  booktitle    = {Static Analysis - 23rd International Symposium, {SAS} 2016, Edinburgh,
                  UK, September 8-10, 2016, Proceedings},
  series       = {Lecture Notes in Computer Science},
  volume       = {9837},
  pages        = {361--382},
  publisher    = {Springer},
  year         = {2016},
  url          = {https://doi.org/10.1007/978-3-662-53413-7\_18},
  doi          = {10.1007/978-3-662-53413-7\_18},
  bibsource    = {dblp computer science bibliography, https://dblp.org}
}

@inproceedings{DBLP:conf/popl/JonesM79,
  author       = {Neil D. Jones and
                  Steven S. Muchnick},
  editor       = {Alfred V. Aho and
                  Stephen N. Zilles and
                  Barry K. Rosen},
  title        = {Flow Analysis and Optimization of Lisp-Like Structures},
  booktitle    = {Conference Record of the Sixth Annual {ACM} Symposium on Principles
                  of Programming Languages, San Antonio, Texas, USA, January 1979},
  pages        = {244--256},
  publisher    = {{ACM} Press},
  year         = {1979},
  url          = {https://doi.org/10.1145/567752.567776},
  doi          = {10.1145/567752.567776},
  bibsource    = {dblp computer science bibliography, https://dblp.org}
}

@article{DBLP:journals/ftpl/ChangDMRR20,
  author       = {Bor{-}Yuh Evan Chang and
                  Cezara Dragoi and
                  Roman Manevich and
                  Noam Rinetzky and
                  Xavier Rival},
  title        = {Shape Analysis},
  journal      = {Found. Trends Program. Lang.},
  volume       = {6},
  number       = {1-2},
  pages        = {1--158},
  year         = {2020},
  url          = {https://doi.org/10.1561/2500000037},
  doi          = {10.1561/2500000037},
  bibsource    = {dblp computer science bibliography, https://dblp.org}
}

@inproceedings{DBLP:conf/cav/FaellaP25,
  author       = {Marco Faella and
                  Gennaro Parlato},
  editor       = {Ruzica Piskac and
                  Zvonimir Rakamaric},
  title        = {Verifying Tree-Manipulating Programs via CHCs},
  booktitle    = {Computer Aided Verification - 37th International Conference, {CAV}
                  2025, Zagreb, Croatia, July 23-25, 2025, Proceedings, Part {I}},
  series       = {Lecture Notes in Computer Science},
  volume       = {15931},
  pages        = {3--28},
  publisher    = {Springer},
  year         = {2025},
  url          = {https://doi.org/10.1007/978-3-031-98668-0\_1},
  doi          = {10.1007/978-3-031-98668-0\_1},
  bibsource    = {dblp computer science bibliography, https://dblp.org}
}

@inproceedings{forester,
  author       = {Luk{\'{a}}s Hol{\'{\i}}k and
                  Ondrej Leng{\'{a}}l and
                  Adam Rogalewicz and
                  Jir{\'{\i}} Sim{\'{a}}cek and
                  Tom{\'{a}}s Vojnar},
  editor       = {Natasha Sharygina and
                  Helmut Veith},
  title        = {Fully Automated Shape Analysis Based on Forest Automata},
  booktitle    = {Computer Aided Verification - 25th International Conference, {CAV}
                  2013, Saint Petersburg, Russia, July 13-19, 2013. Proceedings},
  series       = {Lecture Notes in Computer Science},
  volume       = {8044},
  pages        = {740--755},
  publisher    = {Springer},
  year         = {2013},
  url          = {https://doi.org/10.1007/978-3-642-39799-8\_52},
  doi          = {10.1007/978-3-642-39799-8\_52},
  bibsource    = {dblp computer science bibliography, https://dblp.org}
}

@inproceedings{viper,
  author      = {P. M{\"u}ller and M. Schwerhoff and A. J. Summers},
  title       = {Viper: A Verification Infrastructure for Permission-Based Reasoning},
  booktitle   = {Verification, Model Checking, and Abstract Interpretation (VMCAI)},
  editor      = {B. Jobstmann and K. R. M. Leino},
  year        = {2016},
  publisher   = {Springer-Verlag},
  series      = {LNCS},
  pages       = {41-62},
  volume      = {9583},
  url = {https://doi.org/10.1007/978-3-662-49122-5_2},
  urltext = {[Publisher]}
}

@article{DBLP:journals/pacmpl/KrishnaSW18,
  author       = {Siddharth Krishna and
                  Dennis E. Shasha and
                  Thomas Wies},
  title        = {Go with the flow: compositional abstractions for concurrent data structures},
  journal      = {Proc. {ACM} Program. Lang.},
  volume       = {2},
  number       = {{POPL}},
  pages        = {37:1--37:31},
  year         = {2018},
  url          = {https://doi.org/10.1145/3158125},
  doi          = {10.1145/3158125},
  bibsource    = {dblp computer science bibliography, https://dblp.org}
}

@inproceedings{DBLP:conf/esop/KrishnaSW20,
  author       = {Siddharth Krishna and
                  Alexander J. Summers and
                  Thomas Wies},
  editor       = {Peter M{\"{u}}ller},
  title        = {Local Reasoning for Global Graph Properties},
  booktitle    = {Programming Languages and Systems - 29th European Symposium on Programming,
                  {ESOP} 2020, Held as Part of the European Joint Conferences on Theory
                  and Practice of Software, {ETAPS} 2020, Dublin, Ireland, April 25-30,
                  2020, Proceedings},
  series       = {Lecture Notes in Computer Science},
  volume       = {12075},
  pages        = {308--335},
  publisher    = {Springer},
  year         = {2020},
  url          = {https://doi.org/10.1007/978-3-030-44914-8\_12},
  doi          = {10.1007/978-3-030-44914-8\_12},
  bibsource    = {dblp computer science bibliography, https://dblp.org}
}

@inproceedings{DBLP:conf/tacas/MeyerWW23,
  author       = {Roland Meyer and
                  Thomas Wies and
                  Sebastian Wolff},
  editor       = {Sriram Sankaranarayanan and
                  Natasha Sharygina},
  title        = {Make Flows Small Again: Revisiting the Flow Framework},
  booktitle    = {Tools and Algorithms for the Construction and Analysis of Systems
                  - 29th International Conference, {TACAS} 2023, Held as Part of the
                  European Joint Conferences on Theory and Practice of Software, {ETAPS}
                  2022, Paris, France, April 22-27, 2023, Proceedings, Part {I}},
  series       = {Lecture Notes in Computer Science},
  volume       = {13993},
  pages        = {628--646},
  publisher    = {Springer},
  year         = {2023},
  url          = {https://doi.org/10.1007/978-3-031-30823-9\_32},
  doi          = {10.1007/978-3-031-30823-9\_32},
  bibsource    = {dblp computer science bibliography, https://dblp.org}
}

@inproceedings{DBLP:conf/cav/WolffGEHRW25,
  author       = {Sebastian Wolff and
                  Ekanshdeep Gupta and
                  Zafer Esen and
                  Hossein Hojjat and
                  Philipp R{\"{u}}mmer and
                  Thomas Wies},
  editor       = {Ruzica Piskac and
                  Zvonimir Rakamaric},
  title        = {Arithmetizing Shape Analysis},
  booktitle    = {Computer Aided Verification - 37th International Conference, {CAV}
                  2025, Zagreb, Croatia, July 23-25, 2025, Proceedings, Part {I}},
  series       = {Lecture Notes in Computer Science},
  volume       = {15931},
  pages        = {56--80},
  publisher    = {Springer},
  year         = {2025},
  url          = {https://doi.org/10.1007/978-3-031-98668-0\_3},
  doi          = {10.1007/978-3-031-98668-0\_3},
  bibsource    = {dblp computer science bibliography, https://dblp.org}
}

@inproceedings{DBLP:conf/sas/Garcia-Contreras22,
  author       = {Isabel Garcia{-}Contreras and
                  Arie Gurfinkel and
                  Jorge A. Navas},
  editor       = {Gagandeep Singh and
                  Caterina Urban},
  title        = {Efficient Modular SMT-Based Model Checking of Pointer Programs},
  booktitle    = {Static Analysis - 29th International Symposium, {SAS} 2022, Auckland,
                  New Zealand, December 5-7, 2022, Proceedings},
  series       = {Lecture Notes in Computer Science},
  volume       = {13790},
  pages        = {227--246},
  publisher    = {Springer},
  year         = {2022},
  url          = {https://doi.org/10.1007/978-3-031-22308-2\_11},
  doi          = {10.1007/978-3-031-22308-2\_11},
  bibsource    = {dblp computer science bibliography, https://dblp.org}
}

@inproceedings{DBLP:conf/vmcai/SuNGG25,
  author       = {Yusen Su and
                  Jorge A. Navas and
                  Arie Gurfinkel and
                  Isabel Garcia{-}Contreras},
  editor       = {Shankaranarayanan Krishna and
                  Sriram Sankaranarayanan and
                  Ashutosh Trivedi},
  title        = {Automatic Inference of Relational Object Invariants},
  booktitle    = {Verification, Model Checking, and Abstract Interpretation - 26th International
                  Conference, {VMCAI} 2025, Denver, CO, USA, January 20-21, 2025, Proceedings,
                  Part {I}},
  series       = {Lecture Notes in Computer Science},
  volume       = {15529},
  pages        = {214--236},
  publisher    = {Springer},
  year         = {2025},
  url          = {https://doi.org/10.1007/978-3-031-82700-6\_10},
  doi          = {10.1007/978-3-031-82700-6\_10},
  bibsource    = {dblp computer science bibliography, https://dblp.org}
}

@TECHREPORT{BarFT-RR-17,
  author =	 {Clark Barrett and Pascal Fontaine and Cesare Tinelli},
  title =	 {{The SMT-LIB Standard: Version 2.6}},
  institution =	 {Department of Computer Science, The University of Iowa},
  year =	 2017,
  note =	 {Available at {\tt www.SMT-LIB.org}}
}

@inproceedings{DBLP:conf/tacas/MalikMSSVW18,
  author       = {Viktor Mal{\'{\i}}k and
                  Stefan Marticek and
                  Peter Schrammel and
                  Mandayam K. Srivas and
                  Tom{\'{a}}s Vojnar and
                  Johanan Wahlang},
  editor       = {Dirk Beyer and
                  Marieke Huisman},
  title        = {2LS: Memory Safety and Non-termination - (Competition Contribution)},
  booktitle    = {Tools and Algorithms for the Construction and Analysis of Systems
                  - 24th International Conference, {TACAS} 2018, Held as Part of the
                  European Joint Conferences on Theory and Practice of Software, {ETAPS}
                  2018, Thessaloniki, Greece, April 14-20, 2018, Proceedings, Part {II}},
  series       = {Lecture Notes in Computer Science},
  volume       = {10806},
  pages        = {417--421},
  publisher    = {Springer},
  year         = {2018},
  url          = {https://doi.org/10.1007/978-3-319-89963-3\_24},
  doi          = {10.1007/978-3-319-89963-3\_24},
  bibsource    = {dblp computer science bibliography, https://dblp.org}
}

@software{artifact-zenodo,
  author       = {Esen, Zafer and
                  Rümmer, Philipp and
                  Weber, Tjark},
  title        = {Artifact for the paper "Sound and Complete
                   Invariant-Based Heap Encodings"
                  },
  month        = feb,
  year         = 2026,
  publisher    = {Zenodo},
  doi          = {10.5281/zenodo.18500638},
  url          = {https://doi.org/10.5281/zenodo.18500638},
}
